\documentclass[12pt]{article}

\date{}

\usepackage[top=2.54cm, bottom=2.54cm, left=2.54cm, right=2.54cm]{geometry}
\usepackage{floatrow}
\newfloatcommand{capbtabbox}{table}[][\FBwidth]

\reversemarginpar
\usepackage{xcolor}\definecolor{tum_blue}{RGB}{0, 115, 207}\colorlet{col_section }{tum_blue}

\usepackage[latin1]{inputenc}
\usepackage{tikz}
\tikzset{
  treenode/.style = {shape=rectangle, rounded corners,
                     draw, align=center,
                     top color=white, bottom color=blue!20},
  root/.style     = {treenode, font=\normalsize, bottom color=blue!30},
  decision/.style      = {treenode, font=\normalsize, bottom color=red!30},
  env/.style      = {treenode, font=\normalsize},
  dummy/.style    = {circle,draw}
}
\usetikzlibrary{shapes,arrows}
\usepackage{
siunitx,
	booktabs
	,authblk
	,lmodern
	,multirow
	,subfigure
	,graphicx
	,amssymb
	,amsfonts
	,amsmath
	,amsthm
	,amssymb
	,stmaryrd
	,color
	,array
	,enumerate
        ,microtype
        ,geometry
        ,graphics
	,threeparttable
	,longtable
	,rotating
	,lscape
	,tabularx
	,epsfig
        ,epstopdf
	,setspace
	,titling
}

\usepackage{chngcntr}
\counterwithin{equation}{section}
\usepackage[tableposition=top]{caption}
\newcolumntype{Y}{>{\centering\arraybackslash}X}
\usepackage[multiple]{footmisc}

\usepackage{hyperref}
\usepackage{breakurl}
\hypersetup{
	pdfstartview = FitH,
	pdfauthor = {...},
	pdftitle = {...},
	pdfkeywords = {...; ...; ...; ...},
	linktocpage=true,
	colorlinks=true,      % false: boxed links; true: colored links
    linkcolor=blue,       % color of internal links
    citecolor=blue,       % color of links to bibliography
    filecolor=blue,       % color of file links
    urlcolor=blue         % color of external links
}

\usepackage{xurl}
\usepackage{chapterbib}
\usepackage{natbib}

\newtheorem{definition}{Definition}
\newtheorem{theorem}{Theorem}
\newtheorem{remark}{Remark}
\newtheorem{proposition}{Proposition}
\newtheorem{assumption}{Assumption}

\newtheorem{example}{Example}
\newtheorem{lemma}{Lemma}
\newenvironment{continued}[1][continued]{\begin{trivlist}
\item[\hskip \labelsep {\bfseries #1}] \itshape}{\end{trivlist}}

\newtheorem{assumptionalt}{Assumption}[assumption]

\newenvironment{assumptionp}[1]{
  \renewcommand\theassumptionalt{#1}
  \assumptionalt
}{\endassumptionalt}

\newcommand{\indep}{\rotatebox[origin=c]{90}{$\models$}}

\newcommand{\prob}[1]{\mathbb{P}\!\left\{ #1 \right\}}

\newcommand{\expect}[1]{\mathbb{E}\!\left[ #1 \right]}

\newcommand{\iiddist}{\mathrel{\overset{\text{i.i.d.}}{\kern0pt\resizebox{!}{1.1ex}{$\sim$}}}} 
\newcommand{\Pbf}{\mathbf{P}}
\newcommand{\pbf}{\mathbf{p}}
\newcommand{\sym}[1]{^{#1}}

\title{Combining stated and revealed preferences, \\matched and unmatched\thanks{First version: 14 July 2025. This version: \today. This paper supersedes an earlier version by the first author entitled ``Using stated preferences to understand actual choice''. The authors thank Eleanor Dickens for invaluable research assistance. The authors are grateful to Johannes Abeler, Abi Adams, Manuel Arellano,  St\'ephane Bonhomme, Kirill Evdokimov, Yanqin Fan, Christophe Gaillac, Pamela Giustinelli, Max Kasy, Francesca Molinari, Brendan Pass, Fran\c{c}ois Poinas, Michael Pollmann, Martin Weidner, Basit Zafar, seminar audiences at Duke, UNC, Wisconsin, and participants in the 10th Cornell-PennState Econometrics and IO  Conference, the 2025 Greater New York Metropolitan Area Econometrics Colloquium, the BSE Summer Forum workshop on Microeconometrics and Policy evaluation, and the Econometric Society Meetings in Cairo for stimulating discussions and useful comments. The usual disclaimer applies. Correspondence address: Manor Road Building, Manor Road, Oxford, OX1 3UQ. email: romuald.meango@economics.ox.ac.uk.}}

\author{R. M\'eango, M. Henry, I. Mourifi\'e and Ju Hyun Oh}
\affil{University of Oxford, Penn State, Washington University in St. Louis, Penn State}

\begin{document}

\maketitle

\begin{abstract}
    Can stated preferences inform counterfactual analyses of actual choice? This research proposes a novel approach for researchers who have access to both stated choices in hypothetical scenarios and actual choices, matched or unmatched. The key idea is to use stated choices to control for individual unobserved heterogeneity. If this unobserved heterogeneity is the source of endogeneity, the researcher can correct for its influence in a choice function estimation using actual choices and recover causal effects. Bounds on causal effects are derived in the case where stated choices and actual choices are observed in unmatched data sets. These data combination bounds are of independent interest. We derive valid bootstrap inference for the bounds and show its good performance in a simulation experiment. We illustrate the method in an analysis of retirement decisions based on the Health and Retirement Study.
\end{abstract}



% INTRODUCTION
\section*{Introduction}

Social scientists use two different ways to learn about human behaviour: One possibility is to observe what people do or choose in real life, which economists call the \textit{revealed preference approach}. This approach has the advantage of reflecting actual choices. Its main limitation is that choice relevant states and observed choices are driven by common unobserved underlying factors. Researchers have relied on randomised controlled trials or natural experiments to perform causal analysis. However, randomisation is not always feasible, for technical, political or ethical reasons, and when it is feasible, it is often subject to imperfect compliance.
The other possibility to learn about human behaviour is to ask people what they would do or choose in a hypothetical situation, in what is often called the \textit{stated preference approach}. The main advantages of this approach are that the analyst can design experiments with rich variations in choice attributes, observe individual stated choice in counterfactual scenarios and explore preferences over policies never implemented before. Moreover, the analyst can exploit hypothetical choice scenarios to tackle endogeneity arising from omitted unobserved characteristics and identify causal effects of choice attributes. %Indeed, even if the relevant choice characteristics are not all observed, exogenous variations on choice characteristics manipulated in a survey experiment are sufficient for causal inference. 
The obvious disadvantage of the stated preference approach is summarized by \cite{train2009}: `What people say they will do is often not the same as what they actually do'. This discrepancy is known as the \textit{hypothetical bias}.

Unlike other social sciences, economics has historically seen the hypothetical bias as a disqualifying argument against the stated preference approach, and until recently, economists have, by and large, relied solely on revealed preference data, except in the measurement of the non-use value of natural or cultural resources. This attitude is changing, as observed and advocated by Orazio Attanasio's 2020 Econometric Society presidential address, published in \cite{almaas2024}.\footnote{The turn of the tide coincides with a surge of interest in new measures capturing individuals' subjective expectations as documented by \cite{manski2004} and \cite{almaas2024}. See also \cite{mcfadden2017} for a historical perspective on stated preference data.} 
Stated preference analyses are increasingly used to describe individual preferences over choice attributes that are endogenous, hard to measure in observational data or hard to vary in randomized controlled trials. %, e.g. `how would the demand for college education change if the subjective expected income would rise by one percent?'
An early example is \cite{juster1964}. Recent applications span many areas of economics, including education choices \citep{arcidiacono2020,delavande2019}, mobility decisions \citep{gong2022, kocsar2022, meango2022, batista2025}, health and long-term care investments \citep{kesternich2013, ameriks2020a,boyer2020,ameriks2011joy}, parental investments \citep{attanasio2019,almaas2024}, marriage preference \citep{adams2019, low2024}, major and occupational choices \citep{maestas2023, wiswall2015, wiswall2018,meango2024}, and retirement decisions \citep{ameriks2020b, giustinelli2024}. For recent reviews, see \cite{kocsar2023} and \cite{giustinelli2023}. Despite the enthusiasm and encouraging evidence that carefully designed stated preference elicitation has predictive power for actual choices \citep{hurd2009, hainmueller2015, debresser2019, arcidiacono2020} and yields similar preference estimates as actual choice data \citep{mas2017, wiswall2018}, concerns remain regarding systematic biases of stated preference data. See, for example, \cite{murphy2005} and \cite{hausman2012} and more recently \cite{haghani2021a,haghani2021b}. 

A natural way to mitigate the respective shortcomings of the two
approaches, namely hypothetical bias and endogenous choice attributes, is to combine them. The idea is at least as old as \citet{manski1990use},
who advocates eliciting choice probabilities rather than yes/no point
predictions and derives best-case bounds on behavior from elicited intentions. The bound in \citet{manski1990use} relates a marginal computed from intentions
to a marginal computed from behavior. As such, it is a direct precursor of the data-combination logic we
develop below. \citet{manski1999} subsequently showed how elicited choice
probabilities identify random utility models with random coefficients. The latter program is
implemented by \citet{blass2010} for electricity reliability and
\citet{delavande2015} for voting. These analyses estimate preference
parameters directly from elicited probabilities and use no revealed preference
data. \citet{pantano2013}, by contrast, use elicited probabilities to
recover a latent type and only then proceed to structural estimation. The literature that
explicitly combines stated and revealed data, however, has so far relied on two
restrictions. First, it imposes strong structural links between statement and
action. \citet{morikawa2002}, \citet{pantano2013}, and
\citet{giustinelli2024} restrict the gap between the utility parameters
governing stated and revealed choices. \citet{vanderklaauw2012} and
\citet{wiswall2021} attribute any discrepancy solely to biased information.
\citet{bernheim2022} assume that treatment affects actual choice only through
stated preference. \citet{briggs2020} assume additively separable scalar
heterogeneity, a particular timing of the resolution of uncertainty, and rational
expectations.\footnote{\citet{athey2025combining} apply related ideas to the
combination of observational and experimental data.} Second, this literature mostly requires matched data, even though stated and revealed preferences are typically collected, often by design\footnote{Elicitation and decision are often separated, to avoid actual choices being affected by the survey.}, in separate data sets, where the individuals cannot be matched.

This paper combines revealed and stated preference data, matched or unmatched, to
analyze binary choice with an endogenous attribute. We do so without strong
assumptions on how the two are related. Our point of departure is a shift in the way
we use stated preferences. They are valuable not because they
necessarily coincide with actual choices, but because they carry information about
the individual heterogeneity that drives those choices. Common factors that move
both decisions and outcomes, such as risk aversion, taste for leisure, health
consciousness, or psychological traits, are poorly encoded in observables but can be
drawn out by a well-designed survey. Consider a job-seeker asked, across scenarios
varying wage, security, and non-wage compensation, how likely they are to take a
job. A respondent who reports a $90\%$ chance in every scenario is plainly
different from one who reports $90\%$ when the wage is high and $10\%$ otherwise:
they differ in their sensitivity to job attributes, whatever the level of their
stated probabilities and however much both may overstate their true willingness to
work. Their answers thus sort them into heterogeneity types, and richer elicitation
recovers richer heterogeneity.

We formalize this as follows. An agent makes a binary decision $D \in \{0,1\}$
given a vector of attributes $X$ with finite support, with potential decision $D(x)$ and parameter of interest the average structural function $\mu(x) := \mathbb{E}[D(x)]$. The
attribute $X$ is endogenous because $D(x)$ depends on a vector $\eta$ of
unobserved characteristics that is also the source of endogeneity, in the sense
that~$D(x) \perp\!\!\!\perp X \mid \eta$. Prior to the decision, the agent reports
stated choice probabilities $\mathbf P=(P_1,\dots,P_T)$ for hypothetical attribute values,
with $P_t = \tilde{m}(x_t,\eta)$, for some function~$\tilde m$ of choice attribute~$x_t$ and unobserved heterogeneity~$\eta$. Crucially, the same $\eta$ governs both the
report and the decision; this is the only link the model imposes between stated
and revealed preferences, and it accommodates hypothetical bias. Under a
mild invertibility condition, the~$\sigma$-algebra~$\sigma(\eta)$ is recovered from a subvector of $\mathbf P$, which
delivers the conditional independence $D(x) \perp\!\!\!\perp X \mid \mathbf P$. When the two data sources are matched and the report distributions satisfy the overlap condition $F_{\mathbf{P}}\ll F_{\mathbf{P}\mid X=x}$, this point-identifies the average structural
function as
\[
\mu(x) \;=\; \int \mathbb{E}[D \mid X=x, \mathbf P=\mathbf p]\, dF_\mathbf P(\mathbf p).
\]
Conditioning on stated preferences plays the role that conditioning on a control function plays in the tradition of \citet{heckman1985alternative}, but the control variable is here supplied by the survey rather than by a first-stage model.\footnote{In the latent unconfoundedness assumption of \cite{athey2025combining}, short term outcomes play the same role as elicited probabilities do here. As in the latent unconfoundedness assumption of \cite{athey2025combining}, we only need weak~$[D(x) \perp\!\!\!\perp X \mid \mathbf P$ for each~$x\in\mathcal X]$, not strong~$[(D(x))_{x\in\mathcal X} \perp\!\!\!\perp X \mid \mathbf P]$ unconfoundedness.}

In most applications, stated and revealed preferences are not matched. The
revealed preference sample identifies the conditional law of $D$ given $X$, the
stated preference sample identifies the conditional law of $\mathbf P$ given $X$, but their
joint distribution is unrestricted beyond these marginals, so $\mu(x)$ is only
partially identified. Our second contribution is a sharp characterization of the
resulting identified set. The relevant object is an integral of
$\mathbb{E}[D\mid X=x,\mathbf P=\mathbf p]$ against the marginal of $\mathbf P$. It differs from the pointwise conditional-expectation bounds of \citet{horowitz1995identification},
\citet{cross2002regressions}, and \citet{molinari2006generalization}, because
the integration involves restrictions across values of~$\mathbf p$ that
pointwise bounds ignore, and because we allow~$\mathbf P$ to be continuous or mixed. We obtain the bounds from the Kantorovich dual of the
optimal transport problem with fixed marginals, which yields a one-dimensional
convex program that is convenient for both computation and inference and that
incorporates dependence restrictions naturally. This places our results alongside
a growing use of optimal transport in data combination, from
\citet{horowitz1995identification} and \citet{heckman1997} through
\citet{fan2014identifying,fan2016estimation}, \citet{lee2019identification},
\citet{pacini2019two}, and \citet{gaillac2024linear,gaillac2025partially}; the
closest antecedent is \citet{fan2025partial}, who apply the method of conditioning
of \citet{ruschendorf1991bounds} to moment-equality models under data combination. See \cite{galichon2026econometrician} for a survey.\footnote{In related data combination environments, \citet{ichimura2005identification} and
\citet{bontemps2025functional}, derive conditions for point identification instead of characterizing identified sets as we do.}
These data-combination bounds are of independent interest beyond the
stated-and-revealed-preference setting that motivates them.
We also show how the bounds tighten under additional covariates common to both samples, exclusion restrictions, and restrictions on the dependence between
elicited probabilities and actual choices.

For inference, we establish that the bounds are Hadamard directionally
differentiable functionals of two densities and a conditional mean, and we prove that
the bootstrap procedure of \citet{fang2019inference} delivers valid confidence
regions for them. A simulation study, calibrated to our retirement application,
documents that the bounds are informative and that the inference procedure performs
well across sample sizes and tuning parameters. We illustrate the methodology by revisiting \citet{giustinelli2024} using the Health and Retirement Study. The 2016 wave elicited, for a subsample under age 65, the chance of working for pay in two years under hypothetical future health statuses; the 2018 wave records realized health and labor-market status, allowing
us to combine stated and revealed preferences. We estimate the effect of poor health on the probability of being out of work. Stated preferences imply a large effect of about 24 percentage points. Estimates based on actual choices are more modest and exhibit substantial age heterogeneity, with larger effects among respondents aged 50 to 59 than among those approaching retirement age. Our combined SP4RP estimates are close to the corresponding revealed-preference estimates. Under our maintained assumption that stated probabilities capture the selection-relevant heterogeneity, this comparison suggests that hypothetical bias is quantitatively more important in this application than the correction for selection on unobservables.

The remainder of the paper is organized as follows. Section~\ref{sec:identification}
lays out the econometric framework and presents the identification result under
matched data. Section~\ref{sec:bounds} derives the sharp data-combination bounds and
the tightening results. Section~\ref{sec:partial} discusses the case with imperfect type recovery. Section~\ref{sec:inference} develops the inference procedure. Section~\ref{sec:simulations} reports the simulation evidence, and Section~\ref{sec:application} presents the empirical
application. The last section concludes.

\subsubsection*{Notation and preliminaries}

All random elements considered in the paper are defined on a common complete
probability space~$(\Omega,\mathcal F,\mathbb P)$. We write~$\mathbf 1\{\cdot\}$
for the indicator function, $\Phi$ for the standard normal distribution
function, $\perp\!\!\!\perp$ for (conditional) independence, $\sigma(V)$ for
the~$\sigma$-algebra generated by a random element~$V$, $\operatorname{supp}F$
for the support of a distribution~$F$, and~$F\ll F^\prime$ for absolute
continuity of~$F$ with respect to~$F^\prime$. For a Borel set~$A$, $\bar A$
denotes its closure and~$\vert A\vert$ its Lebesgue measure. Distribution and
density functions carry their conditioning in the subscript, as
in~$F_{\mathbf P\vert X=x}$ and~$f_{\mathbf P\vert X=x}$. Report densities are
defined relative to a common dominating measure~$\varrho$. Counting and Lebesgue
measures cover discrete and absolutely continuous report distributions,
respectively; singular continuous report distributions require another
dominating measure. Integration over reports in the population bound formulas
is understood relative to~$\varrho$. We denote~$Q_V$ the quantile function of a random
variable~$V$, that is~$Q_V(t):=\inf\{v:F_V(v)\ge t\}$ for~$t\in(0,1)$, and~$h$
the density ratio~$h:=f_{\mathbf P}/f_{\mathbf P\vert X=x}$, written~$h_x$ when
its dependence on~$x$ must be made explicit.

Given distributions~$F$ and~$F^\prime$ on~$\mathcal A$ and~$\mathcal A^\prime$,
$\mathcal M(F,F^\prime)$ denotes the set of probability measures
on~$\mathcal A\times\mathcal A^\prime$ with marginals~$F$ and~$F^\prime$, whose
elements are called couplings of~$F$ and~$F^\prime$. Expectations taken under a
specified joint law~$\pi$ are written~$\mathbb E_\pi$. The~$2$-Wasserstein
distance between distributions~$F$ and~$F^\prime$ on a subset of~$\mathbb R^T$
is
\begin{eqnarray*}
    W_2^2(F,F^\prime) & := & \inf_{\pi\in\mathcal M(F,F^\prime)}
    \mathbb E_\pi\left[\|V-V^\prime\|^2\right],
\end{eqnarray*}
where~$\|\cdot\|$ is the Euclidean norm and~$(V,V^\prime)\sim\pi$.

%Set-valued objects appear in the treatment of interval censored reports. For a known finite grid~$\mathcal G:=\{g_1,\ldots,g_K\}\subseteq[0,1]$, the censoring map is~$G(p):=\max\{g\in\mathcal G:g\le p\}$, and~$G^{-1}(g):=\{p:G(p)=g\}$ denotes the associated bin. The notation~$V\rightrightarrows A(V)$ indicates a random closed set, i.e., a map from~$\Omega$ to the closed subsets of~$\mathbb R$ that is measurable in the sense of \cite{molchanov2005theory}; a selection of~$V\rightrightarrows A(V)$ is a random variable~$\xi$ with~$\xi\in A(V)$ almost surely, and, when~$V$ has law~$\pi$, we denote~$\mathcal G_\pi$ the collection of distributions of its selections.

Finally, $\ell^\infty(A)$ denotes the space of bounded Borel measurable real functions on~$A$
equipped with the supremum norm~$\|\cdot\|_\infty$, and we set
\begin{eqnarray*}
    \mathbb B \; := \; \ell^\infty\big([0,1]^d\big)\times
    \ell^\infty\big([0,1]^d\big)\times\mathbb R,
    & \qquad &
    \|\theta\|_{\mathbb B} \; := \;
    \|\theta_1\|_\infty+\|\theta_2\|_\infty+\vert\theta_3\vert.
\end{eqnarray*}
A map~$\phi:\mathbb B\rightarrow\mathbb R$ is Hadamard directionally
differentiable at~$\theta$ if there is a continuous
map~$\phi^\prime_\theta:\mathbb B\rightarrow\mathbb R$, positively homogeneous
of degree one, such that
$\xi_n^{-1}\left(\phi(\theta+\xi_n\upsilon_n)-\phi(\theta)\right)
\rightarrow\phi^\prime_\theta(\upsilon)$ for all~$\xi_n\downarrow0$ and
all~$\upsilon_n\rightarrow\upsilon$ in~$\mathbb B$; $\phi^\prime_\theta$ need not
be linear. The Bernstein basis of order~$r$ on~$[0,1]$ is
$b_{k,r}(p):=(r+1)\binom{r}{k}p^k(1-p)^{r-k}$, $k=0,\ldots,r$, and a Bernstein
density of order~$r$ is~$f(\cdot\,;\omega):=\sum_{k=0}^r\omega_kb_{k,r}$, with
weights~$\omega=(\omega_0,\ldots,\omega_r)$ in the unit simplex.

%

% FRAMEWORK
\section{Econometric framework}
\label{sec:identification}

\subsection{Model specification for revealed preferences}

Consider an economic agent with a vector~$X$ of economic attributes with finite support~$\mathcal X\subseteq \mathbb R^{d_x}$, who is facing binary decision~$D\in\{0,1\}$. We denote~$D(x)$ the potential decision, i.e., the decision the agent would make if their attributes were externally set to~$X=x$.
Both the vector of attributes~$X$ and the actual decision~$D:=D(X)$ are observed. Potential decision~$D(x)$ is observed only when~$X=x$ and not otherwise.
The attribute~$X$ is endogenous in the sense that~$D(x)$ is not independent of~$X$, and hence~$\mathbb E[D\vert X=x]$, which is identified from the data, is generally not equal to~$\mathbb E[D(x)]$, which is the parameter of interest. Finally, let~$\eta$ be a vector of unobserved characteristics. The support~$\mathcal H\subseteq\mathbb R^d$ of~$\eta$ does not vary with~$X$. All random elements are defined on a common probability space.

We posit that~$\eta$ contains the source of endogeneity of~$X$ in the sense of the following assumption.

\begin{assumption}%[Unobserved Heterogeneity]
\label{ass:endogeneity}
    $D(x) \perp\!\!\!\!\perp X \; \vert \; \eta$ for each~$x\in\mathcal X$.
\end{assumption}
Assumption~\ref{ass:endogeneity} is related to the control function approach pioneered by \cite{heckman1985alternative}.
Under this assumption, the structural choice function~$\mathbb E[D(x)\vert\eta]$ satisfies
\begin{eqnarray*}
    m(x,\eta) \; := \; \mathbb E[D(x)\vert \eta] \; = \; \mathbb E[D(x)\vert X=x,\eta] \; = \; \mathbb E[D\vert X=x,\eta],
\end{eqnarray*}
so that recovering~$\eta$ would allow identification of~$m(x,\eta)$ and hence of the following counterfactual choice parameters.

\begin{definition}\label{def:ASF}
    The average structural function is defined for all~$x\in\mathcal X$ by
    \begin{eqnarray*}
        \mu(x) \; := \; \mathbb E[D(x)] \; = \; \int m(x,\eta) \; dF_\eta.
    \end{eqnarray*}
\end{definition}

\begin{example}[Retirement decision]
\label{ex:pilot}
    Our pilot example is a model of retirement decisions based on health status. The retirement decision is governed by the following model:
    \begin{eqnarray*}
        D(x) & = & g(x,\eta,\nu),
    \end{eqnarray*}
    where the choice attribute~$x$ is health status, which we assume here to be in~$\mathcal X=\{0,1\}$ for simplicity, $g$ is an unknown function, $\eta$, assumed to be scalar for simplicity, is an unobserved trait that influences preferences, such as risk aversion or health consciousness, and~$\nu$ is a shock observed by the agent at the time of decision, but not by the analyst. Assume~$\nu\perp\!\!\!\!\perp X \; \vert \; \eta$ so that assumption~\ref{ass:endogeneity} holds.
\end{example}

\subsection{Model specification for stated preferences}

Stated preferences are data collected from a survey of the agents prior to the time of decision. Agents are asked to give an assessment of their probability of making decision~$D=1$ for some hypothetical values of~$x$. Let~$(x_1,\ldots,x_T)$ be a vector of hypothetical values~$x_t\in\mathcal X$. The agent's response to the probability elicitation question is denoted~$P_t$, $t=1,\ldots,T$.

We model~$P_t$, for each~$t=1,\ldots,T$ as
\begin{eqnarray*}
    P_t \; = \; \tilde m(x_t,\eta)% \; := \; \mathbb E[ \; \tilde D(x_t) \; \vert \; \eta \; ]
    ,
\end{eqnarray*}
where the function~$\tilde m(x,\eta)$ is called the stated choice function. 

In the special case where agents have rational expectations and no reporting bias, $P_t=\mathbb E[D(x_t)\vert \eta]$, so that~$\tilde m(x_t,\eta)=m(x_t,\eta)$ for~$t=1,\ldots,T$. In this case, the identification issue boils down to recovering~$m(x,\eta)$ for~$x \not \in (x_1,\ldots,x_T)$. In general, the stated preference function~$\tilde m(x_t,\eta)$ may differ from the structural choice function~$m(x_t,\eta)$ for revealed preferences, due to the agent's perception and/or reporting biases. 
Crucially, the agent's report~$P_t$ is driven by the same vector~$\eta$ of unobserved characteristics. This provides the only link in the model between stated and revealed preferences.

\begin{continued}[Example~\ref{ex:pilot} continued:]
    In the retirement example, agents are asked some time before retirement age what they assess their probability of retiring to be, given hypothetical health statuses. The stated preference report~$P_t$ is modeled as the expectation of an anticipated decision variable
    \begin{eqnarray*}
        \tilde D(x_t) & = & \tilde g(x_t,\eta,\nu),
    \end{eqnarray*}
    where, as in the model for revealed preferences, $\eta$ is observed by the agent but not the analyst. However, $\nu$ 
is interpreted as resolvable uncertainty. It is not known by the agent at the time of stated preference elicitation, but is revealed at the time of decision. The agent is assumed to form their responses~$P_t$ to stated preference elicitation by integrating~$\nu$, so that
    \begin{eqnarray*}
        P_t \; = \; \tilde m(x_t,\eta) \; = \;
        \int \tilde g(x_t,\eta,\nu) \; d\tilde F_{\nu\vert x_t,\eta},
    \end{eqnarray*}
    where the function~$\tilde g$ may be different from its counterpart~$g$ in the revealed preference model, and the perceived distribution~$\tilde F_{\nu\vert x_t,\eta}$ of resolvable uncertainty~$\nu$ may be different from the true distribution in the population. Note that the model above is observationally equivalent to a model stipulating~$\tilde D(x_t)=\tilde g(x_t,\eta,\tilde\nu)$, where~$\tilde\nu$, possibly different from~$\nu$, is open to multiple interpretations.
\end{continued}

\subsection{Identification of revealed preferences using stated preferences}

Identification of the parameters of interest
relies on the ability to control for unobserved heterogeneity~$\eta\in\mathcal H\subseteq\mathbb R^d$ by matching on stated preference reports~$\mathbf P=(P_1,\ldots,P_T)$. For this, we need to be able to recover unobserved heterogeneity from the stated preference reports of an individual.

\begin{assumption}\label{ass:invertibility}
    There is a subset~$(x_1,\ldots,x_d)$ of the vector of stated preference scenarios such that~$\eta\mapsto\tilde m(\eta):=(\tilde m(x_1,\eta),\ldots,\tilde m(x_d,\eta))$ is continuous and one-to-one.
\end{assumption}

Assumption~\ref{ass:invertibility} guarantees that distinct values of~$\eta$ can be distinguished by the system~$(P_1, \ldots, P_d)= (\tilde m(x_1,\eta),\ldots,\tilde m(x_d,\eta)):= \tilde m(\mathbf x,\eta)$. 

\begin{remark}
    A special case of assumption~\ref{ass:endogeneity} involves a scalar type~$\eta_x$ for each value~$x$ that the endogenous variable~$X$ can take, with~$X$ discrete. Assumption~\ref{ass:invertibility} would then be satisfied if~$P_x=\tilde m(x,\eta_x)$ with invertible~$\tilde m(x,\cdot)$.
\end{remark}

\begin{remark}
    An example often used in empirical applications is a multiplicatively separable log-odds model, of which the model of \cite{blass2010} is a special case. It corresponds to:
\begin{eqnarray}\label{eq:odds}
    \mathbf P & = & \Psi \left( r(\mathbf x)+v(\mathbf x)\eta\right),
\end{eqnarray}
where~$\Psi%:\mathcal H\rightarrow[0,1]^d
$ is invertible (logistic in the case of log-odds) and~$v(\mathbf x)$ is a full rank~$d\times d$ matrix. In this case, assumption~\ref{ass:invertibility} holds with~$\eta = v(\mathbf x)^{-1} \left({\Psi^{-1}(\mathbf P) - r(\mathbf x)}\right)$.
More generally, sufficient conditions for assumption~\ref{ass:invertibility} are given in Theorem~2 of \cite{mas1979}.
\end{remark} 

In model~(\ref{eq:odds}), $\mathbf P$ has the same dimension as~$\eta$. This need not be the case for assumption~\ref{ass:invertibility} to hold, as long as the dimension of~$\mathbf P$ is at least as large as the dimension of~$\eta$. All that is needed, is that a subvector of~$\mathbf P$ identifies~$\eta$. This subvector need not be known, as long as it exists. Similarly, the stated preference function~$\tilde m$ and the true dimension of~$\eta$ need not be known. Assumption~\ref{ass:invertibility}, combined with assumption~\ref{ass:endogeneity}, ensures that~$D(x)\perp\!\!\!\!\perp X\vert \mathbf P$, which drives proposition~\ref{prop:identification ASF} below.

Assumption~\ref{ass:invertibility} allows for a wide range of hypothetical biases, provided that responses to hypothetical scenarios are sufficiently different to distinguish individuals who respond differently to choice variables in actual choices. If two individuals react differently to health shocks, possibly because they have a different taste for work, then Assumption~\ref{ass:invertibility} holds if these individuals also respond differently to the hypothetical questions, one giving answers that are much more sensitive to the scenario than the other, even if both wildly overstate their willingness to work. Assumption~\ref{ass:invertibility} requires the type~$\eta$ of individuals to be fixed between scenarios and between stated and revealed preferences, and two individuals with different types to be distinguishable from their survey responses. Assumption~\ref{ass:invertibility} would be violated for instance if an individual with low taste for work were induced by social desirability bias to give the same responses to the survey as an individual with high taste for work would.

\begin{continued}[Example~\ref{ex:pilot} continued:]
    In the retirement example, we assume that the unobserved heterogeneity factor~$\eta$ is scalar. Suppose the agent perceives the relative utility of option~$D=1$ (retiring at~$65$) with health status~$x_0$ to be~$\tilde S(x_0,\eta)-\nu$. Then, the perceived choice function is~$\tilde g(x,\eta,\nu)=1\{\tilde S(x,\eta)\geq\nu\}$. If~$\nu\perp\!\!\!\!\perp \eta$, the perceived shock distribution~$\tilde F_\nu$ is absolutely continuous and increasing, $\tilde S(x_0,\eta)$ is continuous and increasing in~$\eta$, then~$\tilde m(x_0,\eta)=\tilde F_{\nu}(\tilde S(x_0,\eta))$ is continuous and increasing in~$\eta$ and assumption~\ref{ass:invertibility} holds. The average structural function~$\mu(x)$ is identified as follows:
    \begin{eqnarray*}
        \mu(x) \; = \; \int \mathbb E[D(x)\vert \eta] \; dF_\eta \; = \; \int \mathbb E[D\vert X=x,\eta] \; dF_\eta \; = \; \int \mathbb E[D\vert X=x,P_0] \; dF_{P_0},
    \end{eqnarray*} where the first equality is by definition, the second equality holds by assumption~\ref{ass:endogeneity} and the third equality holds by assumption~\ref{ass:invertibility}. Note also that when~$\eta\mapsto\tilde m(x_0,\eta)$ is increasing, it is identified up to normalization and~$\eta$ can be recovered as a monotone transformation of~$P_0$.
\end{continued}

As illustrated in the example, under assumption~\ref{ass:invertibility}, stated preferences allow us to control for unobserved heterogeneity in the regression of decision~$D$ on the choice attributes~$x$, and to identify the parameters of interest.
To formalize this insight, we need the additional overlap assumption on the elicited probabilities. We denote by~$f_{\mathbf P}$ and~$f_{\mathbf P\vert X=x}$ the densities of~$F_{\mathbf P}$ and~$F_{\mathbf P\vert X=x}$ with respect to a common dominating measure~$\varrho$ on~$[0,1]^T$. Counting and Lebesgue measures cover discrete and absolutely continuous report distributions, respectively; singular continuous report distributions require another dominating measure. Integration over reports in the population bound formulas is understood relative to~$\varrho$.
\begin{assumption}[Overlap]\label{ass:full support}
    For all~$x\in\mathcal X$: (1) $F_{\mathbf P}\ll F_{\mathbf P\vert X=x}$; (2) $f_{\mathbf P},f_{\mathbf P\vert X=x}\in\ell^\infty([0,1]^T)$.
\end{assumption}

    % Assumption~\ref{ass:full support} is satisfied when the densities are finite order Bernstein polynomials with positive coefficients: Since the Bernstein basis functions sum to~$1$, such a polynomial is sandwiched between its smallest and largest coefficients, hence bounded and bounded away from zero, so that part~(2) holds and part~(1) holds with a bounded density ratio.

Assumption~\ref{ass:full support} essentially requires the support of~$\mathbf P$ to be independent of~$X$.

% No \emph{strict} overlap is imposed: $h$ is neither required to be continuous, nor bounded, nor bounded away from zero, and may diverge at the boundary of the support, as it does in the design of section~\ref{sec:simulations}. All that is used for identification is that~$h$ is a density ratio, hence integrable with $\int h\,dF_{\mathbf P\vert X=x}=1$. Part~(2) plays no role in the identification results and is used only for the asymptotic theory of section~\ref{sec:inference}.

\begin{proposition}[Identification of the ASF]\label{prop:identification ASF}
    Under assumptions~\ref{ass:endogeneity}-\ref{ass:full support}, the average structural function is identified as~$\mu(x) = \int \mathbb E[ D\vert X=x,\mathbf P=\mathbf p ] \; dF_{\mathbf P}$, where~$\mathbf P=(P_1,\ldots,P_T)$.
\end{proposition}

Proposition~\ref{prop:identification ASF} is the paper's basic identification result, and it is
constructive: $\mu(x)$ is recovered by averaging the regression $\mathbb E[D\mid X=x,\mathbf P=p]$ against the marginal of $\mathbf P$, both of which are identified from the joint distribution of~$(D,X,\mathbf P)$. The
elicited probabilities enter as a traditional \citet{heckman1985alternative} control function, but the latter is supplied by the survey rather than by a
first-stage model. As a result, the identification requires neither knowledge of the stated preference function~$\tilde m$ nor of the true dimension of $\eta$. Assumption~\ref{ass:invertibility} only requires that
some subvector of~$\mathbf P$ recover~$\eta$, and that subvector need not be known. 
Hypothetical bias can be accommodated, in the sense that the hypothetical choice function~$\tilde m$ may differ from the structural choice function~$m$. The
common dependence on~$\eta$ restores the conditional independence~$D(x)\perp\!\!\!\perp X\mid\mathbf P$ which drives the result. 

The identification result of proposition~\ref{prop:identification ASF} rests on the assumption that revealed preferences and stated preferences are collected in the same data set, and hence that the joint distribution of~$(D,X,\mathbf P)$ is identified. The next section derives the main technical results of the paper, which are sharp bounds on the parameter of interest in the general case, where revealed and stated preferences are collected in separate data sets.

%

% DATA COMBINATION
\section{Bounds under data combination}
\label{sec:bounds}

In most empirical environments, stated preferences and revealed preferences are collected in different data sets, and it is difficult or impossible to match individuals from the two data sets. In many cases, it is preferable to avoid eliciting survey responses and revealed preferences from the same individual, as the survey may affect decisions. 

We therefore assume that revealed and stated preferences are collected in separate data sets. In both data sets, individuals are sampled from the same underlying population. The conditional distribution~$F_{D\vert X}$ of~$D$ given~$X$ is identified from the revealed preference data set, and the conditional distribution~$F_{\mathbf P\vert X}$ of~$\mathbf P$ given~$X$ is identified from the stated preference data set. However, the conditional distribution~$F_{D,\mathbf P\vert X}$ of~$(D,\mathbf P)$ given~$X$ is not point identified. It can take any value in the set~$\mathcal M(F_{D\vert X=x},F_{\mathbf P\vert X=x})$, where~$\mathcal M(F,F^\prime)$ is the set of joint distributions with marginals~$F$ and~$F^\prime$.

For each conditional distribution~$F_{D,\mathbf P\vert X}$ for~$(D,\mathbf P)$ given~$X$, proposition~\ref{prop:identification ASF} identifies a unique value for the structural function~$\mathbb E[D\vert X,\mathbf P]$ and for the average structural function~$\mu(x)$. However, since there are multiple distributions~$F_{D,\mathbf P\vert X}$ compatible with the marginals~$F_{D\vert X}$ and~$F_{\mathbf P\vert X}$, both~$\mathbb E[D\vert X,\mathbf P]$ and~$\mu(x)$ are partially identified. 

\subsection{Sharp bounds on the structural function}
We first derive sharp bounds for the conditional expectation~$\mathbb E[D\vert X=x,\mathbf P=\mathbf p]$. Under assumption~\ref{ass:invertibility}, there is a one-to-one mapping between~$\mathbf P$ and~$\eta$. Hence, we abuse notation and call~$m(x,\mathbf p)=\mathbb E[D\vert X=x,\mathbf P=\mathbf p]$ the structural function. The sharp identified set for the structural function~$m(x,\mathbf p)$ contains all the conditional expectations~$\mathbb E[D\vert X=x,\mathbf P=\mathbf p]$ compatible with marginal distributions~$F_{D\vert X}$ and~$F_{\mathbf P\vert X}$.

\begin{definition}
    For each~$x\in\mathcal X$, the sharp identified set for~$\mathbf p\mapsto m(x,\mathbf p)$ is the set of~$\mathbb E[D\vert X=x,\mathbf P=\mathbf p]$ such that~$(D,\mathbf P)\vert X=x$ has distribution in~$\mathcal M(F_{D\vert X=x},F_{\mathbf P\vert X=x})$.
\end{definition}

We first give a characterization of the identified set in the form of sharp bounds on~$m(x,\mathbf p)$.
By definition of~$m(x,\mathbf p)$, we have~
$\mathbb E[D-m(X,\mathbf P)\vert X=x,\mathbf P=\mathbf p]=0,$
which is equivalent to
\begin{eqnarray}\label{eq:Bierens}
    \mathbb E[h(\mathbf P)(D-m(X,\mathbf P))\vert X=x] & = & 0,
\end{eqnarray}
for all~$h\in L^1(F_{\mathbf P\vert X=x})$, where a function~$h:\mbox{Supp}(F_{\mathbf P\vert X=x})\subseteq[0,1]^T\rightarrow\mathbb R$ belongs to~$L^1(F_{\mathbf P\vert X=x})$ when~$\mathbb E[\vert h(\mathbf P)\vert\,\vert X=x]<\infty$. The existence of a distribution in~$\mathcal M(F_{D\vert X=x},F_{\mathbf P\vert X=x})$ such that~(\ref{eq:Bierens}) holds is equivalent to
\begin{eqnarray}\label{eq:primal}
    \min_{F}
    \mathbb E_F[h(\mathbf P)(D-m(X,\mathbf P))\vert X=x] \; \leq \; 0 \; \leq \; \max_{F}
    \mathbb E_F[h(\mathbf P)(D-m(X,\mathbf P))\vert X=x],
\end{eqnarray}
where the optimization is over~$F\in\mathcal M(F_{D\vert X=x},F_{\mathbf P\vert X=x})$, and~$h$ is any function in~$L^1(F_{\mathbf P\vert X=x})$. The bounds are obtained through an application of optimal transport duality to the optimal transport problems on the left and right-hand sides of~(\ref{eq:primal}). Because the decision~$D$ is binary and the marginal of~$\mathbf P$ is held fixed, this duality holds for merely integrable test functions~$h$, with no continuity or boundedness requirement, which is what allows us to take~$h$ equal to the density ratio~$f_{\mathbf P}/f_{\mathbf P\vert X=x}$ in the next subsection under assumption~\ref{ass:full support} alone.
This yields the following proposition.

\begin{proposition}[Sharp bounds on the structural function]
\label{prop:sharp}
A function~$m(x,\cdot)\in L^1(F_{\mathbf P\vert X=x})$ is in the identified set if and only if for all~$h\in L^1(F_{\mathbf P\vert X=x})$ such that~$h\,m(x,\cdot)\in L^1(F_{\mathbf P\vert X=x})$, we have:
\begin{eqnarray*}
    &&\sup_{\varphi\in\mathbb R} \left\{\mathbb E[\min(0,h(\mathbf P)-\varphi)\vert X=x]+\varphi\mathbb E[D\vert X=x]\right\} \\
    && \hskip50pt \leq \; \mathbb E[h(\mathbf P)m(X,\mathbf P)\vert X=x] \\
    && \hskip100pt \leq \; 
    \inf_{\varphi\in\mathbb R} \left\{\mathbb E[\max(0,h(\mathbf P)+\varphi)\vert X=x]-\varphi\mathbb E[D\vert X=x]\right\}.
\end{eqnarray*}
\end{proposition}

% I think the continuity qualifier is only needed for the current method of proof. It could be removed with a different proof.

Proposition~\ref{prop:sharp} characterizes the entire identified set
for the structural function, and it does so not by bounding~$m(x,\mathbf p)$ separately
at each~$\mathbf p$ but through a continuum of moment inequalities indexed by the test
function~$h$: a candidate belongs to the set precisely when the
functional~$\mathbb E[h(\mathbf P)\,m(X,\mathbf P)\mid X=x]$ falls within the displayed interval for
every integrable~$h$. Each such inequality is the Kantorovich dual of
an optimal transport problem with fixed marginals~$F_{D\mid X=x}$ and~$F_{\mathbf P\mid X=x}$, and the dual collapses to a one-dimensional program in the
scalar $\varphi$---concave for the lower bound, convex for the upper---so the
characterization is as convenient for computation as it is for inference. 
%This is also where the construction departs from the pointwise conditional-expectation bounds of \citet{horowitz1995identification}, \citet{cross2002regressions}, and \citet{molinari2006generalization}: testing $m$ against $h$ pools the marginal restrictions across values of $p$, whereas bounding $m(x,p)$ pointwise discards exactly this cross-$p$ information. 
The main benefit of proposition~\ref{prop:sharp} is that every linear functional of the structural function inherits sharp
bounds from a single judicious choice of~$h$. The average structural function is
the leading case, recovered by setting~$h$ equal to the density ratio~$f_\mathbf P/f_{\mathbf P\mid X=x}$, and it is to that object that we now turn.

\subsection{Sharp bounds on the average structural function}

From proposition~\ref{prop:identification ASF}, the average structural function~$\mu(x)$ can be written
\begin{eqnarray*}
    \mu(x) & = & \int \mathbb E[D\vert X=x,\mathbf P=\mathbf p]\;dF_{\mathbf P}.
\end{eqnarray*}
Notice that the integration on the right-hand side is over the unconditional distribution of elicited reports~$\mathbf P$. If we were to integrate over the conditional distribution of~$\mathbf P$ given~$X=x$, we would obtain the identified quantity~$\mathbb E[D\vert X=x]$. Hence, we can see the difference between the parameter of interest~$\mu(x)$ and the identified quantity~$\mathbb E[D\vert X=x]$ as induced by the difference between the unconditional and the conditional distribution of elicited probabilities (which are a stand-in for unobserved heterogeneity~$\eta$).

Under part~(1) of assumption~\ref{ass:full support}, we can measure the difference between conditional and unconditional distributions with the ratio~$h(\mathbf p):=f_{\mathbf P}(\mathbf p)/f_{\mathbf P\vert X=x}(\mathbf p)$. We can therefore apply a change of measure in the expression for~$\mu(x)$ to obtain
\begin{eqnarray*}
    \mu(x) & = & \int  \mathbb E[D\vert X=x,\mathbf P=\mathbf p] \; h(\mathbf P)\; dF_{\mathbf P\vert X=x} \; = \; \mathbb E[\mathbb E[D\vert X=x,\mathbf{P}]h(\mathbf{P})\vert X=x].
\end{eqnarray*}
By iterated expectations, the latter yields
\begin{eqnarray*}
    \mu(x) & = & \mathbb E[\mathbb E[Dh(\mathbf{P})\vert X=x,\mathbf{P}]\vert X=x] \; = \; \mathbb E[Dh(\mathbf{P})\vert X=x].
\end{eqnarray*}
Since the joint distribution~$\pi$ of~$(\mathbf P,D)$ is only constrained by its marginals~$(F_{\mathbf P\vert X=x},F_{D\vert X=x})$, the sharp bounds for the ASF are
\begin{eqnarray}\label{eq:primal bounds}
    \begin{array}{lllll}
      \inf\limits_\pi \mathbb E_\pi[Dh(\mathbf{P})\vert X=x] & \leq & \mu(x) & \leq \sup\limits_\pi \mathbb E_\pi[Dh(\mathbf{P})\vert X=x],
    \end{array}
\end{eqnarray}
where the optimization is over~$\pi\in\mathcal M(F_{\mathbf P\vert X=x},F_{D\vert X=x})$ (collection of joint distributions with the prescribed marginals).
\cite{fan2014identifying} and \cite{fan2025partial} suggest applying the rearrangement inequalities of \cite{cambanis1976inequalities} to this type of functional, yielding the following sharp bounds for the average structural function:
\begin{eqnarray*}
    \begin{array}{lllll}
         \int_0^1Q_{D\vert X=x}(1-t)Q_{h(\mathbf{P})\vert X=x}(t)dt & \leq & \mu(x) & \leq & \int_0^1Q_{D\vert X=x}(t)Q_{h(\mathbf{P})\vert X=x}(t)dt,
    \end{array}
\end{eqnarray*}
where~$Q_V$ is the quantile function of random variable~$V$. 
Since~$D$ is binary, $Q_{D\vert X=x}(t)=1\{t>1-p(x)\}$, where~$p(x)=\mathbb P(D=1\vert X=x)$. Then the bounds can be rewritten
\begin{eqnarray}\label{eq:rearrangement bounds}
    \begin{array}{lllll}
         \int_0^{p(x)}Q_{h(\mathbf{P})\vert X=x}(t)dt & \leq & \mu(x) & \leq & \int_{1-p(x)}^1Q_{h(\mathbf{P})\vert X=x}(t)dt.
    \end{array}
\end{eqnarray}

The rearrangement bounds in~(\ref{eq:rearrangement bounds}) have some drawbacks:
\begin{enumerate}
    \item The characterization is based on the quantile of a ratio of densities, which makes inference delicate.
    \item Restrictions on the copula of~$(D,\mathbf P)$ are difficult to incorporate, because under such additional constraints, the bounds may not be attained by the comonotonic and antimonotonic couplings. Yet, depending on the empirical application at hand, it may be natural to impose additional restrictions on the dependence structure between $(D,\mathbf P)$, which can substantially tighten the bounds.
\end{enumerate}

We overcome these drawbacks with an equivalent form of the bounds based on the dual formulation of the optimization problems in~(\ref{eq:primal bounds}). This duality theory is known as Kantorovich duality of optimal transport. This yields the following sharp bounds for the average structural function. 
\begin{theorem}\label{thm:sharp bounds}
    Under assumptions~\ref{ass:endogeneity}-\ref{ass:full support}, sharp bounds for the average structural function are:
    \begin{eqnarray}\label{eq:dual bounds}
        \begin{array}{ll}
            & \sup\limits_\varphi \left\{ \int \min(\varphi f_{\mathbf P\vert X=x}(\mathbf p),f_{\mathbf P}(\mathbf p))\;d\varrho(\mathbf p) -\varphi\mathbb E[1-D\vert X=x]\right\}
            \\ & \hskip40pt \leq \; \mu(x) \; \leq \;
            \inf\limits_\varphi \left\{ \int \max(-\varphi f_{\mathbf P\vert X=x}(\mathbf p),f_{\mathbf P}(\mathbf p))\;d\varrho(\mathbf p) +\varphi\mathbb E[1-D\vert X=x]\right\},
        \end{array}
    \end{eqnarray}
\end{theorem}
We prove theorem~\ref{thm:sharp bounds} in the appendix. The structure of the proof is the following: The Kantorovich dual of the upper bound in~(\ref{eq:primal bounds}) is
\begin{eqnarray*}
    \hskip10pt\inf_{\varphi,\psi} \left( \mathbb E[\varphi(D)\vert X=x]+\mathbb E[\psi(\mathbf P)\vert X=x]\right) \mbox{ s.t. } \varphi(d)+\psi(\mathbf p) \geq dh(\mathbf p).
\end{eqnarray*}
Hence the semi-dual, obtained by plugging the constraint into the dual objectives, is given by
{\small
\begin{eqnarray*}
    \inf_{\varphi\in\mathbb R} \left( \mathbb E[-\varphi D\vert x]+\mathbb E[\max\{0, h(\mathbf P)+\varphi\}\vert x]\right) 
    & = & \inf_{\varphi\in\mathbb R} \int \max(-\varphi f_{\mathbf P\vert x}(\mathbf p),f_{\mathbf P}(\mathbf p))\;d\varrho(\mathbf p) +\varphi\mathbb E[1-D\vert x],
\end{eqnarray*}}
which is the upper bound in~(\ref{eq:dual bounds}). Sharpness is equivalent to the absence of duality gap between the primal optimal transport problem and its Kantorovich dual.

Theorem~\ref{thm:sharp bounds} 
addresses both drawbacks listed above.
First, each bound is the value of a one-dimensional concave (lower) or convex (upper) program. Its
objective depends on the data only through the densities~$f_{\mathbf P}$ and~$f_{\mathbf P\vert X=x}$ and
the choice probability~$p(x)$. Each density enters in a piecewise linear way, so there is no need to estimate a ratio of densities. We exploit this structure for inference in section~\ref{sec:inference}. Second, we can easily incorporate restrictions on the joint distribution of~$(D,\mathbf P)$, when the latter
are expressed as moment conditions, because they enter the dual as additional multipliers. The bounds then
remain the values of low-dimensional convex programs.

\vskip5pt

\noindent{\bf No endogeneity case:} The bounds of theorem~\ref{thm:sharp bounds} contain~$\mathbb E[D\vert X=x]$. If~$0<\mathbb E[D\vert X=x]<1$, the bounds of theorem~\ref{thm:sharp bounds} reduce to the singleton~$\{\mathbb E[D\vert X=x]\}$ if and only if~$F_{\mathbf P\vert X=x}=F_{\mathbf P}$, i.e., if and only if the type~$\eta$ is independent of the event~$\{X=x\}$ (see appendix~\ref{app:proofs}).

\vskip5pt

\noindent{\bf Relation with worst-case bounds:}
Let~$\mu^-(x):=\mathbb P(D=1,X=x)$ and~$\mu^+(x):=1-\mathbb P(D=0,X=x)$ denote the worst-case bounds available without elicitation. Define~$\pi_x=\mathbb P(X=x)$. If we evaluate the objectives in~(\ref{eq:dual bounds}) at~$\varphi\in\{\pi_x,1\}$ for the lower bound and at~$\varphi\in\{-\pi_x,-1\}$ for the upper bound, we obtain the outer bounds
\begin{eqnarray*}
    \max\{\mu^-(x),\,\mathbb E(D\vert X=x)-\mathrm{TV}_x\} & \leq & \mu(x) \; \leq \; \min\{\mu^+(x),\,\mathbb E(D\vert X=x)+\mathrm{TV}_x\},
\end{eqnarray*}
where~$\mathrm{TV}_x=\int(f_{\mathbf P}(\mathbf p)-f_{\mathbf P\vert X=x}(\mathbf p))_+\;d\varrho(\mathbf p)$ is the total variation distance between the unconditional distribution of elicited probabilities and their distribution given~$X=x$. As a result, the bounds are never wider than the bounds without elicitation or than~$2\mathrm{TV}_x$.

%

% TIGHTENING BOUNDS
\subsection{Tightening the bounds under additional restrictions}

The bounds in~(\ref{eq:dual bounds}) can be tightened for more informative inference under three types of additional information, which can be leveraged individually or in combination. The additional information that can tighten data combination bounds comprises additional observed covariates common to both data sets, exclusion restrictions, and restrictions on the dependence between actual choices and stated expectations.
We shall discuss each strategy in turn before presenting bounds in the most general case, where all three strategies are combined.

\subsubsection{Additional common covariates}

Let~$W$ be a vector of additional covariates observed in both data sets. We assume $W$ is discrete with its values in $\mathcal W = \{ w_1, w_2, \ldots, w_L \}$, and that $\mathbb P(X=x,W=w_l)>0$ for all $x\in\mathcal X$ and $l=1,\ldots,L$. $W$ can help in two ways.
\begin{enumerate}
    \item Assumption~$D(x)\perp\!\!\!\!\perp X \; \vert \; \eta,W$ can be more plausible when~$W$ captures additional confounding;
    \item The data combination bounds may be tighter when assisted by additional common covariates.\footnote{Note that tightening under additional covariates is not guaranteed under current assumptions.} Covariate assisted tightening of bounds in data combination is discussed in \cite{fan2014identifying} and \cite{gaillac2024linear}.
\end{enumerate}
With additional covariates~$W$ present in both data sources, the bounds are obtained by integration of the conditional bounds over~$W$.

\subsubsection{Exclusion restrictions}
\label{sec:ER}

Suppose that the vector of observable characteristics is now~$(X,Z)$, where~$X$ is the vector of choice attributes, and~$Z$ is an additional discrete covariate with values in~$\mathcal Z = \{z_1, z_2, \ldots, z_K\}$ observed in both data sets, with $\mathbb P(X=x,Z=z_k)>0$ for all $x\in\mathcal X$ and $k=1,\ldots,K$. We extend the conditional independence assumption to~$(X,Z)$: namely~$D(x)\perp\!\!\!\!\perp (X, Z) \; \vert \; \eta$. The covariate~$Z$ is excluded in the sense that~$Z$ does not affect the potential outcome~$D(x)$. 

Our exclusion restriction defines a collection of equality constraints as follows: 
\begin{eqnarray}  \label{eq:exclusion} 
    \mathbb E [ D \vert \mathbf P = \mathbf p, X=x] - \mathbb E [ D \vert \mathbf P = \mathbf p, X=x, Z=z_k] = 0 \quad \text{for } F_{\mathbf P \vert X=x, Z=z_k} \text{-almost all } \mathbf p,
\end{eqnarray}
for all $k = 1,2,\ldots, K$. Both conditional expectations in~(\ref{eq:exclusion}) are defined~$F_{\mathbf P\vert X=x,Z=z_k}$-almost everywhere, because~$F_{\mathbf P\vert X=x,Z=z_k}\ll F_{\mathbf P\vert X=x}$, and~$Z$ may shift the support of the reports. The bounds under (\ref{eq:exclusion}) can be expressed through the following Lagrangian formulations of constrained optimal transport problems.
\begin{eqnarray*}
    \begin{array}{ll}
       &\sup\limits_\gamma \inf\limits_\pi \ \mathbb E_\pi \left[ D h(\mathbf P) \vert X=x \right] + \int \sum_{k=1}^K \gamma_k (\mathbf p) \mathbb P \{ Z=z_k \vert X=x \} f_{\mathbf P \vert X=x, Z=z_k} (\mathbf p) \\ 
       &\hskip130pt \left\{ \mathbb E_\pi [ D \vert \mathbf P = \mathbf p, X=x] - \mathbb E_\pi [ D \vert \mathbf P = \mathbf p, X=x, Z=z_k] \right\} d\varrho(\mathbf p) \\ 
       &\hskip25pt \le \; \mu(x) \; \le \\ 
       &\inf\limits_\gamma \sup\limits_\pi \ \mathbb E_\pi \left[ D h(\mathbf P) \vert X=x \right] + \int \sum_{k=1}^K \gamma_k (\mathbf p) \mathbb P \{ Z=z_k \vert X=x \} f_{\mathbf P \vert X=x, Z=z_k} (\mathbf p) \\ 
       &\hskip130pt \left\{ \mathbb E_\pi [ D \vert \mathbf P = \mathbf p, X=x] - \mathbb E_\pi [ D \vert \mathbf P = \mathbf p, X=x, Z=z_k] \right\} d\varrho(\mathbf p), 
    \end{array}
\end{eqnarray*}
where the outer optimization is over $\gamma := (\gamma_1,\gamma_2,\ldots,\gamma_K)'$, with each $\gamma_k:\operatorname{Supp}(F_{\mathbf P\mid X=x})\to\mathbb R$ being measurable and integrable with respect to $F_{\mathbf P\mid X=x,Z=z_k}$, and the inner optimization is over $\pi \in \mathcal M_x$, with
\[
\mathcal M_x := \left\{ \pi \in \mathcal P([0,1]^T \times \{0,1\} \times \mathcal Z): \pi_{\mathbf P,Z}=F_{\mathbf P,Z\mid X=x}, \quad \pi_{D,Z}=F_{D,Z\mid X=x} \right\}.
\]
Here, $\mathcal P(S)$ denotes the set of probability distributions on $S$, and $\pi_{\mathbf P,Z}$ and $\pi_{D,Z}$ denote the corresponding marginals of $\pi$.

Taking the Kantorovich dual of each inner problem, interchanging the order of optimization, and solving the resulting inner problem with respect to $\gamma$ yields
\begin{eqnarray*}
    \begin{array}{ll}
       &\sup\limits_\varphi \ \big( \int \min \left\{ \sum_{k=1}^{K} \varphi_k \mathbb P \left\{ Z=z_k \vert X=x \right\} f_{\mathbf P \vert X=x, Z=z_k} (\mathbf p), f_{\mathbf P}(\mathbf p) \right\} d\varrho(\mathbf p) \\ 
       &\hskip150pt - \sum_{k=1}^K \varphi_k \mathbb P \left\{ Z=z_k \vert X=x \right\} \mathbb E [1-D \vert X=x, Z=z_k]\big) \\
       &\hskip25pt \le \; \mu(x) \; \le \\ 
       &\inf\limits_\varphi \ \big( \int \max \left\{ \sum_{k=1}^{K} \varphi_k \mathbb P \left\{ Z=z_k \vert X=x \right\} f_{\mathbf P \vert X=x, Z=z_k} (\mathbf p), f_{\mathbf P}(\mathbf p) \right\} d\varrho(\mathbf p) \\ 
       &\hskip150pt - \sum_{k=1}^K \varphi_k \mathbb P \left\{ Z=z_k \vert X=x \right\} \mathbb E [1-D \vert X=x, Z=z_k]\big), 
    \end{array}
\end{eqnarray*}
where the optimization is over $\varphi := (\varphi_1,\varphi_2,\ldots,\varphi_K)' \in \mathbb R^K$. Although (\ref{eq:exclusion}) may impose infinitely many equality constraints when $\mathbf P$ is continuous, the final characterization involves only $K$ optimization variables. Also, restricting $\varphi$ to satisfy $\varphi_1=\varphi_2=\cdots=\varphi_K$ recovers the bounds in (\ref{eq:dual bounds}), the upper bound after the change of variable~$\varphi\mapsto-\varphi$, which does not affect its value since~$\varphi$ ranges over all of~$\mathbb R$; the bounds under the exclusion restriction are therefore at least as tight as those without it. 

Note that $D(x)\perp\!\!\!\!\perp (X,Z)\mid\eta$ holds if and only if $D(x)\perp\!\!\!\!\perp X\mid\eta$ and $D(x)\perp\!\!\!\!\perp Z\mid\eta,X$. The first condition, $D(x)\perp\!\!\!\!\perp X\mid\eta$, has already been used to identify our parameter of interest in Proposition~\ref{prop:identification ASF}. Under Assumption~\ref{ass:invertibility}, the second condition restricted to the event~$\{X=x\}$, namely $D(x)\perp\!\!\!\!\perp Z\mid\eta,X=x$, is equivalent to the equality constraints (\ref{eq:exclusion}). Its restriction to the events~$\{X=x^\prime\}$, $x^\prime\neq x$, concerns the potential decision~$D(x)$ of agents for whom it is never observed, and does not further restrict the identified set for~$\mu(x)$, as the proof of theorem~\ref{thm:bounds} in the appendix shows. Hence, the bounds derived above fully exploit the assumption $D(x)\perp\!\!\!\!\perp (X,Z)\mid\eta$ and are sharp.

% DEPENDENCE RESTRICTIONS

\subsubsection{Dependence restrictions}
\label{sec:RD}

The bounds of~(\ref{eq:dual bounds}) are sharp when the marginal distributions~$F_{D\vert X=x}$ and~$F_{\mathbf P\vert X=x}$ are identified from each data set, but there is no additional information about the joint distribution~$F_{(D,\mathbf P)\vert X=x}$. The bounds are attained at extremal couplings. The lower bound assigns~$D=1$ to the individuals with the smallest values of the density ratio~$h(\mathbf P)$ and the upper bound to those with the largest. When the reports~$\mathbf P$ are stated probabilities of making decision~$D=1$, the extremal couplings may require that the agents who choose~$D=1$ in the revealed preference data set be exactly those that report the lowest probability of choosing~$D=1$ in the stated preference data set. One way to impose that the stated probabilities have some predictive content is to restrict the copula in ways that avoid such implausible scenarios as the case, where agents who choose~$D=1$ are the ones that report the lowest probability of doing so.

Since~$D$ is binary, any restriction on the joint distribution of~$(D,\mathbf P)$ is a restriction on the way the distribution of the stated reports~$\mathbf P$ is allocated between the two decision groups~$D=0$ and~$D=1$. Hence it is a restriction on the function~$\mathbf p\mapsto\mathbb E[D\vert \mathbf P=\mathbf p,X=x]$. We consider the class of linear restrictions on this function, hence moment inequalities
\begin{eqnarray}\label{eq:copula}
    \mathbb E[\,\zeta(D,\mathbf P\vert x)\,\vert\, X=x\,] & \leq & 0,
\end{eqnarray}
for any function~$(d,\mathbf p)\mapsto\zeta(d,\mathbf p\vert x)$ such that~$\mathbf p\mapsto\zeta(d,\mathbf p\vert x)$ is integrable with respect to~$F_{\mathbf P\vert X=x}$ for~$d\in\{0,1\}$, as required in assumption~\ref{ass:copula} below. With~$D$ binary, writing~$\zeta_d(\mathbf p\vert x):=\zeta(d,\mathbf p\vert x)$ for~$d\in\{0,1\}$, (\ref{eq:copula}) is equivalent to~$\mathbb E[(D-\mathbb E[D\vert X=x])(\zeta_1-\zeta_0)(\mathbf P\vert x)\vert X=x]\leq\chi_x$, where the constant~$\chi_x:=-\mathbb E[\zeta_0(\mathbf P\vert x)\vert X=x]-\mathbb E[D\vert X=x]\,\mathbb E[(\zeta_1-\zeta_0)(\mathbf P\vert x)\vert X=x]$ is identified: it is a restriction on the covariance between the decision and a known transformation of the reports.

\begin{continued}[Example~\ref{ex:pilot} continued:]
In the retirement example, assume that the stated reports are~$\mathbf P:=(P_x)_x$ and each~$P_x=\tilde{\mathbb E}[D\vert X=x]$ is the reported subjective expected decision. It is plausible to assume that stated probabilities are predictive of actual decisions in the sense that~$P_x$ and~$D$ are positively dependent, conditional on~$X=x$. Positive rank dependence\footnote{Rank dependence only constrains the copula, not the marginals.} is a special form of~(\ref{eq:copula}), namely
$\mathbb E[(D-\mathbb E[D\vert x])(0.5-F_{P_x\vert x}(P_x))\vert x]\leq0$.
A stronger constraint on the way subjective reported probabilities of retiring~$P_x$ are predictive of decision~$D=1$ is the assumption that the structural function is monotone in the report, i.e.,~$\mathbb E[D\vert P_x=p^\prime,X=x]\geq\mathbb E[D\vert P_x=p,X=x]$ if~$p^\prime\geq p$. When~$P_x$ alone recovers~$\eta$, as in the retirement example, where~$\tilde m(x,\cdot)$ is increasing, this is stochastic monotonicity of the potential decision~$D(x)$ in the report~$P_x$, i.e.,~$\mathbb P(D(x)=1\vert P_x=p^\prime)\geq\mathbb P(D(x)=1\vert P_x=p)$ if~$p^\prime\geq p$; in general, selection into~$X=x$ may depend on components of~$\eta$ that~$P_x$ does not recover, and the two conditions differ. If~$0<p(x)<1$, monotonicity of the structural function implies that the distribution of~$P_x$ given~$(D=1,X=x)$ stochastically dominates the distribution of~$P_x$ given~$(D=0,X=x)$, which is equivalent to the continuum of moment inequalities~$\mathbb E[(D-p(x))1\{P_x\leq t\}\vert x]\leq0$ for each~$t\in[0,1]$.

% \paragraph{A scale-free measure of dependence.}
% Rank the reports within the population with~$X=x$ and define
% \begin{eqnarray}\label{eq:gamma}
%     \gamma_x & := & \frac{\mathrm{Cov}\big(D,F_{P_x\vert X=x}(P_x)\,\big\vert\, X=x\big)}{\frac12\,p(x)\big(1-p(x)\big)}.
% \end{eqnarray}
% The denominator is the covariance under the comonotone coupling, so~$\gamma_x\in[-1,1]$. The value~$\gamma_x=1$ obtains when~$D=1$ for exactly the agents whose reports lie in the upper~$p(x)$-tail of~$F_{P_x\vert X=x}$, the value~$\gamma_x=-1$ in the antimonotone case, and~$\gamma_x=0$ under independence. If~$P_x'$ and~$P_x''$ are independent draws from the reports of agents with~$X=x$ who chose~$D=1$ and~$D=0$ respectively, then
% \begin{eqnarray*}
%     \gamma_x & = & \mathbb P(P_x'>P_x'')-\mathbb P(P_x'<P_x''),
% \end{eqnarray*}
% so~$\gamma_x$ is Somers'~$D$ of~$P_x$ with respect to~$D$, or~$2\,\mathrm{AUC}_x-1$, where~$\mathrm{AUC}_x$ is the area under the ROC curve of the stated probability as a predictor of the realized decision. Two features make~$\gamma_x$ the natural object to restrict. It is a functional of the copula of~$(D,P_x)$ given~$X=x$ alone, so restricting it is compatible with any degree of hypothetical bias in the level of~$P_x$, including systematic over- or under-statement. And it has a direct interpretation: it measures how well stated probabilities discriminate between the agents who go on to choose~$D=1$ and those who do not.

\end{continued}

The bounds under~(\ref{eq:copula}) are solutions of constrained optimal transport problems
\begin{eqnarray*}
    \sup_{\lambda\geq0}\inf_{\pi} \mathbb E_\pi[Dh(\mathbf P)+\lambda \zeta(D,\mathbf P\vert x)\vert x]  \leq \mu(x) \leq \inf_{\lambda\geq0}\sup_{\pi} \mathbb E_\pi[Dh(\mathbf P)-\lambda \zeta(D,\mathbf P\vert x)\vert x].
\end{eqnarray*}
Taking the Kantorovich dual of each of the inner programs gives the following characterization:
\begin{eqnarray*}
    \begin{array}{l}
       \sup\limits_{\varphi} \Big( -\varphi \mathbb E[1-D\vert x]+\sup\limits_{\lambda\geq0}\int\min\big\{\left(\varphi+\lambda \zeta(0,\mathbf p\vert x)\right)f_{\mathbf P\vert x}(\mathbf p), \\
       \hskip170pt f_{\mathbf P}(\mathbf p)+\lambda \zeta(1,\mathbf p\vert x)f_{\mathbf P\vert x}(\mathbf p)\big\}d\varrho(\mathbf p)\Big)
       \\ \hskip45pt \leq \; \mu(x) \; \leq \\
        \inf\limits_{\varphi} \Big( -\varphi \mathbb E[1-D\vert x]+\inf\limits_{\lambda\geq0}\int\max\big\{\left(\varphi-\lambda \zeta(0,\mathbf p\vert x)\right)f_{\mathbf P\vert x}(\mathbf p), \\
       \hskip170pt f_{\mathbf P}(\mathbf p)-\lambda \zeta(1,\mathbf p\vert x)f_{\mathbf P\vert x}(\mathbf p)\big\}d\varrho(\mathbf p)\Big).
    \end{array}
\end{eqnarray*}

\subsubsection{Bringing it all together}

Our main identification result provides bounds in the most general case, including additional covariates, potential exclusion restrictions, and potential dependence constraints. First, to allow for additional covariates and potential exclusion restrictions, we replace assumptions~\ref{ass:endogeneity} and \ref{ass:full support} with the following variations.
 
\begin{assumptionp}{\ref{ass:endogeneity}$'$}%[Unobserved heterogeneity]
\label{ass:endogeneity'}
    $D(x) \perp\!\!\!\!\perp (X,Z) \; \vert \; \eta, W$ for each~$x \in \mathcal X$.
\end{assumptionp}

\begin{assumptionp}{\ref{ass:full support}$'$}%[Full support]
\label{ass:full support'}
    For all~$(x,z,w)\in\mathcal X\times\mathcal Z\times\mathcal W$, $\mathbb P(X=x,Z=z,W=w)>0$ and: (1) $F_{\mathbf P \vert W=w} \ll F_{\mathbf P\vert X=x, W=w}$; (2) $f_{\mathbf P \vert W=w},f_{\mathbf P\vert X=x, W=w}, f_{\mathbf P\vert X=x, Z=z, W=w}\in\ell^\infty([0,1]^T)$.
\end{assumptionp}

Second, we impose the dependence restriction conditionally on the additional covariates~$W$ and the excluded variable~$Z$, in the form of a conditional counterpart of~(\ref{eq:copula}) in each~$(X,Z,W)$ cell. 

\begin{assumption}\label{ass:copula} There is a function~$(d,\mathbf p)\mapsto\zeta(d,\mathbf p\vert x,z,w)$, known or identified from the stated preference data, such that the following holds for each~$(x,z,w) \in \mathcal X \times \mathcal Z \times \mathcal W$:
\begin{enumerate}
    \item The function~$\mathbf p\mapsto\zeta(d,\mathbf p\vert x,z,w)$ is integrable with respect to~$F_{\mathbf P\vert X=x,Z=z,W=w}$ for each~$d\in\{0,1\}$;
    \item $\mathbb E[\zeta(D,\mathbf P\vert x,z,w)\vert X=x,Z=z,W=w] \leq0.$
\end{enumerate}
\end{assumption} 

Assumptions~\ref{ass:endogeneity'}, \ref{ass:invertibility}, and~\ref{ass:full support'} imply the following equality constraints:
\[
\mathbb E [D \vert \mathbf P = \mathbf p, X=x, W=w_l] - \mathbb E [D \vert \mathbf P = \mathbf p, X=x, Z=z_k, W=w_l] = 0
\]
for $F_{\mathbf P \vert X=x, Z=z_k, W=w_l}$-almost every $\mathbf p$, for all $k=1,\ldots,K$ and $l=1,\ldots,L$. The second part of Assumption~\ref{ass:copula} imposes inequality constraints. Solving the optimal transport problem subject to both sets of constraints yields the following bounds on the average structural function~$\mu(x)$:

\begin{theorem}
    \label{thm:bounds}
    Under assumptions~\ref{ass:endogeneity'}, \ref{ass:invertibility}, \ref{ass:full support'} and~\ref{ass:copula}, sharp bounds on~$\mu(x)$ for each $x \in \mathcal{X}$ are: 
    \begin{eqnarray*}
        &&\sum_{l=1}^{L} \left( \sup\limits_{\varphi_l \: \in \: \mathbb R^K} \sup\limits_{\lambda_l \: \in \: \mathbb R_+^K} L(\lambda_l, \varphi_l,x,l) \right) \mathbb P \{W=w_l\}\\
        &&\hskip80pt \leq \; \mu(x) \; \leq \; \sum_{l=1}^{L} \left( \inf\limits_{\varphi_l \: \in \: \mathbb R^K} \inf\limits_{\lambda_l \: \in \: \mathbb R_+^K} U(\lambda_l, \varphi_l,x,l) \right) \mathbb P \{W=w_l\}, 
    \end{eqnarray*}
    where $s_{k,l} := \mathbb P \{ Z=z_k \vert X=x, W=w_l\}$, $f_{k,l} (\mathbf p) := f_{\mathbf P \vert X=x, Z=z_k, W=w_l} (\mathbf p)$, 
    \begin{eqnarray*}
        \begin{array}{ll}
           L(\lambda_l, \varphi_l,x,l) \ := &- \sum_{k=1}^{K} \varphi_{k,l} s_{k,l} \mathbb E \left[ 1-D \vert X=x, Z=z_k, W=w_l \right] \\ 
           &+ \int \min \big\{ \sum_{k=1}^K  (\varphi_{k,l} + \lambda_{k,l} \zeta(0,\mathbf p \vert x, z_k, w_l)) s_{k,l} f_{k,l} (\mathbf p), \\ 
           &\hskip55pt f_{\mathbf P \vert W=w_l} (\mathbf p) + \sum_{k=1}^K  \lambda_{k,l} \zeta(1,\mathbf p \vert x, z_k, w_l) s_{k,l} f_{k,l}(\mathbf p) \big\} d\varrho(\mathbf p),  
        \end{array}
    \end{eqnarray*}
    and 
    \begin{eqnarray*}
        \begin{array}{ll}
           U(\lambda_l, \varphi_l,x,l) \ := &- \sum_{k=1}^{K} \varphi_{k,l} s_{k,l} \mathbb E \left[ 1-D \vert X=x, Z=z_k, W=w_l \right] \\ 
           &+ \int \max \big\{ \sum_{k=1}^K  (\varphi_{k,l} - \lambda_{k,l} \zeta(0,\mathbf p \vert x, z_k, w_l)) s_{k,l} f_{k,l}(\mathbf p), \\ 
           &\hskip55pt f_{\mathbf P \vert W=w_l} (\mathbf p) - \sum_{k=1}^K  \lambda_{k,l} \zeta(1,\mathbf p \vert x, z_k, w_l) s_{k,l} f_{k,l}(\mathbf p) \big\} d\varrho(\mathbf p),  
        \end{array}
    \end{eqnarray*}
\end{theorem}

Theorem~\ref{thm:bounds} is the paper's main identification result for the unmatched case and provides a finite-dimensional characterization of the bounds. For each $W=w_l$, the bound is obtained by optimizing over the $2K$-dimensional vector $(\lambda_l,\varphi_l)\in\mathbb R_+^K\times\mathbb R^K$. The dimension of the elicited-probability vector $\mathbf P$ affects the integral over $\mathbf p$, but not the number of optimization variables. Both bounds are convex optimization problems: the upper-bound objective is jointly convex in $(\varphi_l,\lambda_l)$, and the lower-bound objective is jointly concave. Partial minimization of $U$ over $\lambda_l\ge0$ and partial maximization of $L$ over $\lambda_l\ge0$ therefore yield convex and concave functions of $\varphi_l$, respectively. Under the maintained assumptions, both partially optimized functions are finite. The dual optima need not be attained at finite $(\lambda_l,\varphi_l)$ under Assumptions~\ref{ass:full support'} and~\ref{ass:copula}, so the bounds are expressed as suprema and infima. When $\zeta\equiv0$ and $K=1$, the dual optima are attained at finite $\varphi_l$ whenever $0<\mathbb E[D\mid X=x,W=w_l]<1$, by step~5 of the proof of theorem~\ref{thm:sharp bounds}. Theorem~\ref{thm:bounds} imposes one dependence restriction per cell. Finitely many restrictions per cell, $\zeta^1,\ldots,\zeta^M$, are handled identically, with a vector~$\lambda_{k,l}\in\mathbb R_+^M$ of multipliers in place of the scalar, so that the number of optimization variables per value of~$W$ becomes~$K(1+M)$; a continuum of restrictions, such as the stochastic dominance restriction of section~\ref{sec:RD}, requires a functional multiplier and is not covered by the finite-dimensional characterization. 

Suppose \(W\) is degenerate (\(L=1\)). The theorem nests all the preceding bounds as special cases. Setting \(\zeta\equiv 0\) eliminates the multiplier \(\lambda\) and the corresponding inner optimization, yielding the bounds under the exclusion restriction alone in Section~\ref{sec:ER}. If \(Z\) is also degenerate (\(K=1\)), these further reduce to the baseline bounds in~\eqref{eq:dual bounds}, up to the change of variable~$\varphi\mapsto-\varphi$ in the upper bound noted in section~\ref{sec:ER}. Alternatively, taking \(Z\) to be degenerate (\(K=1\)) while retaining \(\zeta\) recovers the bounds under the dependence restriction alone in Section~\ref{sec:RD}.

%

% PARTIAL RECOVERY
\section{Partial recovery of unobserved types}
\label{sec:partial}

Assumption~\ref{ass:invertibility} requires the elicited reports to separate types. We now consider cases, where the elicited reports are not rich enough to separate types and assumption~\ref{ass:invertibility} fails. We formalize the failure by
assuming that there is a latent report~$\mathbf P$, which satisfies
assumption~\ref{ass:invertibility}, but which is unobserved.  Instead, we observe reports~$\tilde{\mathbf P}$
that only partially record~$\mathbf P$. We write~$\pi_x:=\mathbb P(X=x)$ throughout, and denote~$\mu^-(x):=\mathbb P(D=1,X=x)$ and~$\mu^+(x):=1-\mathbb P(D=0,X=x)$ the worst case bounds available from revealed preferences alone, without any elicitation. We show in this section that strong additional assumptions, typically shape restrictions, are needed to tighten the worst-case bounds~$\mu(x)\in[\mu^-(x),\mu^+(x)]$. In other words, in this framework, it is very costly to relax assumption~\ref{ass:invertibility}, even if we work with matched data, i.e., if stated and revealed preferences are observed for the same individuals.

\begin{assumption}[Partial recovery]\label{ass:coarsening}
    There is a known correspondence~$G$ from~$[0,1]^T$ to the subsets of~$[0,1]^T$, with non-empty values and Borel measurable graph~$\{(\mathbf p,\tilde{\mathbf p}):\tilde{\mathbf p}\in G(\mathbf p)\}$, such that~$\tilde{\mathbf P}\in G(\mathbf P)$ almost surely.
\end{assumption}

Under assumption~\ref{ass:coarsening}, the observed report~$\tilde{\mathbf P}$ tells us that the true latent report~$\mathbf P$ is in~$G^{-1}(\tilde{\mathbf P}):=\{\mathbf p:\tilde{\mathbf P}\in G(\mathbf p)\}$. The leading case is coarsening by a partition, which covers bracketed survey responses (when the question in the survey is, say ``Which range does your probability of retiring at age 65 fall into: 0-20\%, 20-40\%, 40-60\%, 60-80\%, or 80-100\%?''), rounding, and coarse classification of respondents into types.

\begin{example}[Coarsening by a partition]\label{ex:partition}
    Let~$\{B_1,\ldots,B_K\}$ be a partition of the report space. Consider the correspondence defined by~$G(\mathbf p)=B_i$ whenever~$\mathbf p\in B_i$. The observed report is then an arbitrary point of the cell that contains the true latent report: 
    $G^{-1}(\tilde{\mathbf p})=B_i$ when~$\tilde{\mathbf p}\in B_i$. The only relevant information from the observed report is the cell membership. Interval censored reports are a special case. 
\end{example}

% \begin{example}[Incomplete elicitation]\label{ex:incomplete}
%     In the retirement example, suppose the unobserved type~$\eta=(\eta_1,\eta_2)$ is two dimensional, comprising, for instance, taste for leisure and health consciousness, so that the two reports~$\mathbf P=(P_0,P_1)$ are needed for assumption~\ref{ass:invertibility} to hold, but only~$P_0$ is elicited. Then~$\tilde{\mathbf P}=P_0$ and~$G(\mathbf p)=\{\tilde{\mathbf p}:\tilde p_0=p_0\}$.
% \end{example}

%

\subsection{Worst-case bounds}
\label{sec:impossibility}

We first show that, even when stated and revealed preferences are matched, partial recovery of types essentially reduces the sharp bounds to the worst case bounds.

\begin{theorem}[Worst-case]\label{thm:impossibility}
    Assume that~\ref{ass:endogeneity}--\ref{ass:full support} and~\ref{ass:coarsening} hold, with~$G$ a coarsening by a partition as in example~\ref{ex:partition}, and that the distribution of~$(D,\tilde{\mathbf P},X)$ is identified. Assume in addition that~$G^{-1}(\tilde{\mathbf p})$ is~$F_{\tilde{\mathbf P}}$-almost surely not a singleton. Then, the worst case bounds~$[\mu^-(x),\mu^+(x)]$ are sharp if and only if
    \begin{eqnarray*}
        \mathbb P\left(\,\mathbb E\left[D\vert X=x,G^{-1}(\tilde{\mathbf P})\right]\in\{0,1\}\,,\,X\ne x\,\right) & = & 0.
    \end{eqnarray*}
\end{theorem}

Theorem~\ref{thm:impossibility} shows that, when types are only partially recovered, stated preferences can only help tighten the bounds on~$\mu(x)$ through the observed reports at which agents with attribute~$x$ all make the same decision. In the context of example~\ref{ex:partition}, the only tightening relative to the worst case bounds is obtained from cells where the decision is unanimous. On any other cell, the mass can be moved to a region where~$D=1$, or to a region where~$D=0$.

\subsection{Shape restrictions under interval coarsening}
\label{sec:coarsening}

In this section, we discuss monotonicity restrictions that help tighten the worst-case bounds in the case of coarsened reports. Stating such monotonicity restrictions requires ordered reports. Hence we concentrate on the case, with scalar unobserved heterogeneity and scalar and continuously distributed elicited probabilities. Throughout this subsection, $\mathbf P=P_x$, the stated expectation of~$D$ in state~$x$, and~$G$ is coarsening by a partition of~$[0,1]$ into non degenerate intervals~$B_1<\cdots<B_K$, as in example~\ref{ex:partition}. 

We propose two monotonicity assumptions.

\begin{assumption}[Monotone response]\label{ass:monotone m}
    The function~$p\mapsto \mathbb E[D\vert x,p]$ is non decreasing.
\end{assumption}

Assumption~\ref{ass:monotone m} requires agents who state a higher probability of choosing~$D=1$ at attribute value~$x$ to be more likely to do so when they face that value. It is the monotonicity restriction discussed in section~\ref{sec:RD}.

\begin{assumption}[Monotone types]\label{ass:MLR1}
    $p\mapsto\mathbb P(X=x\vert P_x=p)$ is non decreasing.
\end{assumption}

Assumption~\ref{ass:MLR1} requires agents who state a higher probability of choosing~$D=1$ at attribute value~$x$ to be more likely to realize that value. By Bayes' rule, the density ratio~$h=f_P/f_{P\vert X=x}$ is equal to~$\pi_x/\mathbb P(X=x\vert P_x=\cdot)$, so that assumption~\ref{ass:MLR1} is a monotone likelihood ratio property. %It holds if and only if~$L$ is concave, which implies~$\sigma_1\geq\cdots\geq\sigma_K$.
The two assumptions are not nested and can be imposed jointly.

\subsubsection{Matched data}
\label{sec:coarsening matched}

With matched data, the cell averages~$s_i:=\mathbb E[D\vert X=x,P_x\in B_i]$ of the structural choice function are identified. The average structural function~$\mu(x)$ is the integral of the structural function~$p\mapsto \mathbb E[D\vert X=x,P_x=p]$ over the marginal distribution of~$P_x$:
\begin{eqnarray*}
    \mu(x) & = & \int_0^1 \mathbb E[D\vert X=x,P_x=p] \; dF_{P_x}(p) \;
     = \; \sum_{i=1}^K\int_{B_i}\mathbb E[D\vert X=x,P_x=p] \; dF_{P_x}(p) \\
     & = & \mu^-(x)+\sum_{i=1}^K\mathbb E\left[\mathbb E[D\vert X=x,P_x]1\{P_x\in B_i,X\neq x\}\right].
\end{eqnarray*}
Under assumption~\ref{ass:monotone m}, with the conventions~$s_0:=0$ and~$s_{K+1}:=1$,
the structural function is bounded on the cell~$B_i$ by~$\underline s_i:=s_i1\{s_i=s_{i+1}\}+s_{i-1}1\{s_i<s_{i+1}\}$ and by~$\bar s_i:=s_i1\{s_{i-1}=s_i\}+s_{i+1}1\{s_{i-1}<s_i\}$, and the counterfactual mass~$\mathbb P(P_x\in B_i, X\ne x)$ can be concentrated arbitrarily close to the part of the cell where the structural function is smallest (resp. largest). Under the overlap assumption~\ref{ass:full support}, this mass cannot in general be placed exactly there, so the bounds below are sharp as an infimum and a supremum but need not be attained; the proof in appendix~\ref{app:partial} establishes sharpness through a sequence of admissible allocations. Hence, we have the following sharp bounds.
\begin{proposition}[Sharp bounds under monotone response]\label{prop:coarsened monotone}
    Under assumptions~\ref{ass:endogeneity}--\ref{ass:full support}, \ref{ass:coarsening} and~\ref{ass:monotone m}, with coarsening by a partition, the sharp bounds on~$\mu(x)$ are
    \begin{eqnarray*}
        \mu^-(x)+\sum_{i=1}^K\underline s_i\;\mathbb P(P_x\in B_i, X\ne x)
        & \leq \; \mu(x) \; \leq &
        \mu^-(x)+\sum_{i=1}^K\bar s_i\;\mathbb P(P_x\in B_i, X\ne x).
    \end{eqnarray*}
\end{proposition}
The width of the bounds of proposition~\ref{prop:coarsened monotone} is~$\sum_i(\bar s_i-\underline s_i)\;\mathbb P(P_x\in B_i, X\ne x)$. Without additional restrictions, the agents within each cell with~$X\neq x$ can all be assigned stated probabilities at one end of the interval. Assumption~\ref{ass:MLR1} rules out such concentrations, as it restricts the distribution of latent reports among agents with~$X\neq x$ relative to their distribution among agents with~$X=x$. Define the concentration curve~$L$ as follows:
\begin{eqnarray*}
    \begin{array}{lll}
       L  & : & [0,1] \; \rightarrow \; [0,1-\pi_x] \\
                && v \; \mapsto \; L(v) := \mathbb P(F_{P_x\vert X=x}(P_x)\leq v,X\ne x).
    \end{array} 
\end{eqnarray*}
By construction, $L(0)=0$ and~$L(1)=1-\pi_x$ and~$L$ is non-decreasing. It is also continuous by part~(1) of assumption~\ref{ass:full support}. In addition, for each~$i=1\ldots,K,$ its value
\begin{eqnarray}\label{eq:concentration}
    L\left( \; \sum_{l\leq i} \mathbb P(P_x\in B_l\vert X=x) \right) & = & \sum_{l\leq i} \; \mathbb P(P_x\in B_l,X\ne x)
\end{eqnarray}
is identified from the stated coarsened reports. Now, assumption~\ref{ass:MLR1} holds if and only if the concentration curve~$L$ is concave. We therefore define the set of admissible concentration curves as follows:
\begin{definition}[Admissible concentration curves]
    The set~$\mathcal L$ of admissible concentration curves is defined as the set of functions~$L:[0,1]\rightarrow[0,1-\pi_x]$ that are continuous, non-decreasing and concave, and satisfy~$L(0)=0$ and~(\ref{eq:concentration}) for~$i=1,\ldots,K$.
\end{definition}
Denoting
\begin{eqnarray*}
    c_i \; := \; \sum_{l\leq i} \mathbb P(P_x\in B_l\vert X=x) & \mbox{and} & \kappa_i^d \; := \; \mathbb P(D=d,P_x\in B_i\vert X=x),
\end{eqnarray*}
we therefore have the following sharp bounds under monotone types.
\begin{proposition}[Sharp bounds under monotone types]\label{prop:coarsened MLR}
    Under assumptions~\ref{ass:endogeneity}--\ref{ass:full support}, \ref{ass:coarsening} and~\ref{ass:MLR1}, with coarsening by a partition, the sharp bounds on~$\mu(x)$ are
    {\small\begin{eqnarray*}
        \mu^-(x) \; + \; \min_{L\in\mathcal L} \; \sum_{i=1}^K\big(L(c_i)-L(c_i-\kappa_i^1)\big)
        & \leq \; \mu(x) \; \leq &
        \mu^+(x) \; - \; \min_{L\in\mathcal L} \; \sum_{i=1}^K\big(L(c_i)-L(c_i-\kappa_i^0)\big).
    \end{eqnarray*}}
\end{proposition}
The upper bound sets the structural function~$m(x,\cdot)$ to one on the lowest~$\kappa_i^1$ ranks of each cell~$B_i$, and to zero on the rest of the cell, and symmetrically for the lower bound. Each minimum is the value of a linear program in the slopes of~$L$ on the segments~$(c_{i-1},c_i-\kappa_i^d]$ and~$(c_i-\kappa_i^d,c_i]$ with~$2K$ variables, $K$ equality constraints (the identified values of~$L$) and~$2K$ inequality constraints (monotonicity and concavity of~$L$). 

To obtain the bounds when both assumptions are combined, notice that assumption~\ref{ass:monotone m} has a similar characterization to assumption~\ref{ass:MLR1}. Define the curve
\begin{eqnarray*}
    \begin{array}{lll}
       R  & : & [0,1] \; \rightarrow \; [0,p(x)] \\
                && v \; \mapsto \; R(v) := \mathbb P\big(D=1,F_{P_x\vert X=x}(P_x)\leq v\,\big\vert\, X=x\big).
    \end{array}
\end{eqnarray*}
The curve~$R$ is absolutely continuous, with~$R(0)=0$, and its values~$R(c_i)=\sum_{l\leq i}\kappa_l^1$ at the cell boundaries are identified. Assumption~\ref{ass:monotone m} therefore holds if and only if~$R$ is convex. Let~$\mathcal R$ be the set of convex functions~$R:[0,1]\rightarrow[0,p(x)]$ such that~$R(0)=0$, $0\leq R(v)-R(u)\leq v-u$ for~$0\leq u\leq v\leq1$, and~$R(c_i)=\sum_{l=1}^i\kappa_l^1$, $i=1,\ldots,K$. Sharp bounds under assumptions~\ref{ass:monotone m} and~\ref{ass:MLR1} are given by the following proposition.
\begin{proposition}[Sharp bounds under monotone types and monotone response]\label{prop:coarsened joint}
    Under assumptions~\ref{ass:endogeneity}--\ref{ass:full support}, \ref{ass:coarsening}, \ref{ass:monotone m} and~\ref{ass:MLR1}, with coarsening by a partition, the sharp bounds on~$\mu(x)$ are
    \begin{eqnarray*}
        \mu^-(x)+\inf_{L\in\mathcal L,\,R\in\mathcal R}\int_0^1R^\prime(u)\,dL(u) & \leq \; \mu(x) \; \leq & \sum_{i=1}^Ks_i\,\mathbb P(P_x\in B_i).
    \end{eqnarray*}
\end{proposition}

The upper bound is the value we would obtain from proposition~\ref{prop:identification ASF} if the indicators~$1\{P_x\in B_i\}$, $i=1,\ldots,K,$ were sufficient to control for endogeneity, i.e., if we had~$D(x)\perp\!\!\!\!\perp X\,\vert\, (1\{P_x\in B_i\})_i$. The bound is attained when the structural function is constant on each cell. The lower bound is the infimum of a bilinear functional over two convex sets of curves. It has no closed form in general. The objective is linear in each curve separately, so that the infimum can be approached over piecewise linear concave~$L$ and piecewise linear convex~$R$; the locations of their kinks inside the cells are optimization variables, however, and the resulting finite-dimensional program is bilinear rather than linear.

\subsubsection{Unmatched data}
\label{sec:coarsening unmatched}

Suppose now that the distribution of~$(D,X)$ is identified from the revealed preference data and the distribution of~$(\tilde P_x,X)$, where~$\tilde P_x$ is the coarsened report, is identified from the stated preference data, but the distribution of~$(D,\tilde P_x,X)$ is not. As in the matched case, the latent report~$P_x$ is observed only through its cell, so that the full law of~$(P_x,X)$ is not identified. Then, the cell-averaged responses~$s_i$ are no longer identified. Instead, the relevant identified quantities are~$p(x):=\mathbb P(D=1\vert X=x)$, ~$\mathbb P(P_x\in B_i\vert X=x)$ and~$\mathbb P(P_x\in B_i,X\neq x)$, hence the cell boundaries~$c_i$, and the values~(\ref{eq:concentration}) of the concentration curve~$L$.

Under assumptions~\ref{ass:monotone m} and~\ref{ass:MLR1}, the sharp bounds have a simple closed form.
\begin{proposition}[Sharp bounds with unmatched data]\label{prop:censored}
    Under assumptions~\ref{ass:endogeneity}--\ref{ass:full support}, \ref{ass:coarsening}, \ref{ass:monotone m} and~\ref{ass:MLR1}, with coarsening by a partition, the sharp bounds on~$\mu(x)$ are
    \begin{eqnarray*}
        \mu^-(x) \; + \; \min_{L\in\mathcal L}\big(L(1)-L(1-p(x))\big) & \leq \; \mu(x) \; \leq & p(x).
    \end{eqnarray*}
\end{proposition}

The upper bound does not involve the stated preference data. Under assumption~\ref{ass:monotone m}, the structural function~$m(x,\cdot)$ is non decreasing, and under assumption~\ref{ass:MLR1}, the density ratio~$h$ is non increasing. Chebyshev's inequality yields~$\mu(x)\leq\mathbb E[m(x,P_x)\vert X=x]\,\mathbb E[h(P_x)\vert X=x]=p(x).$
The lower bound is the same as in proposition~\ref{prop:coarsened MLR}, except that individuals with~$D=1$ are now allocated to the highest reported probabilities, i.e., to the cells~$B_K,B_{K-1},\ldots$ up to probability~$p(x)$.

%

% INFERENCE
\section{Inference}\label{sec:inference}

\subsection{Inference with matched data}
The average structural function is equal to~$\mu(x)=\int \mathbb E[D\vert X=x,\mathbf P=\mathbf p]\;dF_{\mathbf P}$ from proposition~\ref{prop:identification ASF}. Let~$\hat{\mathbb E}[D\vert X=x,\mathbf P=\mathbf p]$ be an estimator of the conditional expectation. Then, the sample analogue 
\begin{eqnarray}
    \hat\mu(x) & := & \frac{1}{n}\sum_{i=1}^{n}\hat{\mathbb E}[D\vert X=x,\mathbf P=\mathbf p_i]
\end{eqnarray}
can be used as an estimator for~$\mu(x)$. Assumptions~\ref{ass:endogeneity}--\ref{ass:full support} do not by themselves deliver regular inference for~$\hat\mu(x)$. The influence function of~$\mu(x)$ is~$1\{X=x\}(D-m(x,\mathbf P))/\mathbb P(X=x\vert\mathbf P)+m(x,\mathbf P)-\mu(x)$, whose variance contains the term~$\mathbb E[m(x,\mathbf P)(1-m(x,\mathbf P))/\mathbb P(X=x\vert\mathbf P)]$; this term is finite under strict overlap, i.e., when the density ratio~$h=\pi_x/\mathbb P(X=x\vert\mathbf P)$ is bounded, but assumption~\ref{ass:full support} allows~$h$ to be unbounded, in which case it may be infinite and~$\sqrt n$-regular inference need not be available. Under strict overlap, together with smoothness conditions on~$m(x,\cdot)$ and on the density of~$\mathbf P$ and bandwidth or tuning conditions controlling the first-stage bias, a kernel or series estimator of the conditional expectation yields a~$\sqrt n$-asymptotically normal~$\hat\mu(x)$ for which the nonparametric bootstrap is valid. We do not state these conditions formally. In the application, we use the double/debiased machine learning estimator of \cite{chernozhukov2018double}, and the reported standard errors rest on the conditions of that paper rather than on a kernel argument.

\subsection{Inference with unmatched data}

% \textcolor{red}{Kernel estimators follow a pointwise CLT, but as a function, $\sqrt{nh_n}(\hat f_n(p)-\mathbb  E\hat f_n(p))$ does not converge to a well defined process (Gaussian or otherwise). Hence we need to restrict our inference to discrete $\mathbf P$. We should propose inference for the censored outcomes only, hence discrete~$\mathbf P$. Maybe another possibility is the use of fixed bandwith kernel estimators, which follow a functional CLT.}

Fix the value~$x$ of interest. To ease notation, we consider inference based on the bounds from~(\ref{eq:dual bounds}) without exclusion restrictions, additional covariates or dependence constraints. 
%The bounds in theorem~\ref{thm:bounds} can be similarly dealt with. Hem...maybe? 
The inference results in this subsection concern report vectors with continuous densities relative to Lebesgue measure on the report cube. They do not cover singular continuous report distributions. When a subvector of reports recovers the relevant type and satisfies the density and inference assumptions, the analysis may instead be conducted using that subvector. Accordingly, in this subsection~$\varrho$ is Lebesgue measure, $d\mathbf p$ denotes Lebesgue integration, and~$\mathbf P$ denotes the report vector or subvector used for inference, with values in~$[0,1]^d$. Assume that the densities~$f_{\mathbf P}$ and~$f_{\mathbf P\vert X=x}$ are continuous functions of~$\mathbf p$ on~$[0,1]^d$. We derive a one-sided confidence bound for the upper bound:
\begin{eqnarray*}
    \mbox{UB}_x & := & \min_{\tiny
    \begin{array}{c}
        \varphi\in[-M,M]
    \end{array}}
    \left\{ \int \max(-\varphi f_{\mathbf P\vert X=x}(\mathbf p),f_{\mathbf P}(\mathbf p))\;d\mathbf p +\varphi\mathbb E[1-D\vert X=x]\right\},
\end{eqnarray*}
for some~$M>0$. By step~5 of the proof of theorem~\ref{thm:sharp bounds}, when~$0<\mathbb E[D\vert X=x]<1$ the dual optimizers are finite quantiles of the density ratio~$h$, and the truncated program has the same value as the sharp bound~(\ref{eq:dual bounds}) whenever~$M\geq Q_{h(\mathbf P)\vert X=x}(\max\{p(x),1-p(x)\})$; otherwise, and in particular when~$\mathbb E[D\vert X=x]\in\{0,1\}$, the truncated program is a valid but possibly conservative outer bound. We treat~$M$ as fixed throughout. The lower bound can be treated symmetrically.

\paragraph{Sampling scheme.} The two ingredients of~$\mbox{UB}_x$ are identified from different data sets. The densities~$f_{\mathbf P\vert X=x}$ and~$f_{\mathbf P}$ are identified from the stated preference sample~$\mathcal S_{\mathbf P}:=\{(\mathbf P_j,X_j):j=1,\ldots,n_{\mathbf P}\}$, and the conditional mean~$e_x:=\mathbb E[1-D\vert X=x]$ is identified from the revealed preference sample~$\mathcal S_D:=\{(D_i,X_i):i=1,\ldots,n_D\}$. The two samples have different sizes, and no individual is observed in both. We write~$\mathcal D_n:=(\mathcal S_{\mathbf P},\mathcal S_D)$ for the data and impose the following sampling conditions.

\begin{assumption}\label{ass:two samples}
(1) $\mathcal S_{\mathbf P}$ is an i.i.d. sample of size~$n_{\mathbf P}$ from the law of~$(\mathbf P,X)$, $\mathcal S_D$ is an i.i.d. sample of size~$n_D$ from the law of~$(D,X)$, and~$\mathcal S_{\mathbf P}\perp\!\!\!\perp\mathcal S_D$;
(2) $n_{\mathbf P}/(n_{\mathbf P}+n_D)\rightarrow c\in(0,1)$ as~$n_{\mathbf P}+n_D\rightarrow\infty$.
\end{assumption}

Part~(1) formalizes the data combination scheme of section~\ref{sec:bounds}: both samples are drawn from the same population, and the joint law of~$(D,\mathbf P)$ never enters the data. Part~(2) requires the two sample sizes to grow at the same rate, so that neither sample is asymptotically negligible relative to the other. All asymptotic statements are indexed by the effective sample size~$n:=n_{\mathbf P}\,n_D/(n_{\mathbf P}+n_D)$.

\paragraph{Estimators and bootstrap.} Define the parameter~$\theta_x$ with
\begin{eqnarray*}
    \theta_x & := & \left( f_{\mathbf P\vert X=x},f_{\mathbf P},e_x\right),
\end{eqnarray*}
as well as a standard choice of estimator~$\hat\theta_x$ with
\begin{eqnarray*}
    \hat\theta_x & := & \left( \hat f_{\mathbf P\vert X=x},\hat f_{\mathbf P},\hat e_x\right), \qquad \hat e_x \; := \; \frac{\sum_{i=1}^{n_D}(1-D_i)\mathbf 1\{X_i=x\}}{\sum_{i=1}^{n_D}\mathbf 1\{X_i=x\}},
\end{eqnarray*}
where~$\hat f_{\mathbf P\vert X=x}$ and~$\hat f_{\mathbf P}$ are estimators of~$f_{\mathbf P\vert X=x}$ and~$f_{\mathbf P}$ computed from~$\mathcal S_{\mathbf P}$ (the former from the subsample with~$X_j=x$), that satisfy part~(1) of assumption~\ref{ass:bootstrap} below with rate~$r_n$. Natural choices for the densities are parametric Beta distributions or finite-order Bernstein polynomials, or distributional regression estimators. Finally, define the bootstrapped estimator~$\theta^\ast_x$ with
\begin{eqnarray*}
    \theta^\ast_x & := & \left( f^\ast_{\mathbf P\vert X=x},f^\ast_{\mathbf P},e^\ast_x\right),
\end{eqnarray*}
where~$f^\ast_{\mathbf P\vert X=x}$ and~$f^\ast_{\mathbf P}$ are computed from a standard nonparametric bootstrap sample~$\mathcal S^\ast_{\mathbf P}$ of size~$n_{\mathbf P}$ drawn from~$\mathcal S_{\mathbf P}$, and~$e^\ast_x$ from a standard nonparametric bootstrap sample~$\mathcal S^\ast_D$ of size~$n_D$ drawn from~$\mathcal S_D$. The two samples are resampled independently of each other: the bootstrap weights generating~$\mathcal S^\ast_{\mathbf P}$ and~$\mathcal S^\ast_D$ are mutually independent and independent of~$\mathcal D_n$. We impose the following conditions on the inputs; $\mathbb B$ is the space defined in the preliminaries.

\begin{assumption}\label{ass:bootstrap} The following two statements hold:
    \begin{enumerate}
        \item $r_n(\hat\theta_x-\theta_x)$ converges weakly in~$\mathbb B$ to a tight Borel measurable random element~$\mathbb G_x$, for some sequence~$r_n\uparrow\infty$ indexed by the effective sample size~$n$;
        \item $\theta^\ast_x$ is a measurable function of~$\mathcal D_n$ and of bootstrap weights independent of~$\mathcal D_n$, and~$r_n(\theta^\ast_x-\hat\theta_x)$ converges conditionally weakly in probability to~$\mathbb G_x$, in the sense that
        \begin{eqnarray*}
            \sup_{g\in\mathrm{BL}_1(\mathbb B)}\Big\vert\,\mathbb E\big[g\big(r_n(\theta^\ast_x-\hat\theta_x)\big)\,\big\vert\,\mathcal D_n\big]-\mathbb E\big[g(\mathbb G_x)\big]\Big\vert & \rightarrow & 0
        \end{eqnarray*}
        in outer probability, where~$\mathrm{BL}_1(\mathbb B)$ is the set of real functions on~$\mathbb B$ bounded by one and Lipschitz with constant one.
    \end{enumerate}
\end{assumption}

Part~(1) is the functional central limit theorem for the inputs. Part~(2) is the input-bootstrap condition of \cite{fang2019inference}. It is not implied by part~(1) and must be verified for the chosen density estimators and resampling scheme. 

We show that the bootstrap procedure of \cite{fang2019inference} is valid in this context, under assumptions~\ref{ass:two samples} and~\ref{ass:bootstrap} on the sampling scheme and the inputs, and a mild additional regularity assumption on the density ratio~$h$. 

\begin{assumption}\label{ass:nullcontact}
    The contact set for the density ratio is Lebesgue-null for all~$x\in\mathcal X$:
\begin{equation}\label{eq:nullcontact}
\forall c\in[-M,M],\; \bigl|\{\,\mathbf p : f_\mathbf P(\mathbf p)=c\,f_{\mathbf P|X=x}(\mathbf p)\,\}\bigr|=0.
\end{equation}
\end{assumption}

\begin{remark}
    Condition~\eqref{eq:nullcontact} holds automatically when~$f_{\mathbf P}$ and~$f_{\mathbf P|X=x}$ are finite-order Bernstein densities: the ratio $f_{\mathbf P}/f_{\mathbf P|X=x}$ is real analytic on~$[0,1]^d$, so each of its level sets is Lebesgue-null, providing the ratio is non-constant. See for instance \cite{mityagin2015zero}.
\end{remark}

\begin{proposition}[Directional differentiability]\label{prop:directional}
Let $M>0$ be fixed and let
\[
\mathbb{B}:=\ell^\infty\!\big([0,1]^d\big)\times \ell^\infty\!\big([0,1]^d\big)\times\mathbb{R},
\qquad
\|\theta\|_{\mathbb B}:=\|\theta_1\|_\infty+\|\theta_2\|_\infty+|\theta_3|,
\]
and define $U:[-M,M]\times\mathbb{B}\to\mathbb{R}$ by
\[
U(\varphi,\theta):=\int_{[0,1]^d}\max\{-\varphi\,\theta_1(\mathbf p),\,\theta_2(\mathbf p)\}\,d\mathbf p+\varphi\,\theta_3 .
\]
With $\theta_x:=\big(f_{\mathbf P|X=x},\,f_{\mathbf P},\,\mathbb{E}[1-D\mid X=x]\big)$, under assumption~\ref{ass:nullcontact}, the value map
\[
\bar\phi:\ \theta\longmapsto UB_x=\min_{\varphi\in[-M,M]}U(\varphi,\theta)
\]
is Hadamard directionally differentiable at $\theta_x$ on $\mathbb{B}$.
\end{proposition}

Following \cite{hong2018numerical}, estimate the directional derivative of~$\bar\phi$ by the numerical derivative
\begin{eqnarray*}
    \widehat{\bar\phi}^\prime_n(h) & := & \frac{1}{\xi_n} \left( \bar\phi(\hat\theta_x+\xi_n h)-\bar\phi(\hat\theta_x)\right),
\end{eqnarray*}
defined for all~$h\in \ell^\infty([0,1]^d)\times \ell^\infty([0,1]^d)\times\mathbb{R}$, with a tuning sequence~$\xi_n$, indexed by the effective sample size, satisfying~$\xi_n\to0$ and~$\xi_n r_n\to\infty$ (for instance~$\xi_n=n^{-1/3}$ when~$r_n=\sqrt n$). Under assumption~\ref{ass:bootstrap} and these conditions, Lemma~S.3.8 in \cite{fang2019inference} shows that Assumption~4 of \cite{fang2019inference} holds. Hence, by their Theorem~3.2,~$\widehat{\bar\phi}^\prime_n(r_n(\theta^\ast_x-\hat\theta_x))$ provides a valid approximation of the distribution of~$r_n(\bar\phi(\hat\theta_x)-\mbox{UB}_x)$. Symmetrically, defining~$\underline\phi:\theta_x\mapsto\mbox{LB}_x$,~$\widehat{\underline\phi}^\prime_n(r_n(\theta^\ast_x-\hat\theta_x))$ approximates the distribution of~$r_n(\underline\phi(\hat\theta_x)-\mbox{LB}_x)$. These yield the asymptotically valid~$(1-\alpha)$-level confidence region
\begin{eqnarray*}
    \left[ \; \underline\phi(\hat\theta_x)-r_n^{-1}\underline q_{1-\alpha/2} \; , \; \bar\phi(\hat\theta_x)-r_n^{-1}\bar q_{\alpha/2} \; \right],
\end{eqnarray*}
where~$\underline q_{1-\alpha/2}$ and~$\bar q_{\alpha/2}$ are quantiles of the distributions of~$\widehat{\underline\phi}^\prime_n(r_n(\theta^\ast_x-\hat\theta_x))$ and~$\widehat{\bar\phi}^\prime_n(r_n(\theta^\ast_x-\hat\theta_x))$ respectively, and~$r_n$ is the rate in the effective sample size of assumption~\ref{ass:bootstrap}. In the parametric case~$r_n=\sqrt n$, the excess length of the region is of order~$(1/n_{\mathbf P}+1/n_D)^{1/2}$.

\section{Simulations}
\label{sec:simulations}

We conduct two simulation experiments with distinct objectives. The first
examines the informativeness of the data combination bounds and the extent
to which they are tightened by additional covariates common to both data
sets, exclusion restrictions, and restrictions on the dependence between
stated and revealed preferences. This first experiment is conducted at the
population level and does not require a least-favourable data generating
process. The second experiment studies the finite-sample performance of the
inference procedure. For this purpose, we use a simpler model without additional covariates, exclusion restrictions, or dependence restrictions. We generate unmatched observations from a distribution compatible with structural models whose average structural functions attain either endpoint of the identified set, and evaluate coverage at those endpoints.

\subsection{Theoretical bounds}

\subsubsection{Population model for the comparison of bounds}
\label{sec:simulation_bounds_dgp}

We begin with a data generating process designed to illustrate the three
tightening devices introduced in Section~\ref{sec:bounds}. The
choice variable $X$, the excluded variable $Z$, and gender $G$ are binary.
Women and men are equally represented: $G\sim\operatorname{Bernoulli}(1/2)$.
Conditional on gender, unobserved heterogeneity is distributed according to
$\eta\mid G=1\sim N(0,0.25),$ and $\eta\mid G=0\sim N(0,1).$
The potential retirement decision under health status $x\in\{0,1\}$ is
\[
    D(x)
    =
    \mathbf 1\{\Phi(x+2\eta)\geq\nu\},
    \qquad
    \nu\sim U[0,1],
\]
where $\nu$ is independent of $(\eta,G)$. 
% It follows that
% \[
%     \Pr\{D(x)=1\mid\eta\}
%     =
%     \Phi(x+2\eta).
% \]
The stated probability of retirement under hypothetical health status $x$
is
\[
    P_x=\Phi(x+\eta).
\]
The stated and revealed choice probabilities therefore differ. However, both are increasing functions of $\eta$.
Health status at the time of the actual decision and the excluded variable
are generated according to
\[
    X=\mathbf 1\{\Phi(\eta)\leq\nu_x\},
    \qquad
    Z=\mathbf 1\{\Phi(\eta)\leq\nu_z\},
\]
where $\nu_x$ and $\nu_z$ are independent $U[0,1]$ random variables,
independent of $(\eta,G,\nu)$. The observed retirement decision is
$D=XD(1)+(1-X)D(0).$
The variable $Z$ affects the distribution of $\eta$ conditional on the
observed data but does not enter the potential retirement decision. Hence,
$Z$ satisfies assumption~\ref{ass:endogeneity'}. Gender $G$ is observed in both
the revealed and stated preference data sets and can therefore be used as an
additional common covariate.

The true average structural function is
\[
    \mu(x)
    =
    \mathbb E[\Phi(x+2\eta)].
\]
Using the identity
\[
    \mathbb E[\Phi(a+bY)]
    =
    \Phi\left(
        \frac{a}{\sqrt{1+b^2\sigma^2}}
    \right),
    \qquad
    Y\sim N(0,\sigma^2),
\]
we obtain
\[
    \mu(x)
    =
    \frac{1}{2}
    \Phi\left(\frac{x}{\sqrt{5}}\right)
    +
    \frac{1}{2}
    \Phi\left(\frac{x}{\sqrt{2}}\right).
\]

\subsubsection*{Positive rank dependence}

Because $P_x=\Phi(x+\eta)$ and
$\Pr\{D(x)=1\mid\eta\}=\Phi(x+2\eta)$ are both increasing in $\eta$,
the stated probability and the actual retirement decision are positively
rank dependent within each $(X,Z,G)$ cell. We impose the moment restriction
\[
    \mathbb E\left[
        \zeta_{xBg}(D,P_x)
        \mid X=x,Z\in B,G=g
    \right]
    \leq -\frac{1}{4}d_{xBg}(1-d_{xBg}),
\]
where
\[
    \zeta_{xBg}(d,p)
    :=
    (d-d_{xBg})
    \left\{
        0.5
        -
        F_{P_x\mid X=x,Z\in B,G=g}(p)
    \right\},
\]
and
\[
    d_{xBg}
    :=
    \mathbb E[D\vert X=x,Z\in B,G=g].
\]
Writing~$U:=F_{P_x\mid X=x,Z\in B,G=g}(P_x)$ for the conditional rank of the report, which is uniform on~$[0,1]$ within the cell, the restriction reads~$\mathrm{Cov}(D,U\mid X=x,Z\in B,G=g)\geq d_{xBg}(1-d_{xBg})/4$. It constrains the conditional covariance between~$D$ and the conditional rank of~$P_x$, not~$P_x$ itself, to be at least half of its comonotonic maximum~$d_{xBg}(1-d_{xBg})/2$. This is a quantitative strengthening of positive rank dependence, which only requires this covariance to be non-negative; the tightening attributed to the dependence restriction in section~\ref{sec:bounds comparison} reflects this stronger restriction.

\subsubsection{Theoretical bounds comparison}
\label{sec:bounds comparison} 

We compare the bounds obtained using only revealed preference data with the data-combination bounds obtained from unmatched stated and revealed preference data. We then examine the additional tightening provided by the dependence and exclusion restrictions. Gender is used as a common covariate in all data-combination specifications.

We report the following intervals.

\begin{enumerate}
    \item \emph{RP-only bounds.}
    These use only the revealed preference distribution of $(D,X)$, without any stated preference information:
    \[
    \mathbb P(D=1,X=x)
    \ \leq\ 
    \mu(x)
    \ \leq\ 
    1-\mathbb P(D=0,X=x).
    \]

    \item \emph{Gender-assisted bounds.}
    We combine the unmatched stated and revealed preference data, condition on $G$, compute the bounds separately by gender, and average the gender-specific bounds. We ignore $Z$ and impose no dependence restriction.

    \item \emph{Dependence-restricted bounds with gender.}
    In addition to conditioning on gender, we impose positive rank dependence between $D$ and $P_x$ conditional on $X=x$ and $G=g$, but ignore $Z$.

    \item \emph{Exclusion-restriction bounds with gender.}
    We condition on gender, impose the exclusion restriction, and then average the resulting bounds over gender. We impose no dependence restriction.

    \item \emph{Combined bounds.}
    We jointly use gender, the exclusion restriction, and positive rank dependence.
\end{enumerate}

All quantities are population quantities. We evaluate the required densities, conditional means, and distribution functions by numerical integration and solve the scalar optimization problems over $\varphi$ and $\lambda$ numerically. Figure~\ref{fig:theoretical_bounds} compares the resulting intervals with the true value of $\mu(x)$.

\begin{figure}[h!]
    \centering
    \includegraphics[width=0.82\linewidth]{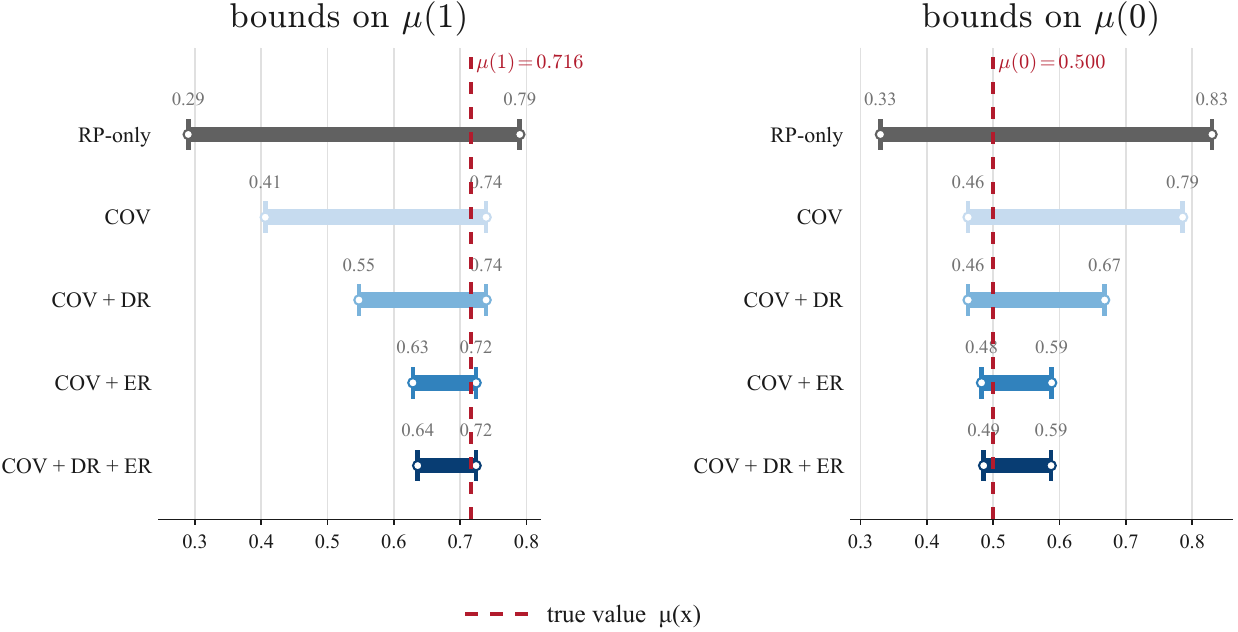}
    \caption{Population bounds on the average structural function. The figure reports the bounds using revealed preference data only (RP-only), the gender-assisted data-combination bounds (COV), the bounds additionally imposing the dependence restriction (COV + DR) or the exclusion restriction (COV + ER), and the bounds combining all three sources of information (COV + DR + ER). The vertical marker gives the true value of $\mu(x)$.}
    \label{fig:theoretical_bounds}
\end{figure}

Figure~\ref{fig:theoretical_bounds} illustrates the identifying value of stated preference data. The bounds using only revealed preference data are $[0.29,0.79]$ for $\mu(1)$ and $[0.33,0.83]$ for $\mu(0)$, each with width $0.50$. Combining stated and revealed preference data using gender as a common covariate narrows the intervals to $[0.41,0.74]$ and $[0.46,0.79]$, respectively. Thus, even without imposing the dependence or exclusion restrictions, the addition of stated preference information substantially narrows the identified interval relative to using revealed preference data alone.

Figure~\ref{fig:theoretical_bounds} also shows that the additional restrictions differ substantially in their informativeness. In this design, the exclusion restriction is more informative than the dependence restriction when each is added separately to the gender-assisted bounds. The restrictions are complementary: using gender together with both restrictions yields bounds of $[0.64,0.72]$ for $\mu(1)$ and $[0.49,0.59]$ for $\mu(0)$, with widths of $0.08$ and $0.10$, respectively, compared with $0.50$ for the RP-only bounds. The experiment therefore illustrates that combining stated and revealed preference data with economically credible auxiliary restrictions can make the bounds highly informative, particularly when the restrictions are imposed jointly.

\subsection{Inference simulations}

\subsubsection{Data generating process and sampling}
\label{sec:simulation_inference_dgp}
\label{sec:simulation_samples}

The second simulation experiment evaluates the finite-sample performance of
the inference procedure developed in Section~\ref{sec:inference}.
For this purpose, we remove gender, the excluded variable~$Z$, and the
dependence constraint. We focus on inference for~$\mu(0)$ under the baseline
bounds~(\ref{eq:dual bounds}).

Let~$\eta\sim N(0,1),$ and~$\nu_x,\nu\stackrel{\mathrm{iid}}{\sim}U[0,1],$ with all
variables mutually independent. Define
\[
    P_x=\Phi(x+\eta),
    \qquad
    X=\mathbf 1\{\Phi(\eta)\leq\nu_x\},
    \qquad
    D(x)=\mathbf 1\{\Phi(x+2\eta)\geq\nu\},
\]
and let the observed decision be~$D=XD(1)+(1-X)D(0).$

Because revealed and stated preferences are collected in separate samples,
the joint distribution of~$(D,\mathbf P)$ given~$X$ never enters the data. The
observable content of the design is exhausted by the two marginals~$F_{\mathbf P\vert X=x}$
and~$\mathbb E[D\vert X=x]$, and any data generating process reproducing those
marginals generates the same observations. The coupling specified above is
therefore irrelevant for the sampling experiment.

For every Monte Carlo repetition, we generate two mutually independent
samples of equal size~$n$.
The revealed preference sample is~$\mathcal S_R=\{(D_i^R,X_i^R):i=1,\ldots,n\}.$
For each~$i$, we independently draw~$\eta_i^R,$ $\nu_{xi}^R,$ $\nu_i^R$ and
construct
\[
    X_i^R
    =
    \mathbf 1\{\Phi(\eta_i^R)\leq\nu_{xi}^R\},
    \qquad
    D_i^R
    =
    \mathbf 1\{\Phi(X_i^R+2\eta_i^R)\geq\nu_i^R\}.
\]
Only~$(D_i^R,X_i^R)$ is retained; the stated probabilities of the same
individual are latent and discarded.
Independently, the stated preference sample is~$\mathcal S_S=\{(P_{0j}^S,P_{1j}^S,X_j^S):j=1,\ldots,n\}.$
For each~$j$, we independently draw~$\eta_j^S,$ $\nu_{xj}^S$ and construct
\[
    X_j^S
    =
    \mathbf 1\{\Phi(\eta_j^S)\leq\nu_{xj}^S\},
    \qquad
    P_{0j}^S=\Phi(\eta_j^S),
    \quad
    P_{1j}^S=\Phi(1+\eta_j^S).
\]
Only~$(P_{0j}^S,P_{1j}^S,X_j^S)$ is retained; no decision is generated in this
sample. The random variables used to generate the two samples are mutually
independent, and hence~$\mathcal S_R\indep\mathcal S_S.$ The retained revealed
preference sample is a collection of Bernoulli draws with success
probability~$q_{X_i^R}$, and the pair~$(D,\mathbf P)$ is never formed in either
sample.

The averaged structural function~$\mu(0)$ is known only
to lie in~$[LB_0,UB_0]$. Coverage is
accordingly evaluated at the two endpoints, which are the least favorable
values for the confidence region.

\subsubsection*{Derivation of the bounds in closed form} We now derive~$LB_0$ and~$UB_0$ in closed form. Since~$\Pr(X=0\mid\eta)=\Phi(\eta)$
and~$\mathbb E[\Phi(\eta)]=1/2,$ and since~$P_0=\Phi(\eta)$ is uniform on~$[0,1]$,
\[
    f_{P_0}(p)=1,
    \qquad
    f_{P_0\mid X=0}(p)=2p,
    \qquad
    h_0(p)=\frac{f_{P_0}(p)}{f_{P_0\mid X=0}(p)}=\frac{1}{2p}.
\]
The conditional choice probability is~$q_0:=\Pr(D=1\mid X=0)=2\,\mathbb E[\Phi(2\eta)\Phi(\eta)].$
For~$U_1,U_2$ standard normal and independent of~$\eta$, the identity~$\mathbb E[\Phi(a\eta)\Phi(b\eta)]=\Pr(U_1-a\eta\leq0,\,U_2-b\eta\leq0)$
expresses the expectation as a bivariate normal orthant probability with
correlation~$\rho=ab/\sqrt{(1+a^2)(1+b^2)},$ so that~$\mathbb E[\Phi(a\eta)\Phi(b\eta)]=1/4+\arcsin(\rho)/(2\pi).$
Taking~$a=2$ and~$b=1$ gives~$\rho=2/\sqrt{10}$ and
\[
    q_0
    =
    \frac12
    +
    \frac1\pi\arcsin\frac{2}{\sqrt{10}}
    \;=\;
    0.717953.
\]
Substituting~$f_{P_0}$ and~$f_{P_0\mid X=0}$ in the upper bound
of~(\ref{eq:dual bounds}) and writing~$\varphi=-\psi$ yields the objective
\[
    \int_0^1\max(2\psi p,1)\,dp-\psi(1-q_0)
    =
    \psi q_0+\frac{1}{4\psi},
    \qquad \psi\geq\tfrac12,
\]
minimized at~$\psi^\ast=1/(2\sqrt{q_0})$ with value~$\sqrt{q_0}$; the objective
is decreasing in~$\psi$ on~$[0,1/2)$, so the minimizer is the interior
point~$\psi^\ast$. Symmetrically, the lower bound objective is~$1-1/(4\varphi)-\varphi(1-q_0)$
for~$\varphi\geq1/2$, maximized at~$\varphi^\ast=1/(2\sqrt{1-q_0})$ with
value~$1-\sqrt{1-q_0}$. Hence
\[
    LB_0=1-\sqrt{1-q_0}=0.468919,
    \qquad
    UB_0=\sqrt{q_0}=0.847321,
\]
and the identified set has length~$0.378402$. 

\subsubsection{Inference procedure}
\label{sec:simulation inference}

For a fixed $x$, define
\[
    \theta_x
    :=
    \left(
        f_{P_x\mid X=x},
        f_{P_x},
        e_x
    \right),
    \qquad
    e_x:=\mathbb E[1-D\mid X=x].
\]
For $\theta=(\theta_1,\theta_2,\theta_3)$, define the lower- and upper-bound
functionals
\[
    \mathcal L(\theta)
    :=
    \sup_{\varphi\in[-M,M]}
    \left\{
        \int_0^1
        \min\{
            \varphi\theta_1(p),
            \theta_2(p)
        \}\,dp
        -
        \varphi\theta_3
    \right\}
\]
and
\[
    \mathcal U(\theta)
    :=
    \inf_{\varphi\in[-M,M]}
    \left\{
        \int_0^1
        \max\{
            -\varphi\theta_1(p),
            \theta_2(p)
        \}\,dp
        +
        \varphi\theta_3
    \right\}.
\]
Thus,
\[
    LB_x=\mathcal L(\theta_x),
    \qquad
    UB_x=\mathcal U(\theta_x).
\]
We work with the scalar report~$P_0=\Phi(\eta)$ for~$\mu(0)$, rather than fitting a bivariate density to~$(P_0,P_1)=(\Phi(\eta),\Phi(1+\eta))$: the pair is supported on a curve in~$[0,1]^2$ and has no density with respect to Lebesgue measure, whereas~$P_0$ recovers~$\eta$ and has well-defined densities~$f_{P_0}$ and~$f_{P_0\mid X=0}$, as required in section~\ref{sec:inference}.

We estimate the two densities using finite-order Bernstein polynomial
densities. Let
\[
    b_{k,r}(p)
    :=
    (r+1)\binom{r}{k}p^k(1-p)^{r-k},
    \qquad k=0,\ldots,r.
\]
A Bernstein density of order $r$ takes the form
\[
    f(p;\omega)
    =
    \sum_{k=0}^r\omega_k b_{k,r}(p),
    \qquad
    \omega_k\geq0,
    \qquad
    \sum_{k=0}^r\omega_k=1.
\]
We estimate the weights by maximum likelihood. Let~$\widehat f_{P_x}$
be the Bernstein density estimator based on the stated preference sample,
and let~$\widehat f_{P_x\mid X=x}$
be the corresponding estimator based on the subsample with $X_j^S=x$.
We consider polynomial orders $r=5$ and $r=10$.
For the inference experiment at $x=0$, the true densities~$f_{P_0}(p)=1,$ and~$f_{P_0\mid X=0}(p)=2p$ belong to the Bernstein family for every order $r\geq1$.
The conditional outcome mean is estimated from the revealed preference
sample:
\[
    \widehat e_x
    :=
    \frac{
        \sum_{i=1}^n
        (1-D_i^R)\mathbf 1\{X_i^R=x\}
    }{
        \sum_{i=1}^n
        \mathbf 1\{X_i^R=x\}
    }.
\]
Define
\[
    \widehat\theta_x
    :=
    \left(
        \widehat f_{P_x\mid X=x},
        \widehat f_{P_x},
        \widehat e_x
    \right),
\quad \mbox{ and } \quad    \widehat{LB}_x
    :=
    \mathcal L(\widehat\theta_x),
    \quad
    \widehat{UB}_x
    :=
    \mathcal U(\widehat\theta_x).
\]
For each bootstrap repetition $b=1,\ldots,B$, we independently resample the
two data sets and produce~$\mathcal S_R^{*,b}$ from~$\mathcal S_R,$ and~$\mathcal S_S^{*,b}$ from~$\mathcal S_S.$
We recompute all estimators from the two bootstrap samples and denote the
resulting parameter estimator by~$\widehat\theta_x^{*,b}.$ Let $\xi_n$ satisfy~$\xi_n\rightarrow0,$ and~$\sqrt n\,\xi_n\rightarrow\infty.$
The numerical directional derivatives are
\[
    \widehat T_{L,x}^{\,b}
    :=
    \frac{
        \mathcal L\left[
            \widehat\theta_x
            +
            \xi_n\sqrt n
            \left(
                \widehat\theta_x^{*,b}
                -
                \widehat\theta_x
            \right)
        \right]
        -
        \mathcal L(\widehat\theta_x)
    }{
        \xi_n
    }
\]
and
\[
    \widehat T_{U,x}^{\,b}
    :=
    \frac{
        \mathcal U\left[
            \widehat\theta_x
            +
            \xi_n\sqrt n
            \left(
                \widehat\theta_x^{*,b}
                -
                \widehat\theta_x
            \right)
        \right]
        -
        \mathcal U(\widehat\theta_x)
    }{
        \xi_n
    }.
\]
The bootstrap distributions of
$\widehat T_{L,x}^{\,b}$ and $\widehat T_{U,x}^{\,b}$
approximate the distributions of
\[
    \sqrt n(\widehat{LB}_x-LB_x)
    \quad\text{and}\quad
    \sqrt n(\widehat{UB}_x-UB_x),
\]
respectively.
Let~$\widehat c_{L,x}(\tau)$ and~$\widehat c_{U,x}(\tau)$
denote the empirical $\tau$-quantiles of
$\{\widehat T_{L,x}^{\,b}\}_{b=1}^B$ and
$\{\widehat T_{U,x}^{\,b}\}_{b=1}^B$, respectively. A
$1-\alpha$ confidence region for $\mu(x)$ is
\[
    CR_{1-\alpha}(x)
    =
    \left[
        \widehat{LB}_x
        -
        \frac{
            \widehat c_{L,x}(1-\alpha/2)
        }{
            \sqrt n
        },
        \quad
        \widehat{UB}_x
        -
        \frac{
            \widehat c_{U,x}(\alpha/2)
        }{
            \sqrt n
        }
    \right].
\]

\subsubsection{Simulation results}
\label{sec:simulation results}

We repeat the procedure over $1,000$ Monte Carlo repetitions, each with~$B=1,000$ bootstrap draws. We consider sample sizes~$n\in\{500,1000,2000\},$
Bernstein polynomial orders~$r\in\{5,10\},$
and tuning parameters
\[
    \xi_n
    \in
    \left\{
        0.5n^{-1/3},
        0.75n^{-1/3},
        n^{-1/3},
        1.5n^{-1/3}
    \right\}.
\]

For each specification, we report:

\begin{enumerate}
    \item the probability that the $95\%$ confidence region covers
    the theoretical lower bound~$LB_0$;

    \item the probability that the $95\%$ confidence region covers
    the theoretical upper bound~$UB_0$;

    \item the average excess length of the confidence region, defined as
    \[
        \operatorname{Length}\{CR_{0.95}(0)\}
        -
        (UB_0-LB_0).
    \]
\end{enumerate}

\begin{table}[t]
  \centering
  \small
  \setlength{\tabcolsep}{4pt}
  \begin{tabular}{ll*{8}{c}}
    \toprule
    & & \multicolumn{2}{c}{$\xi_n=\tfrac12 n^{-1/3}$}
      & \multicolumn{2}{c}{$\tfrac34 n^{-1/3}$}
      & \multicolumn{2}{c}{$n^{-1/3}$}
      & \multicolumn{2}{c}{$\tfrac32 n^{-1/3}$}\\
    \cmidrule(lr){3-4}\cmidrule(lr){5-6}
    \cmidrule(lr){7-8}\cmidrule(lr){9-10}
    & $n$ & $r{=}5$ & $r{=}10$
          & $r{=}5$ & $r{=}10$
          & $r{=}5$ & $r{=}10$
          & $r{=}5$ & $r{=}10$\\
    \midrule
    \multirow{3}{*}{Coverage of LB$_0$}
     & 500  & 0.978 & 0.978 & 0.980 & 0.976 & 0.981 & 0.978 & 0.985 & 0.976\\
     & 1000 & 0.974 & 0.972 & 0.975 & 0.975 & 0.980 & 0.973 & 0.983 & 0.974\\
     & 2000 & 0.974 & 0.980 & 0.975 & 0.982 & 0.975 & 0.981 & 0.981 & 0.982\\
    \midrule
    \multirow{3}{*}{Coverage of UB$_0$}
     & 500  & 0.990 & 0.987 & 0.991 & 0.987 & 0.991 & 0.988 & 0.992 & 0.988\\
     & 1000 & 0.979 & 0.976 & 0.980 & 0.978 & 0.980 & 0.978 & 0.982 & 0.976\\
     & 2000 & 0.986 & 0.985 & 0.986 & 0.985 & 0.986 & 0.988 & 0.988 & 0.987\\
    \midrule
    \multirow{3}{*}{Excess length}
     & 500  & 0.103 & 0.104 & 0.106 & 0.105 & 0.108 & 0.106 & 0.111 & 0.106\\
     & 1000 & 0.073 & 0.073 & 0.074 & 0.075 & 0.075 & 0.075 & 0.077 & 0.075\\
     & 2000 & 0.052 & 0.053 & 0.053 & 0.053 & 0.053 & 0.054 & 0.055 & 0.054\\
    \bottomrule
  \end{tabular}
  \caption{Finite-sample performance of the confidence region for $\mu(0)$.
  ``Coverage of LB$_0$'' (resp. ``UB$_0$'') is the empirical coverage rate of each of the end-points of the identified set; ``Excess length'' is the region
  length minus the identified-set length $UB_0-LB_0=0.3784$. Sample size in each of the two independent data sets is $n$, $r$
  is the Bernstein polynomial order, and $\xi_n$ is the numerical-derivative
  tuning parameter. Based on 1,000 Monte Carlo and 1,000 bootstrap repetitions.}
  \label{tab:inference}
\end{table}
Table~\ref{tab:inference} reports the results. 
% Simulation performed with mc_core.py and mc_inference.py
Table~\ref{tab:inference} shows that the proposed confidence region performs well. Empirical coverage of the endpoints of the identified set is uniformly at or above the nominal $95\%$ level at both the lower and upper endpoints. Coverage is in fact close or slightly above~$97.5\%$ which one expects, given that coverage is effectively one sided when the identified set is not reduced to a singleton. Coverage is stable across the two Bernstein polynomial orders and across the numerical-derivative tuning parameters considered. The excess length of the confidence region decreases substantially with sample size, from approximately $0.10$--$0.12$ when $n=500$ to $0.05$--$0.06$ when $n=2{,}000$, while its sensitivity to the polynomial order and tuning parameter is limited. These results suggest that the bootstrap procedure provides reliable coverage and yields increasingly informative confidence regions as the sample size grows.

%

% APPLICATION
\section{Application to retirement decisions}
\label{sec:application}

We apply our methodology to the analysis of retirement decisions and revisit the work and findings of \cite{giustinelli2024}. The dataset comes from the publicly available Health and Retirement Study, a longitudinal panel study that surveys a representative sample of approximately 20,000 people in the United States of America every two years. Within the 2016 wave, an experimental module was administered to a random subsample of 1,790 individuals under the age of 65. This experimental module elicited their preferences for staying in or leaving the labor market. The survey records a subjective rating of their health by the respondent: `excellent',`very good', `good', `fair', or `poor'. It also records their subjective beliefs about the chance of being in one of the health categories and the chance of working for pay in the next two years. Most importantly for our methodology, the survey also elicits the chance of working in the next two years for different health scenarios. The exact question states: 
\begin{quote}
    If your health is [health status] two years from now, what are the chances that you will be working for pay?
\end{quote}
The question is repeated three times for different health categories: `excellent/very good or excellent', `very good or good/good/fair or good', and `fair or poor/poor'. Which health status is presented within a health category depends on the initial health of the respondent.\footnote{Details of the experimental modules are available at \url{https://hrs.isr.umich.edu/sites/default/files/meta/2016/core/qnaire/online/5_Module.pdf}, last accessed on July 18, 2025. The module also contains the same stated preference questions about a four-year horizon.} The stated and revealed health categories are closely, although not perfectly, aligned. We use the response to the second hypothetical health scenario to construct $P_0$. For respondents whose health in 2016 is very good, good, or fair (83\% of the sample) this scenario asks about working under good health. For respondents initially in excellent health, the corresponding scenario is ``very good or good.'' This distinction is unlikely to be consequential: the responses to the first two, more favorable, health scenarios are highly correlated (0.91) and have very similar means. The main overlap concerns respondents initially in poor health, who account for 7\% of the sample and are asked about ``fair or good'' health in the second scenario. We use the response to the unfavorable health scenario to construct $P_1$, corresponding to fair or poor health. Therefore, we interpret $(P_0,P_1)$ as closely approximating the same
good-versus-poor-health contrast as $X=0$ versus $X=1$ in the RP analysis.\footnote{As a check on this harmonization, our main robustness exercise restricts the sample to respondents in good health in 2016, thereby removing the respondents for whom the overlap in the stated health categories is potentially most relevant. The results are very similar to those in the full sample.}

The 2018 wave collects information about the realized health status and labor market status in the same categories as in 2016. This allows us to match stated and revealed preferences for the respondents who completed the stated preference module. The remaining sample is unmatched.

\subsection{Descriptive Analysis}\label{subsec:descriptive}
% Because of our interest in the question of retirement, we restrict the sample to those aged between 50 and 65 who are still participating in the labor force in 2016; that is, either `In work now', `unemployed and looking for a job', or `temporarily laid off'. This leaves us with a sample of 6,982 respondents; among them, 1,090 (16\%) are matched to their stated preferences. Table \ref{tab:descriptive_sample} provides descriptive statistics about the sample. In 2016, the average age was 57, the sample was fairly balanced between men and women, and 7 out of 10 respondents were in a couple. Close to 20\% of the population rated their health as fair or poor, and 9\% of the respondents were out of work in 2016. The 2018 wave suffers from attrition, with about 19\% of health and labor market statuses missing. In the non-attrited sample, one out of five rates their health as `fair or poor'. The inactivity rate has increased to 13\%, and 18\% of the population is now out of work in 2018.

\begin{table}[]
\caption{Descriptive statistics of key variables}
\centering
\begin{tabular}[t]{lrrrrrr}
\toprule
  & Missing (\%) & Mean & SD & Min & Med. & Max\\
\midrule
\multicolumn{7}{l}{\textbf{Descriptive in 2016}}\\
Age &  0 & \num{56.95} & \num{4.06} & \num{50} & \num{57} & \num{65}\\
Gender (1 if `Female')& 0 & \num{0.53} & \num{0.50} & \num{0} & \num{1} & \num{1}\\
In a couple (1 if `Yes')& 0 & \num{0.69} & \num{0.46} & \num{0} & \num{1} & \num{1}\\
`fair or poor  health' in 2016 & 0 & \num{0.20} & \num{0.40} & \num{0} & \num{0} & \num{1}\\
In work in 2016 & 0 & \num{0.91} & \num{0.28} & \num{0} & \num{1} & \num{1}\\ \hline
\multicolumn{7}{l}{\textbf{Descriptive in 2018}}\\
`fair or poor health' in 2018&  19 & \num{0.19} & \num{0.39} & \num{0} & \num{0} & \num{1}\\
Active in 2018 & 19 & \num{0.87} & \num{0.34} & \num{0} & \num{1} & \num{1}\\
Employed in 2018 & 19 & \num{0.82} & \num{0.39} & \num{0} & \num{1} & \num{1}\\ \hline
% work\_likely\_62 &  20 & \num{57.83} & \num{37.20} & \num{0} & \num{60} & \num{996}\\
% work\_likely\_after65 & 16 & \num{45.62} & \num{34.06} & \num{0} & \num{50} & \num{100}\\ \hline
\multicolumn{7}{l}{\textbf{Stated Preferences: Chance of working in 2018 if health is...}} \\
Completed SP module & 0 & \num{0.16} & \num{0.36} & \num{0} & \num{0} & \num{1}\\ 
\multicolumn{7}{l}{\textit{Chance of working in 2018 if health is...}} \\
`Very good' to `Excellent'  & 84 & \num{86.10} & \num{24.77} & \num{0} & \num{100} & \num{100}\\
`Fair or good' to `Very good'& 84 & \num{85.68} & \num{24.60} & \num{0} & \num{100} & \num{100}\\
`Poor' to `Fair or poor'  & 85 & \num{61.05} & \num{32.30} & \num{0} & \num{60} & \num{100}\\
\bottomrule
\end{tabular}
\label{tab:descriptive_sample}
\end{table}

Because our interest is in retirement, we restrict attention to respondents aged 50 to 65 who were still participating in the labor force in 2016: they were either working, unemployed and looking for work, or temporarily laid off. This leaves 6,982 respondents, of whom 1,090 (16\%) completed the stated-preference module and can therefore be matched to their elicited probabilities. Table~\ref{tab:descriptive_sample} summarizes the resulting sample.

At baseline, respondents are on average 57 years old, 53\% are women, and 69\% live with a partner. Although all respondents are initially attached to the labor force by construction, there is substantial variation in both health and employment: 20\% report fair or poor health and 9\% are not currently working. About 19.5\% of respondents observed at baseline are missing either health or employment information in 2018.\footnote{Attrition is not significantly related to baseline age, partnership status, health, or employment, although women are 3.6 percentage points less likely to attrit. Importantly for our analysis, among respondents completing the SP module, attrition is not significantly related to either $P_0$ or $P_1$, individually or jointly ($p=0.58$), conditional on baseline characteristics. Thus, while attrition is not fully balanced on observables, we find no evidence that it is systematically related to the elicited heterogeneity used in our analysis.} Among respondents observed in 2018, roughly one in five reports fair or poor health, 18\% are not working, and 13\% have left the labor force.

The stated-preference questions suggest that respondents expect health to matter considerably for these labor-market transitions. When asked about 2018 under hypothetical health scenarios, respondents report an average probability of working of about 86\% under the two more favorable health scenarios. The two answers are remarkably similar (their correlation is 0.91) suggesting that respondents perceive relatively little difference among health states ranging from good to excellent. Their predictions change sharply, however, when asked to contemplate poor health: the average stated probability of working falls to 61\%. Thus, taken at face value, the stated-preference experiment implies that a deterioration to fair or poor health lowers the probability of working by about 24 percentage points.

%This large stated effect is the natural benchmark for our empirical analysis. However, interpreting its revealed-preference counterpart causally is difficult because realized health is not randomly assigned. Table~\ref{tab:balance_health} in appendix \ref{app:additional_results} makes this selection particularly transparent. Among respondents observed in 2018, 19\% report fair or poor health. Their demographic characteristics are not very different from those in better health: the two groups have similar age and gender distributions, although respondents in poor health are somewhat less likely to live with a partner. Their prior health, by contrast, is strikingly different. Among those who report fair or poor health in 2018, 66\% were already in fair or poor health in 2016, compared with only 9\% among those in good or better health in 2018. Poor realized health is therefore strongly predictable from characteristics determined before the labor-market decision we study.

Differences in labor-market attachment are already visible before the health status of interest is realized. Among respondents who will report fair or poor health in 2018, 15\% are already out of work in 2016, compared with 7\% among those who will remain in better health. By 2018, these shares rise to 31\% and 15\%, respectively. Poor health is therefore associated with substantially lower employment, but the baseline difference cautions against interpreting this association as causal: individuals who experience poor health in 2018 were already less attached to the labor market two years earlier.

Crucially for our approach, some of these differences are also visible in respondents' stated preferences before the 2018 outcomes are realized. Those who subsequently experience fair or poor health report lower probabilities of working in every hypothetical health scenario. Their stated probability of working is about 5-7 percentage points lower than for respondents who subsequently remain in good health. Hence, the elicited probabilities contain information not only about how respondents say they would react to a health shock but also about heterogeneity in their underlying propensity to work. %This is precisely the distinction exploited by our method: we do not require the stated probabilities to provide unbiased forecasts of actual behavior; instead, we use their cross-sectional variation to recover information about the heterogeneity that jointly shapes health and retirement decisions.

\subsection{Matched Samples estimation} \label{subsec:matched_estimation}
We begin with the case where stated and revealed preferences can be matched at the
individual level. This allows us to ask whether controlling for the heterogeneity
revealed by the stated probabilities materially changes the estimated effect of
health on employment. Throughout this section, the outcome $D$ is an indicator for
being out of work in 2018, so that $\mu(x)=\mathbb{E}[D(x)]$ is the probability of
being out of work if health were set to $x$, with $x=0$ denoting good health and
$x=1$ fair or poor health. The average treatment effect $\mu(1)-\mu(0)$ measures the increase in the probability of being out of work associated with poor health.\footnote{Because health is self-reported and measured contemporaneously with employment, the causal interpretation also requires that employment shocks do not themselves affect reported health status, conditional on the heterogeneity captured by our controls. We cannot rule out such outcome-dependent reporting in these data, and interpret the estimates subject to this qualification.}

\begin{table}[!htbp]
\centering
\setlength{\tabcolsep}{12pt}
\begin{tabular}{l@{\hspace{2em}}ccc}
\toprule
 & RP only & SP only & SP4RP \\
\midrule

\multicolumn{4}{l}{\textbf{Panel A: Whole Sample}} \\
$N$ & 5,623 & 880 & 879 \\
$\mu(1)$ & 0.310 (0.014) & 0.385 (0.011) & 0.313 (0.040) \\
$\mu(0)$ & 0.152 (0.005) & 0.143 (0.008) & 0.176 (0.014) \\
ATE & 0.158$\sym{***}$ (0.015) & 0.242$\sym{***}$ (0.009) & 0.138$\sym{***}$ (0.042) \\

\midrule
\multicolumn{4}{l}{\textbf{Panel B: Age 50 to 59}} \\
$N$ & 3,979 & 637 & 637 \\
$\mu(1)$ & 0.291 (0.016) & 0.327 (0.012) & 0.283 (0.048) \\
$\mu(0)$ & 0.109 (0.006) & 0.088 (0.007) & 0.120 (0.015) \\
ATE & 0.182$\sym{***}$ (0.017) & 0.240$\sym{***}$ (0.010) & 0.163$\sym{***}$ (0.051) \\

\midrule
\multicolumn{4}{l}{\textbf{Panel C: Age 60 to 65}} \\
$N$ & 1,644 & 243 & 242 \\
$\mu(1)$ & 0.362 (0.028) & 0.536 (0.021) & 0.427 (0.087) \\
$\mu(0)$ & 0.254 (0.012) & 0.287 (0.021) & 0.311 (0.033) \\
ATE & 0.108$\sym{***}$ (0.031) & 0.249$\sym{***}$ (0.017) & 0.116 (0.092) \\

\bottomrule
\end{tabular}
\caption{Estimated structural functions $\mu(x) = \mathbb{E}[D(x)]$ and average treatment effects $\text{ATE} = \mu(1) - \mu(0)$, by age group, with the outcome $D = 1$ if not working in 2018. RP only is the doubly robust estimator based solely on the revealed preference data, controlling for gender $W$ but without correcting for endogeneity. SP only is the sample average of the stated probabilities $P_1$ and $P_0$ in the stated preference data. SP4RP is our control function estimator applied to the matched RP and SP samples, which additionally conditions on $(P_1,P_0)$ and is computed by double machine learning \citep{chernozhukov2018double}. Standard errors are in parentheses. Significance stars refer only to the test $H_0:\text{ATE}=0$. $\sym{*}$ $p<0.10$, $\sym{**}$ $p<0.05$,
$\sym{***}$ $p<0.01$.}
\label{tab:matched_work}
\end{table}

Table~\ref{tab:matched_work} compares three approaches. The first, denoted RP,
uses revealed choices and controls for observed characteristics but does not use
the stated probabilities to address selection on unobservables. The second, SP,
takes the elicited probabilities at face value: $\mu(0)$ and $\mu(1)$ are estimated
by the average stated probabilities of being out of work under good and poor
health, respectively. The third, SP4RP, is our preferred matched-data estimator.
It estimates actual choices while using the vector of stated probabilities
$(P_0,P_1)$ as a control for unobserved heterogeneity. We implement SP4RP using
double machine learning with random forests and cross-fitting.

Panel A reports the results for the full sample. The RP estimates imply that the
probability of being out of work would increase from 15.2\% under good health to
31.0\% under poor health, an ATE of 15.8 percentage points. Respondents anticipate
a considerably larger response in the stated-preference experiment: the
corresponding probabilities are 14.3\% and 38.5\%, implying an ATE of 24.2
percentage points. By contrast, SP4RP produces an ATE of 13.8 percentage points,
close to the RP estimate. Its two average structural functions are also similar
to their RP counterparts: SP4RP estimates probabilities of being out of work of
17.6\% under good health and 31.3\% under poor health.

This comparison is informative about the respective roles of endogeneity and hypothetical bias. %If selection on the unobserved characteristics captured by thestated probabilities were quantitatively important, we would expect the SP4RP estimates to differ substantially from RP. Instead, controlling for this heterogeneity changes the estimated ATE only modestly, from 15.8 to 13.8 percentage points. The much larger difference is between stated and revealed behavior.
Under our maintained assumption that the pair $(P_0,P_1)$ captures the selection-relevant heterogeneity, the comparison points to hypothetical bias as the quantitatively more important concern in this application. Conditioning on the stated probabilities changes the RP estimate only modestly, from 15.8 to 13.8 percentage points, whereas taking those probabilities literally implies a substantially larger effect. This interpretation is necessarily conditional on the control function assumption: the similarity of RP and SP4RP estimates does not itself validate that assumption.

The age-specific estimates reveal where this discrepancy is concentrated.
Among respondents aged 50-59 (Panel B), the three approaches give a relatively
similar assessment of the probability of being out of work under poor health:
29.1\% for RP, 32.7\% for SP, and 28.3\% for SP4RP. The resulting ATE is 18.2
percentage points for RP and 16.3 percentage points for SP4RP, compared with
24.0 percentage points for SP. Thus, for younger respondents, stated preferences
somewhat exaggerate the effect of poor health, but all three approaches point to
a substantial increase in the probability of leaving employment.

The contrast is considerably sharper among respondents aged 60-65 (Panel C).
RP estimates that poor health increases the probability of being out of work from
25.4\% to 36.2\%, an ATE of 10.8 percentage points. Stated preferences imply a
much larger effect of 24.9 percentage points. This difference comes primarily
from the poor-health counterfactual: respondents report a 53.6\% probability of
being out of work under poor health, compared with 36.2\% in the RP estimates.
SP4RP lies much closer to RP, with estimated probabilities of 31.1\% under good
health and 42.7\% under poor health, and an ATE of 11.6 percentage points.
The latter estimate is imprecise, however, reflecting the small number of older
respondents for whom stated and revealed preferences can be matched.

Taken together, the age profiles suggest that the gap between hypothetical and
realized responses is particularly pronounced close to retirement age.
Respondents aged 60-65 anticipate a large employment response to poor health,
whereas their realized choices imply a substantially smaller effect. At the same
time, using their stated probabilities to control for unobserved heterogeneity
does not produce evidence of a comparably large correction to the RP estimates.

\subsection{Unmatched samples estimation}
\label{subsec:unmatched_estimation}

We now turn to the empirically more common setting in which stated and revealed preferences cannot be linked at the individual level. We proceed as if the SP and RP observations came from separate samples. Our dataset helps to gauge the cost of not being able to link stated and revealed preferences. Since our inference theory only covers the baseline bounds, we treat results under additional covariates, dependence and exclusion restrictions as exploratory. %We therefore observe the distribution of $(P_0,P_1,X)$ in the SP sample and that of $(D,X)$ in the RP sample, but not their joint distribution. In contrast to the matched case, the average structural functions $\mu(0)$ and $\mu(1)$, and hence the ATE, are no longer point identified.
\setcounter{figure}{1}
\begin{figure}[p]
\centering
\subfigure[Panel A: Whole Sample]{\includegraphics[width=0.49\textwidth]{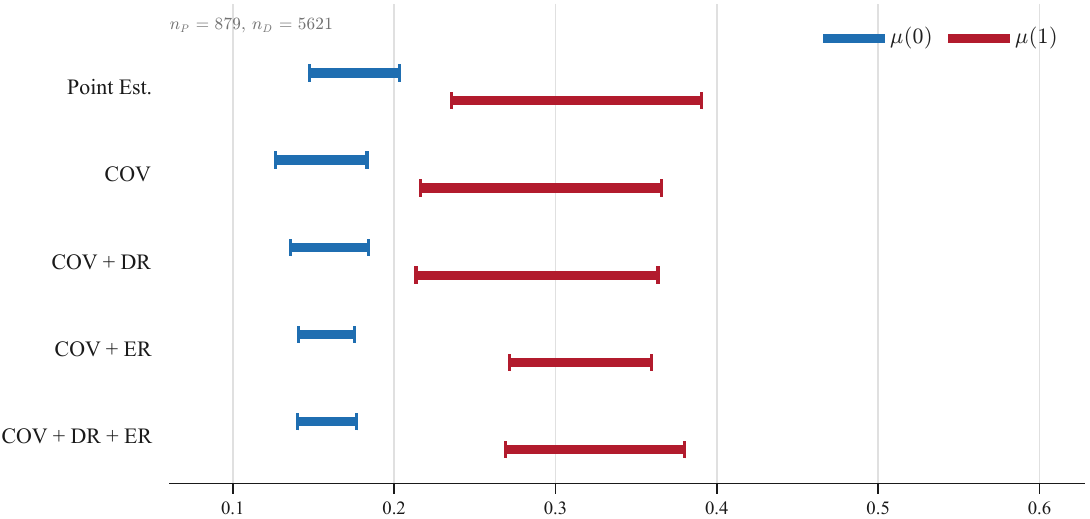}\hfill
\includegraphics[width=0.49\textwidth]{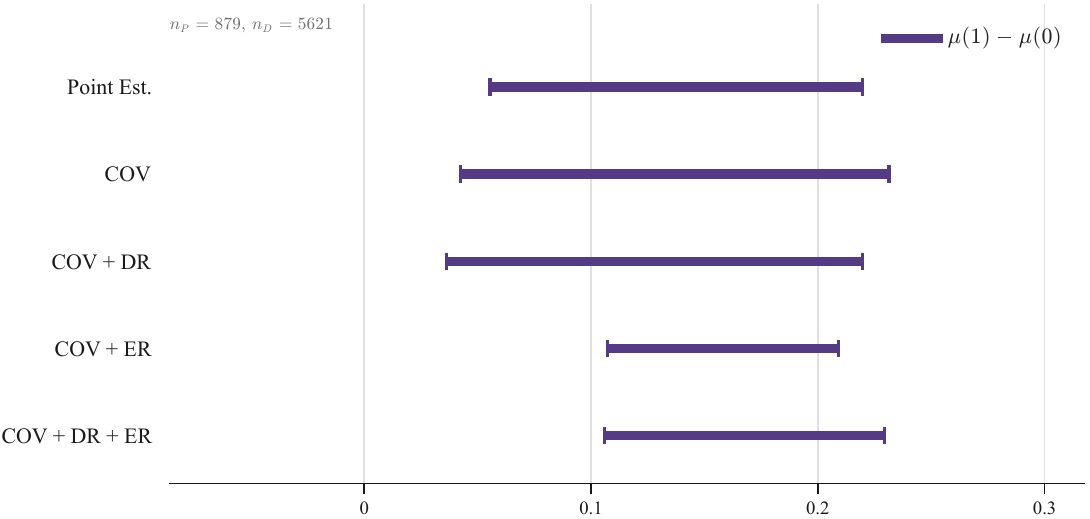}} \\
\subfigure[Panel B: Age 50 to 59]{\includegraphics[width=0.49\textwidth]{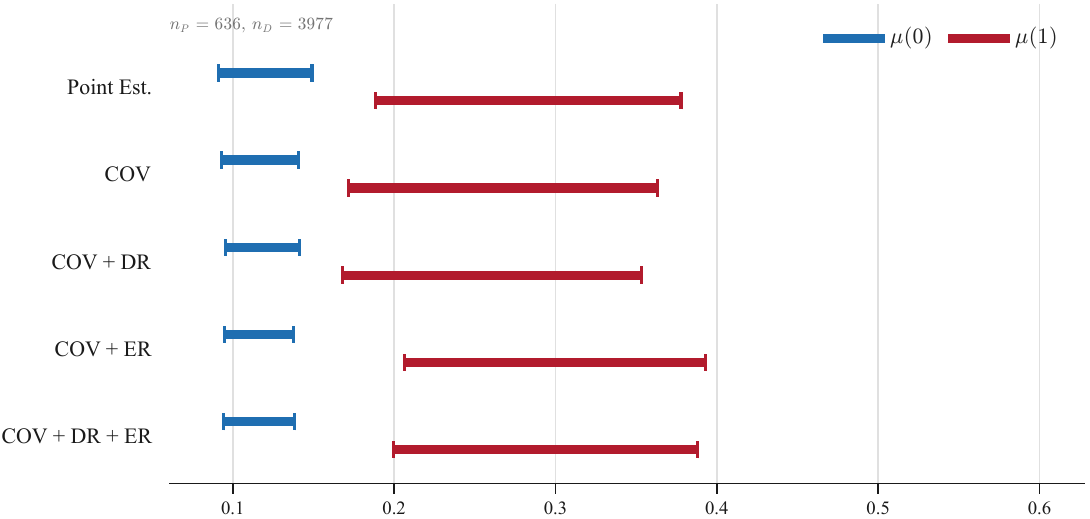}\hfill
\includegraphics[width=0.49\textwidth]{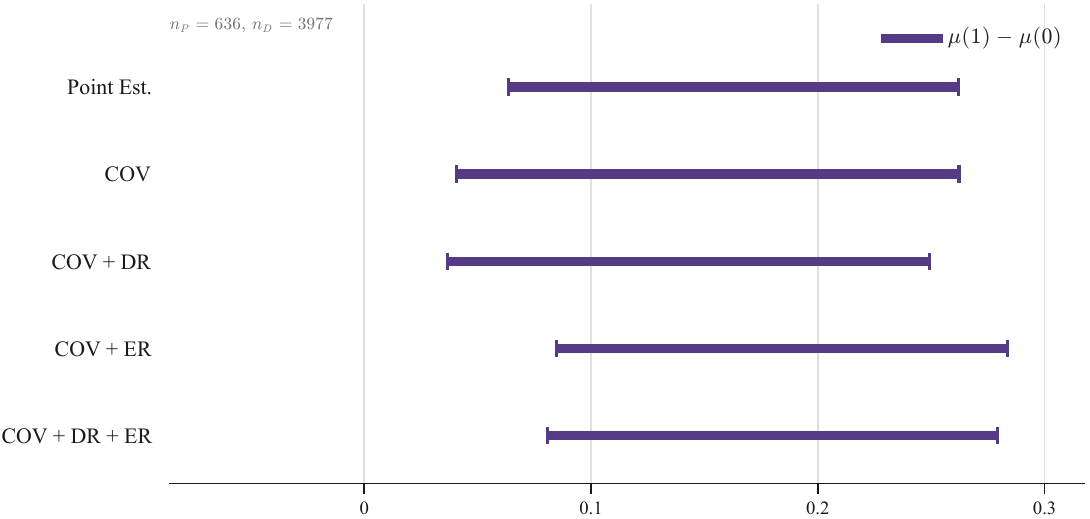}} \\
\subfigure[Panel C: Age 60 to 65]{\includegraphics[width=0.49\textwidth]{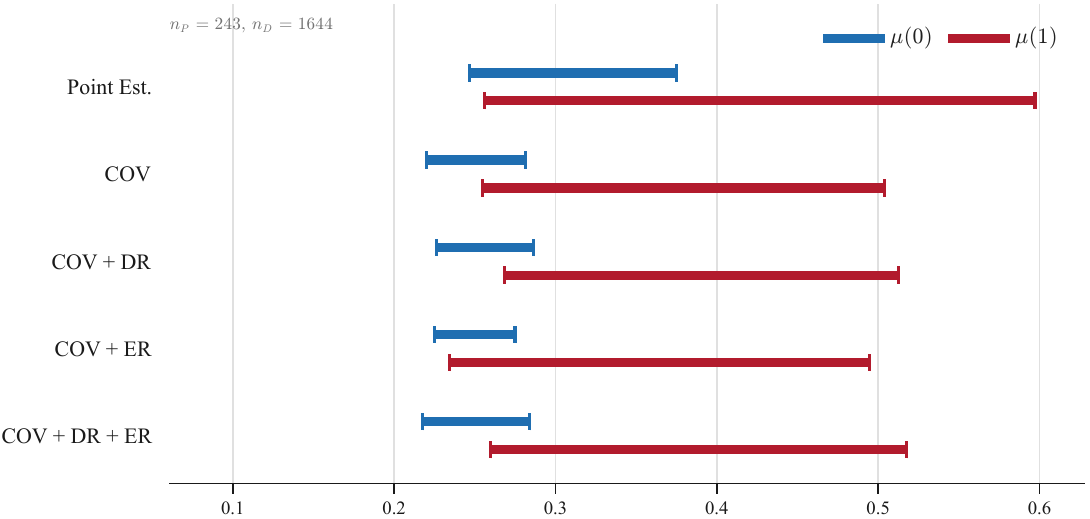}\hfill
\includegraphics[width=0.49\textwidth]{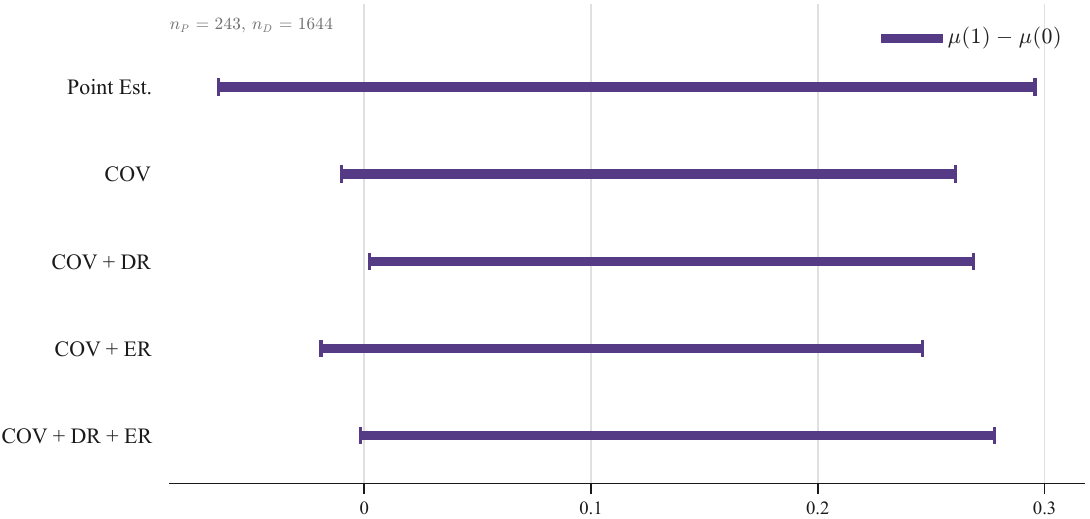}}
\caption{95\% confidence intervals for $\mu(0)$, $\mu(1)$, and $\text{ATE} = \mu(1) - \mu(0)$, by age group, with the outcome $D = 1 - \text{in\_work\_2018}$. In each panel, the left plot shows $\mu(0)$ and $\mu(1)$ and the right plot shows the ATE. In each plot, the top row (Point Est.) is the CI of the SP4RP point estimate based on the matched sample, calculated from the asymptotic distribution in \citet{chernozhukov2018double}. The remaining rows are bootstrap confidence regions for the estimated bounds based on the unmatched samples, with 1,000 bootstrap replications and Bernstein polynomial density estimators of order 5. COV conditions on gender $W$, DR imposes the rank-dependence restriction, and ER imposes the exclusion restrictions with the instrument $Z$. $n_{\mathbf P}$ and $n_D$ are the sizes of the SP and RP samples.}
\label{fig:work}
\end{figure}

Figure~\ref{fig:work} reports 95\% confidence regions for $\mu(0)$, $\mu(1)$, and $\mu(1)-\mu(0)$. The first row of each panel reproduces the confidence interval around the SP4RP point estimate obtained from the matched sample. The remaining rows report confidence regions for the bounds obtained from the unmatched samples. We first use gender as a covariate common to the two samples (COV). We then add the positive rank-dependence restriction (DR). We additionally explore the identifying content of an exclusion restriction using health in 2016 as the excluded variable (ER).\footnote{The exclusion restriction using baseline health $Z$ should be interpreted with some caution, since earlier health may affect subsequent employment through channels not fully captured by health in 2018. We therefore use ER primarily to illustrate how an excluded variable can sharpen the data-combination bounds. In practice, the resulting confidence regions are similar to those obtained without ER.} Finally, we apply both restrictions jointly. The unmatched bounds use Bernstein polynomial estimates of the distribution of $(P_0,P_1)$ and the confidence regions are constructed using the bootstrap procedure of Section~\ref{sec:inference}. Fitting bivariate Bernstein densities assumes that the joint law of $(P_0,P_1)$ is absolutely continuous with respect to Lebesgue measure on $[0,1]^2$; this premise is maintained here and not established. If the empirical report vector is instead singular, the implementation requires either a justified subvector of reports that recovers the type, or a support-adapted density estimator with its own inference argument; a change of notation does not suffice.

Consider first the whole sample in Panel A. Even without linking stated and revealed preferences at the individual level, the data are informative about the effect of health on employment: the confidence region excludes zero and points toward an economically substantial positive effect. Using gender as the only additional covariate, the confidence region for the ATE is approximately $[0.04,0.23]$. Adding the rank-dependence restriction changes the region relatively little. The exclusion restriction is more informative: the lower endpoint rises to about 0.11 and the upper endpoint decreases to 0.21. Combining the dependence and exclusion restrictions produces a confidence region of approximately $[0.11,0.23]$. These intervals overlap substantially with the confidence interval around the matched SP4RP estimate, providing a useful validation of the unmatched analysis.

The average structural functions help explain this pattern. The bounds on $\mu(0)$, the probability of being out of work under good health, are relatively tight under all specifications. There is substantially more uncertainty about $\mu(1)$, reflecting the fact that relatively few respondents are observed in poor health. 

The age-specific results reveal an important difference. Among respondents
aged 50-59, the conclusion is particularly robust. All specifications imply
a positive and economically substantial effect of poor health on the
probability of being out of work. For respondents aged 60-65, uncertainty is considerably greater. The data provide much weaker evidence about the magnitude of the health effect for workers close to retirement age than for younger workers.

The distribution of stated preferences restricts the unobserved heterogeneity sufficiently for the unmatched samples to produce informative bounds, and the resulting confidence regions are broadly consistent with the point-identified SP4RP estimates from the matched sample. The matched analysis showed that taking the stated responses literally substantially overstates the employment response to poor health. The unmatched analysis shows that the same stated responses can nevertheless be informative about actual behavior: once combined with revealed choices, they restrict the range of counterfactual employment responses even without individual-level matching.

\section*{Discussion}

In this paper, we consider agents making binary decisions based on an endogenous choice attribute. We propose to use stated choice probabilities to correct for the endogeneity. We %eschew structural assumptions on the relation between stated choice probabilities and actual choice, and instead 
assume that stated choice probabilities reveal the unobserved heterogeneity that also governs actual choices. For the common case, where stated choice probabilities and actual choices are observed in different data sets, we derive new bounds for counterfactual choice under data combination. These bounds are useful beyond the combination of stated and revealed preferences considered in this work. Although we study binary choice, most of the ideas and results extend to discrete and continuous choice, but the data combination bounds would take a more complex form. The inference procedure proposed here is best suited to empirical environments with scalar or at most bivariate stated probability reports. For empirical environments with elicited preferences from a rich set of counterfactual scenarios, we propose, in ongoing research, to identify a finite number of latent types from elicited preferences using methods from the large factor model and group fixed effects literature.

\newpage
% APPENDIX
\begin{appendix}

\section{Proofs of results in the main text}
\label{app:proofs}

\setcounter{lemma}{0}
\renewcommand{\thelemma}{\thesection.\arabic{lemma}}

\subsection{Conventions and auxiliary results}
\label{app:conventions}

Throughout this appendix, the value~$x\in\mathcal X$ of interest is fixed. The \emph{identified set}~$\mathcal I(x)$ for~$\mu(x)$ under a given collection of assumptions is the set of values of~$\mu(x)$ over all data generating processes that satisfy these assumptions and reproduce the features of the data that are identified in the sampling scheme under consideration (the law of~$(D,\mathbf P,X)$ or of~$(D,\tilde{\mathbf P},X)$ with matched data, the laws of~$(D,X)$ and of~$(\mathbf P,X)$ or~$(\tilde{\mathbf P},X)$ with unmatched data). Bounds~$[LB,UB]$ are \emph{valid} if~$\mathcal I(x)\subseteq[LB,UB]$, and \emph{sharp} if, in addition, $LB=\inf\mathcal I(x)$ and~$UB=\sup\mathcal I(x)$. In every case treated below, $\mathcal I(x)$ is shown to be an interval, so that sharp bounds are the endpoints of its closure; a sharp bound is \emph{attained} when it belongs to~$\mathcal I(x)$. We write~$p(x):=\mathbb P(D=1\vert X=x)$ and~$\pi_x:=\mathbb P(X=x)$ as in the main text, and we use the notation~$\mathcal M(F,F^\prime)$ of the preliminaries for the set of couplings of~$F$ and~$F^\prime$.

The first two lemmas isolate the structure that all data combination bounds in the paper share: because~$D$ is binary, a coupling of~$F_{D\vert X=x}$ with a distribution of reports is described by a single function with values in~$[0,1]$, and, for a fixed integrable cost and fixed marginals, the resulting optimal transport problem reduces to a one-dimensional dual program whose value and optimizers are explicit. The constrained problems of section~\ref{sec:RD} and the coarsened-report problems of section~\ref{sec:partial} build on this reduction but retain an additional layer of optimization, over the multiplier~$\lambda$ or over concentration curves, for which no closed form is claimed. Neither lemma requires any regularity of the cost beyond integrability, which is what allows the density ratio~$f_{\mathbf P}/f_{\mathbf P\vert X=x}$ to serve as a cost under the overlap assumption~\ref{ass:full support} alone.

\begin{lemma}[Couplings with a binary marginal]\label{lem:couplings}
    Let~$\nu$ be a probability measure on a measurable space~$(\mathcal P,\mathcal S)$, $q\in[0,1]$, $\beta_q$ the Bernoulli distribution with parameter~$q$ on~$\{0,1\}$, and~$\mathcal M_q:=\mathcal M(\beta_q,\nu)$. Define
    \begin{eqnarray*}
        \mathcal R_q & := & \left\{ r:\mathcal P\rightarrow[0,1] \mbox{ measurable} \; : \; \int r\,d\nu \; = \; q \right\}.
    \end{eqnarray*}
    For~$\pi\in\mathcal M_q$, let~$r_\pi$ be the Radon--Nikodym derivative of~$A\mapsto\pi(\{1\}\times A)$ with respect to~$\nu$. Then~$\pi\mapsto r_\pi$ is a bijection between~$\mathcal M_q$ and~$\mathcal R_q$ (modulo~$\nu$-null sets), with inverse~$r\mapsto\pi_r$, where~$\pi_r(\{1\}\times A):=\int_Ar\,d\nu$ and~$\pi_r(\{0\}\times A):=\int_A(1-r)\,d\nu$. If~$(D,\mathbf P)\sim\pi\in\mathcal M_q$, then~$\mathbb E_\pi[D\vert\mathbf P]=r_\pi(\mathbf P)$ almost surely and, for every~$g\in L^1(\nu)$,
    \begin{eqnarray}\label{eq:coupling}
        \mathbb E_\pi[Dg(\mathbf P)] & = & \int g\,r_\pi\,d\nu,
    \end{eqnarray}
    the integral being absolutely convergent.
\end{lemma}

\begin{proof}
    For~$\pi\in\mathcal M_q$, the set function~$A\mapsto\pi(\{1\}\times A)$ is a measure bounded above by the second marginal~$\pi(\{0,1\}\times A)=\nu(A)$. It is therefore absolutely continuous with respect to~$\nu$, with Radon--Nikodym derivative~$r_\pi$ satisfying~$0\leq r_\pi\leq1$ $\nu$-almost everywhere, and~$\int r_\pi\,d\nu=\pi(\{1\}\times\mathcal P)=q$ because the first marginal is~$\beta_q$. Hence~$r_\pi\in\mathcal R_q$. Conversely, for~$r\in\mathcal R_q$, $\pi_r$ is a probability measure on~$\{0,1\}\times\mathcal P$ with second marginal~$\nu$ and first marginal~$\pi_r(\{1\}\times\mathcal P)=\int r\,d\nu=q$, so that~$\pi_r\in\mathcal M_q$, and~$r_{\pi_r}=r$ $\nu$-almost everywhere. Since a coupling with a binary first marginal is determined by~$A\mapsto\pi(\{1\}\times A)$ and~$\nu$, the two maps are mutually inverse. Finally, for~$A\in\mathcal S$, $\mathbb E_\pi[D1\{\mathbf P\in A\}]=\pi(\{1\}\times A)=\int_Ar_\pi\,d\nu$, which is the defining property of~$\mathbb E_\pi[D\vert\mathbf P]=r_\pi(\mathbf P)$, and~(\ref{eq:coupling}) follows by the usual extension from indicators to~$L^1(\nu)$, the integral converging absolutely because~$\vert g\,r_\pi\vert\leq\vert g\vert$.
\end{proof}

\begin{lemma}[Duality for binary transport problems]\label{lem:duality}
    In the setting of lemma~\ref{lem:couplings}, let~$g\in L^1(\nu)$ and define, for~$\varphi\in\mathbb R$,
    \begin{eqnarray*}
        U_g(\varphi) \; := \; \int\max(0,g+\varphi)\,d\nu-\varphi q
        & \quad\mbox{and}\quad &
        L_g(\varphi) \; := \; \int\min(0,g-\varphi)\,d\nu+\varphi q.
    \end{eqnarray*}
    \begin{enumerate}
        \item[(i)] The maximum and minimum of~$\pi\mapsto\mathbb E_\pi[Dg(\mathbf P)]$ over~$\mathcal M_q$ are attained and
        \begin{eqnarray}
            \max_{\pi\in\mathcal M_q}\mathbb E_\pi[Dg(\mathbf P)] & = & \inf_{\varphi\in\mathbb R}U_g(\varphi), \label{eq:dual upper}\\
            \min_{\pi\in\mathcal M_q}\mathbb E_\pi[Dg(\mathbf P)] & = & \sup_{\varphi\in\mathbb R}L_g(\varphi). \label{eq:dual lower}
        \end{eqnarray}
        \item[(ii)] If~$q\in\{0,1\}$, then~$\mathcal R_q=\{q\}$ and both values equal~$q\int g\,d\nu$. If~$q\in(0,1)$, let~$Q_g$ denote the quantile function of~$g$ under~$\nu$, $c:=Q_g(1-q)$, $c^\prime:=Q_g(q)$, and let~$\theta,\theta^\prime\in[0,1]$ satisfy~$\nu(g>c)+\theta\nu(g=c)=q$ and~$\nu(g<c^\prime)+\theta^\prime\nu(g=c^\prime)=q$. Then the maximum in~(\ref{eq:dual upper}) is attained at~$\pi_{r^\ast}$ with~$r^\ast:=1\{g>c\}+\theta1\{g=c\}$ and the infimum at~$\varphi=-c$, with
        \begin{eqnarray}\label{eq:dual value}
            \inf_{\varphi\in\mathbb R}U_g(\varphi) \; = \; U_g(-c) \; = \; \int_{\{g>c\}}g\,d\nu+c\big(q-\nu(g>c)\big);
        \end{eqnarray}
        the minimum in~(\ref{eq:dual lower}) is attained at~$\pi_{r_\ast}$ with~$r_\ast:=1\{g<c^\prime\}+\theta^\prime1\{g=c^\prime\}$ and the supremum at~$\varphi=c^\prime$. More generally, the minimizers of~$U_g$ are exactly the~$\varphi$ such that~$-\varphi$ is a~$(1-q)$-quantile of~$g$ under~$\nu$, and the maximizers of~$L_g$ are exactly the~$q$-quantiles of~$g$ under~$\nu$.
        \item[(iii)] The two values have the rearrangement form
        \begin{eqnarray}\label{eq:rearrangement lemma}
            \max_{\pi\in\mathcal M_q}\mathbb E_\pi[Dg(\mathbf P)] \; = \; \int_{1-q}^1Q_g(t)\,dt
            & \quad\mbox{and}\quad &
            \min_{\pi\in\mathcal M_q}\mathbb E_\pi[Dg(\mathbf P)] \; = \; \int_0^qQ_g(t)\,dt.
        \end{eqnarray}
    \end{enumerate}
\end{lemma}

\begin{proof}
    By lemma~\ref{lem:couplings}, the problems in~(i) are the maximization and minimization of~$J(r):=\int g\,r\,d\nu$ over~$r\in\mathcal R_q$.

    \emph{Weak duality.} For~$r\in\mathcal R_q$ and~$\varphi\in\mathbb R$, $(g+\varphi)r\leq\max(0,g+\varphi)$ and~$(g-\varphi)r\geq\min(0,g-\varphi)$ pointwise, because~$0\leq r\leq1$. Integrating against~$\nu$ and using~$\int r\,d\nu=q$ gives~$L_g(\varphi)\leq J(r)\leq U_g(\varphi)$ for every~$r\in\mathcal R_q$ and every~$\varphi\in\mathbb R$.

    \emph{The case~$q\in\{0,1\}$.} Then~$\mathcal R_q=\{q\}$ modulo~$\nu$-null sets, so that the primal values are~$q\int g\,d\nu$ and both extrema are attained. If~$q=0$, then~$U_g(\varphi)=\int\max(0,g+\varphi)\,d\nu$ decreases to~$0$ as~$\varphi\downarrow-\infty$ by monotone convergence, the integrand being dominated by~$\max(0,g)\in L^1(\nu)$ for~$\varphi\leq0$; with weak duality, $\inf_\varphi U_g=0$. If~$q=1$, then~$U_g(\varphi)=\int\max(-\varphi,g)\,d\nu$ decreases to~$\int g\,d\nu$ as~$\varphi\uparrow+\infty$ by dominated convergence, since~$\vert\max(-\varphi,g)\vert\leq\vert g\vert$ for~$\varphi\geq0$; with weak duality, $\inf_\varphi U_g=\int g\,d\nu$. The claims for~$L_g$ follow from the identity~$L_g(\varphi)=-U_{-g}(\varphi)$, which holds because~$\min(0,g-\varphi)=-\max(0,-g+\varphi)$, and from~$\min_{\pi\in\mathcal M_q}\mathbb E_\pi[Dg(\mathbf P)]=-\max_{\pi\in\mathcal M_q}\mathbb E_\pi[D(-g)(\mathbf P)]$.

    \emph{The case~$q\in(0,1)$.} The quantile~$c=Q_g(1-q)=\inf\{t:\nu(g\leq t)\geq1-q\}$ is finite, because~$g$ is~$\nu$-almost surely finite, so that~$\nu(g\leq t)$ tends to~$1>1-q$ as~$t\uparrow+\infty$ and to~$0<1-q$ as~$t\downarrow-\infty$. By definition of the quantile and right-continuity of~$t\mapsto\nu(g\leq t)$, $\nu(g<c)\leq1-q\leq\nu(g\leq c)$, that is, $\nu(g>c)\leq q\leq\nu(g\geq c)$, so that~$\theta\in[0,1]$ exists and~$r^\ast\in\mathcal R_q$. Evaluating~$U_g$ at~$\varphi=-c$ and using~$q-\nu(g>c)=\theta\nu(g=c)$,
    \begin{eqnarray*}
        U_g(-c) \; = \; \int_{\{g>c\}}(g-c)\,d\nu+cq
        & = & \int_{\{g>c\}}g\,d\nu+c\big(q-\nu(g>c)\big) \; = \; \int g\,r^\ast\,d\nu \; = \; J(r^\ast).
    \end{eqnarray*}
    By weak duality, $J(r^\ast)\leq\sup_{\mathcal R_q}J\leq\inf_\varphi U_g\leq U_g(-c)=J(r^\ast)$, so that all four quantities coincide: the maximum is attained at~$r^\ast$, the infimum at~$-c$, and~(\ref{eq:dual upper}) and~(\ref{eq:dual value}) hold. Symmetrically, $\nu(g<c^\prime)\leq q\leq\nu(g\leq c^\prime)$ gives~$\theta^\prime$ and~$r_\ast\in\mathcal R_q$, and
    \begin{eqnarray*}
        L_g(c^\prime) \; = \; \int_{\{g<c^\prime\}}(g-c^\prime)\,d\nu+c^\prime q
        \; = \; \int_{\{g<c^\prime\}}g\,d\nu+c^\prime\big(q-\nu(g<c^\prime)\big) \; = \; J(r_\ast),
    \end{eqnarray*}
    which with weak duality proves~(\ref{eq:dual lower}) and the attainment claims. For the characterization of the dual optimizers, note that~$U_g$ is finite and convex, and that the difference quotients of~$\varphi\mapsto\max(0,g+\varphi)$ are bounded by one in absolute value; by dominated convergence, $U_g$ has one-sided derivatives~$U_g^{\prime+}(\varphi)=\nu(g\geq-\varphi)-q$ and~$U_g^{\prime-}(\varphi)=\nu(g>-\varphi)-q$. Hence~$\varphi$ minimizes~$U_g$ if and only if~$U_g^{\prime-}(\varphi)\leq0\leq U_g^{\prime+}(\varphi)$, i.e., if and only if~$\nu(g>-\varphi)\leq q\leq\nu(g\geq-\varphi)$, which says that~$-\varphi$ is a~$(1-q)$-quantile of~$g$. Likewise~$L_g$ is finite and concave with~$L_g^{\prime+}(\varphi)=q-\nu(g\leq\varphi)$ and~$L_g^{\prime-}(\varphi)=q-\nu(g<\varphi)$, so that its maximizers are the~$\varphi$ with~$\nu(g<\varphi)\leq q\leq\nu(g\leq\varphi)$, the~$q$-quantiles of~$g$.

    \emph{Rearrangement form.} If~$q\in\{0,1\}$, (\ref{eq:rearrangement lemma}) reduces to~$\int_0^1Q_g(t)\,dt=\int g\,d\nu$ and to~$0$, which are the values found above. Let~$q\in(0,1)$ and~$V$ be uniform on~$(0,1)$, so that~$Q_g(V)$ has the law of~$g$ under~$\nu$. Writing $F_g(t):=\nu(g\leq t)$, we have $Q_g(t)\leq c$ if and only if $t\leq F_g(c)$. Moreover, the set~$\{t\in(0,1):Q_g(t)<c\}$ agrees with $(0,F_g(c-))$ up to a Lebesgue-null set. Hence $Q_g=c$ on~$(F_g(c-),F_g(c)]$. Hence~$Q_g=c$ on~$(\nu(g<c),\nu(g\leq c)]$, the interval~$[1-q,\nu(g\leq c)]$ has length~$\nu(g\leq c)-(1-q)=q-\nu(g>c)$, and
    \begin{eqnarray*}
        \int_{1-q}^1Q_g(t)\,dt & = & c\big(q-\nu(g>c)\big)+\mathbb E\left[Q_g(V)1\{V>F_g(c)\}\right] \\
        & = & c\big(q-\nu(g>c)\big)+\mathbb E\left[Q_g(V)1\{Q_g(V)>c\}\right] \\
        & = & c\big(q-\nu(g>c)\big)+\int_{\{g>c\}}g\,d\nu,
    \end{eqnarray*}
    which is~$U_g(-c)$ by~(\ref{eq:dual value}). The same computation with~$c'$, using that~$\{t\in(0,1):Q_g(t)<c'\}$ agrees with $(0,\nu(g<c'))$ up to a Lebesgue-null set and that~$\nu(g<c')\leq q\leq\nu(g\leq c')$, gives$\int_0^q Q_g(t)\,dt=\int_{\{g<c'\}} g\,d\nu +c'\bigl(q-\nu(g<c')\bigr)=L_g(c').$
\end{proof}

\subsection{Proofs for section~\ref{sec:identification}}

\begin{proof}[Proof of proposition~\ref{prop:identification ASF}]
The argument has three steps: recover~$\eta$ from~$\mathbf P$ as a measurable transformation, transfer the conditional independence of assumption~\ref{ass:endogeneity} from~$\eta$ to~$\mathbf P$, and assemble the identification formula, whose integrand is well defined by assumption~\ref{ass:full support}.

\emph{Step 1: $\sigma(\mathbf P)=\sigma(\eta)$.} By construction~$P_t=\tilde m(x_t,\eta)$ for~$t=1,\dots,T$, so~$\mathbf P$ is a measurable function of~$\eta$ and $\sigma(\mathbf P)\subseteq\sigma(\eta)$. For the reverse inclusion, assumption~\ref{ass:invertibility} provides a subvector~$(x_1,\dots,x_d)$ of scenarios for which
\[
\tilde{\mathbf m}:\eta\longmapsto\bigl(\tilde m(x_1,\eta),\dots,\tilde m(x_d,\eta)\bigr)
\]
is continuous and one-to-one on the support~$\mathcal H=\operatorname{supp}(F_\eta)\subseteq\mathbb R^{d}$. Since~$\mathcal H$ is closed in~$\mathbb R^{d}$, it is the union of an increasing sequence of compact sets~$\mathcal H_j$. On each~$\mathcal H_j$, the continuous injection~$\tilde{\mathbf m}\vert_{\mathcal H_j}$ is a homeomorphism onto its compact image (a continuous bijection from a compact set onto a Hausdorff space), so that~$\tilde{\mathbf m}(A\cap\mathcal H_j)$ is a Borel set for every Borel~$A\subseteq\mathcal H$. Hence~$\tilde{\mathbf m}(A)=\bigcup_j\tilde{\mathbf m}(A\cap\mathcal H_j)$ is Borel, which means that~$\tilde{\mathbf m}^{-1}$ is a Borel map on~$\tilde{\mathbf m}(\mathcal H)$ and that~$\eta=\tilde{\mathbf m}^{-1}(P_1,\dots,P_d)$ is a measurable function of~$\mathbf P$. Thus~$\sigma(\eta)\subseteq\sigma(\mathbf P)$, and the two inclusions give~$\sigma(\mathbf P)=\sigma(\eta)$.

\emph{Step 2: conditional independence given~$\mathbf P$.} Because~$\sigma(\mathbf P)=\sigma(\eta)$, conditioning on~$\mathbf P$ and on~$\eta$ yield the same conditional expectations, so that assumption~\ref{ass:endogeneity} reads~$D(x)\perp\!\!\!\perp X\mid\mathbf P$. In particular, $\mathbb E[D(x)1\{X=x\}\vert\mathbf P]=\mathbb E[D(x)\vert\mathbf P]\,\mathbb P(X=x\vert\mathbf P)$, whence
\[
\mathbb E[D(x)\mid\mathbf P=\mathbf p]=\mathbb E[D(x)\mid X=x,\mathbf P=\mathbf p]
\]
for~$F_{\mathbf P\mid X=x}$-almost every~$\mathbf p$, i.e., on the set~$\{\mathbb P(X=x\vert\mathbf P=\mathbf p)>0\}$ where the right-hand side is defined.

\emph{Step 3: assembly.} By part~(1) of assumption~\ref{ass:full support}, $F_{\mathbf P}\ll F_{\mathbf P\mid X=x}$, so that the regression~$\mathbf p\mapsto\mathbb E[D\mid X=x,\mathbf P=\mathbf p]$ is defined~$F_{\mathbf P}$-almost everywhere, the equality of step~2 holds~$F_{\mathbf P}$-almost everywhere, and the integrals below are well defined. Using definition~\ref{def:ASF}, the law of iterated expectations, step~2, and the consistency relation~$D=D(x)$ on~$\{X=x\}$,
\begin{align*}
\mu(x)
&= \mathbb E[D(x)]
&& \text{(definition~\ref{def:ASF})}\\
&= \int \mathbb E[D(x)\mid \mathbf P=\mathbf p]\,dF_{\mathbf P}(\mathbf p)
&& \text{(iterated expectations)}\\
&= \int \mathbb E[D(x)\mid X=x,\mathbf P=\mathbf p]\,dF_{\mathbf P}(\mathbf p)
&& \text{(step 2)}\\
&= \int \mathbb E[D\mid X=x,\mathbf P=\mathbf p]\,dF_{\mathbf P}(\mathbf p)
&& (D=D(x)\text{ on }\{X=x\}),
\end{align*}
which is the claimed identification.
\end{proof}

\subsection{Proofs for section~\ref{sec:bounds}}

Throughout this subsection, $\nu:=F_{\mathbf P\vert X=x}$, $q:=p(x)=\mathbb E[D\vert X=x]$, and~$\mathcal M:=\mathcal M(F_{D\vert X=x},\nu)$, which is the set~$\mathcal M_q$ of lemma~\ref{lem:couplings} with~$\mathcal P=[0,1]^T$. The map~$\pi\mapsto r_\pi$ and the set~$\mathcal R_q$ are those of lemma~\ref{lem:couplings}.

\begin{proof}[Proof of proposition~\ref{prop:sharp}]
\emph{Step 1: the identified set is~$\mathcal R_q$.} By definition, $m(x,\cdot)$ belongs to the identified set if and only if there is~$\pi\in\mathcal M$ such that~$\mathbb E_\pi[D\vert\mathbf P=\mathbf p]=m(x,\mathbf p)$ for~$\nu$-almost every~$\mathbf p$. By lemma~\ref{lem:couplings}, this holds if and only if~$m(x,\cdot)=r_\pi$ $\nu$-almost everywhere for some~$\pi\in\mathcal M$, that is, if and only if
\begin{eqnarray}\label{eq:primalset}
    0 \; \leq \; m(x,\cdot) \; \leq \; 1 \;\; \nu\mbox{-a.e.}
    \qquad\mbox{and}\qquad \int m(x,\mathbf p) \; d\nu(\mathbf p) \; = \; q.
\end{eqnarray}

\emph{Step 2: necessity.} Let~$m(x,\cdot)$ satisfy~(\ref{eq:primalset}) and let~$\pi:=\pi_{m(x,\cdot)}\in\mathcal M$. For every~$h\in L^1(\nu)$, the product~$h\,m(x,\cdot)$ is integrable and, by~(\ref{eq:coupling}), $\mathbb E[h(\mathbf P)m(X,\mathbf P)\vert X=x]=\mathbb E_\pi[h(\mathbf P)D\vert X=x]$. By part~(i) of lemma~\ref{lem:duality} applied with~$g=h$, this quantity lies between~$\sup_\varphi L_h(\varphi)$ and~$\inf_\varphi U_h(\varphi)$, which are the two bounds displayed in the proposition. In particular, the extrema in~(\ref{eq:primal}) are attained, which justifies the notation there, and~(\ref{eq:Bierens}) holds under~$\pi$ for every~$h\in L^1(\nu)$.

\emph{Step 3: sufficiency.} Conversely, suppose the displayed inequalities hold for every~$h\in L^1(\nu)$ with~$h\,m(x,\cdot)\in L^1(\nu)$. We use them only for bounded~$h$, for which the integrability requirement is automatic since~$m(x,\cdot)\in L^1(\nu)$. Let~$A\subseteq[0,1]^T$ be measurable. Taking~$h=-1\{\cdot\in A\}$ and~$\varphi=0$ in the upper bound gives~$-\int_Am(x,\cdot)\,d\nu\leq\int\max(0,-1\{\mathbf p\in A\})\,d\nu=0$, and since~$A$ is arbitrary, $m(x,\cdot)\geq0$ $\nu$-almost everywhere. Taking~$h=1\{\cdot\in A\}$ and~$\varphi=0$ gives~$\int_Am(x,\cdot)\,d\nu\leq\nu(A)$ for all~$A$, hence~$m(x,\cdot)\leq1$ $\nu$-almost everywhere. Finally, $h\equiv1$ with~$\varphi=-1$ gives~$\int m(x,\cdot)\,d\nu\leq q$, while~$h\equiv-1$ with~$\varphi=1$ gives~$-\int m(x,\cdot)\,d\nu\leq-q$; together, $\int m(x,\cdot)\,d\nu=q$. Thus~(\ref{eq:primalset}) holds and~$m(x,\cdot)$ is in the identified set by step~1.
\end{proof}

\begin{proof}[Proof of theorem~\ref{thm:sharp bounds}]
As in the statement of the theorem, $\varrho$ denotes the dominating measure with respect to which~$f_{\mathbf P}$ and~$f_{\mathbf P\vert X=x}$ are defined in assumption~\ref{ass:full support}. Define~$h(\mathbf p):=f_{\mathbf P}(\mathbf p)/f_{\mathbf P\vert X=x}(\mathbf p)$ on~$\{f_{\mathbf P\vert X=x}>0\}$ and~$h(\mathbf p):=0$ elsewhere.

\emph{Step 1: $h$ is a Radon--Nikodym derivative.} Since~$\nu(f_{\mathbf P\vert X=x}=0)=0$ and~$F_{\mathbf P}\ll\nu$ by part~(1) of assumption~\ref{ass:full support}, we have~$F_{\mathbf P}(f_{\mathbf P\vert X=x}=0)=0$, so that~$f_{\mathbf P}=0$ almost everywhere on~$\{f_{\mathbf P\vert X=x}=0\}$. Hence, for every Borel set~$A$,
\begin{eqnarray}\label{eq:RN}
    F_{\mathbf P}(A) \; = \; \int_{A\cap\{f_{\mathbf P\vert X=x}>0\}} f_{\mathbf P}\;d\varrho(\mathbf p)
    \; = \; \int_A h\,f_{\mathbf P\vert X=x}\;d\varrho(\mathbf p) \; = \; \int_A h\;d\nu,
\end{eqnarray}
so that~$h=dF_{\mathbf P}/d\nu$, $h\geq0$, $\int h\,d\nu=1$, and in particular~$h\in L^1(\nu)$. No other property of~$h$ is used below.

\emph{Step 2: the identified set for~$\mu(x)$.} Under assumptions~\ref{ass:endogeneity}--\ref{ass:full support}, proposition~\ref{prop:identification ASF} gives~$\mu(x)=\int m(x,\mathbf p)\,dF_{\mathbf P}(\mathbf p)$ with~$m(x,\cdot)=\mathbb E[D\vert X=x,\mathbf P=\cdot\,]$, and by~(\ref{eq:RN}),
\begin{eqnarray*}
    \mu(x) & = & \int m(x,\mathbf p)\,h(\mathbf p)\;d\nu(\mathbf p) \; =: \; J\big(m(x,\cdot)\big),
\end{eqnarray*}
the integral converging absolutely since~$0\leq m(x,\cdot)\leq 1$ and~$h\in L^1(\nu)$. The data identify~$\nu$, $F_{D\vert X=x}$ and~$F_{\mathbf P}$, hence~$h$ and~$J$, but not the coupling of~$D$ and~$\mathbf P$ given~$X=x$, which ranges over~$\mathcal M$; by step~1 of the proof of proposition~\ref{prop:sharp}, $m(x,\cdot)$ ranges over~$\mathcal R_q$ (modulo~$\nu$-null sets, which do not affect~$J$). Any~$r\in\mathcal R_q$ is the structural function of a data generating process satisfying assumptions~\ref{ass:endogeneity}--\ref{ass:full support} and reproducing the identified marginals: take~$\eta:=\mathbf P$, keep the identified law of~$(\mathbf P,X)$, let~$D(x)=1\{r(\mathbf P)\geq V\}$ and~$D(x^\prime)=1\{p(x^\prime)\geq V^\prime\}$ for~$x^\prime\neq x$, with~$V,V^\prime$ independent uniform variables on~$(0,1)$, independent of~$(\mathbf P,X)$, and~$D:=D(X)$, so that~$D(x^{\prime\prime})\perp\!\!\!\!\perp X\vert\mathbf P$ for every~$x^{\prime\prime}$, $\mathbb E[D\vert X=x]=\int r\,d\nu=q$ and~$\mathbb E[D\vert X=x^\prime]=p(x^\prime)$. The identified set for~$\mu(x)$ is therefore~$\mathcal I(x)=\{J(r):r\in\mathcal R_q\}$, and by~(\ref{eq:coupling}), $J(r_\pi)=\mathbb E_\pi[h(\mathbf P)D\vert X=x]$ for every~$\pi\in\mathcal M$.

\emph{Step 3: $\mathcal I(x)$ is the closed interval between the primal values.} The set~$\mathcal R_q$ is convex and~$J$ is linear on it, so~$\mathcal I(x)$ is a convex subset of~$\mathbb R$, i.e., an interval. Since~$h\in L^1(\nu)$, part~(i) of lemma~\ref{lem:duality} applies with~$g=h$: the supremum and infimum of~$J$ over~$\mathcal R_q$ are attained, and they equal the dual values
\begin{eqnarray*}
    \max_{\pi\in\mathcal M}\mathbb E_\pi[h(\mathbf P)D\vert X=x] \; = \; \inf_{\varphi\in\mathbb R} U_x(\varphi), & \qquad &
    U_x(\varphi) \; := \; \mathbb E[\max(0,h(\mathbf P)+\varphi)\vert X=x]-\varphi q, \\
    \min_{\pi\in\mathcal M}\mathbb E_\pi[h(\mathbf P)D\vert X=x] \; = \; \sup_{\varphi\in\mathbb R} L_x(\varphi), & \qquad &
    L_x(\varphi) \; := \; \mathbb E[\min(0,h(\mathbf P)-\varphi)\vert X=x]+\varphi q.
\end{eqnarray*}
Hence~$\mathcal I(x)=[\sup_\varphi L_x(\varphi),\inf_\varphi U_x(\varphi)]$, and both endpoints are attained. The endpoints are the primal values in~(\ref{eq:primal bounds}).

\emph{Step 4: rewriting the dual objectives in terms of densities.} Fix~$\varphi\in\mathbb R$. On~$\{f_{\mathbf P\vert X=x}>0\}$, multiplying through by~$f_{\mathbf P\vert X=x}$,
\begin{eqnarray*}
    \max(0,h+\varphi)\,f_{\mathbf P\vert X=x} & = & \max\big(0,f_{\mathbf P}+\varphi f_{\mathbf P\vert X=x}\big)
    \; = \; \max\big(-\varphi f_{\mathbf P\vert X=x},f_{\mathbf P}\big)+\varphi f_{\mathbf P\vert X=x}.
\end{eqnarray*}
On~$\{f_{\mathbf P\vert X=x}=0\}$, the left-hand side is zero, and the right-hand side reduces to~$f_{\mathbf P}$, which is zero almost everywhere there by step~1. Both sides are integrable, being dominated by~$\vert\varphi\vert f_{\mathbf P\vert X=x}+f_{\mathbf P}$. Integrating and using~$\int f_{\mathbf P\vert X=x}\,d\varrho(\mathbf p)=1$,
\begin{eqnarray*}
    U_x(\varphi) & = & \int\max\big(-\varphi f_{\mathbf P\vert X=x}(\mathbf p),f_{\mathbf P}(\mathbf p)\big)\;d\varrho(\mathbf p)
    +\varphi-\varphi q \\
    & = & \int\max\big(-\varphi f_{\mathbf P\vert X=x}(\mathbf p),f_{\mathbf P}(\mathbf p)\big)\;d\varrho(\mathbf p)
    +\varphi\,\mathbb E[1-D\vert X=x].
\end{eqnarray*}
Identically, $\min(0,h-\varphi)f_{\mathbf P\vert X=x}=\min(\varphi f_{\mathbf P\vert X=x},f_{\mathbf P})-\varphi f_{\mathbf P\vert X=x}$ almost everywhere, whence
\begin{eqnarray*}
    L_x(\varphi) & = & \int\min\big(\varphi f_{\mathbf P\vert X=x}(\mathbf p),f_{\mathbf P}(\mathbf p)\big)\;d\varrho(\mathbf p)
    -\varphi\,\mathbb E[1-D\vert X=x].
\end{eqnarray*}
The objectives in~(\ref{eq:dual bounds}) therefore coincide with~$L_x$ and~$U_x$ for every~$\varphi\in\mathbb R$. Together with step~3, this proves~(\ref{eq:dual bounds}) and the sharpness of the bounds.

\emph{Step 5: dual optimizers and the rearrangement form.} By part~(ii) of lemma~\ref{lem:duality}, if~$q\in(0,1)$ the infimum of~$U_x$ is attained at~$\varphi=-Q_{h(\mathbf P)\vert X=x}(1-q)$ and the supremum of~$L_x$ at~$\varphi=Q_{h(\mathbf P)\vert X=x}(q)$, both of which are finite; if~$q\in\{0,1\}$, both bounds equal~$q$. In particular, when~$q\in(0,1)$, the optimization in~(\ref{eq:dual bounds}) can be restricted to a compact interval~$[-M,M]$ with~$M\geq Q_{h(\mathbf P)\vert X=x}(\max\{q,1-q\})$, as is done in section~\ref{sec:inference}, without changing the values of the bounds. When~$q\in\{0,1\}$, the dual values are approached as~$\varphi\rightarrow\mp\infty$, as shown in the proof of lemma~\ref{lem:duality}, but need not be attained: if~$q=0$ and~$h$ is essentially unbounded under~$\nu$, the upper program restricted to~$[-M,M]$ has value~$\int(h-M)_+\,d\nu>0$ rather than~$0$. In every case the truncated programs deliver valid outer bounds, since restricting the domain of an infimum (supremum) can only increase (decrease) its value. By part~(iii) of lemma~\ref{lem:duality} with~$p(x)=q$, the dual bounds~(\ref{eq:dual bounds}) coincide with the rearrangement bounds~(\ref{eq:rearrangement bounds}).
\end{proof}

\begin{proof}[Proof of the claims following theorem~\ref{thm:sharp bounds}]
We use the notation of the proof of theorem~\ref{thm:sharp bounds} and write~$Q_h:=Q_{h(\mathbf P)\vert X=x}$, so that~$LB_x=\int_0^qQ_h(t)\,dt$ and~$UB_x=\int_{1-q}^1Q_h(t)\,dt$ by step~5 of that proof, with~$\int_0^1Q_h(t)\,dt=\int h\,d\nu=1$.

\emph{No endogeneity case.} Let~$q\in(0,1)$. If~$F_{\mathbf P\vert X=x}=F_{\mathbf P}$, then~$h=1$ $\nu$-almost everywhere, $Q_h\equiv1$, and~$LB_x=UB_x=q$. Conversely, suppose~$LB_x=UB_x$. If~$q\leq1/2$, monotonicity of~$Q_h$ gives
\begin{eqnarray*}
    \int_0^qQ_h(t)\,dt \; \leq \; q\,Q_h(q) \; \leq \; q\,Q_h(1-q) \; \leq \; \int_{1-q}^1Q_h(t)\,dt,
\end{eqnarray*}
and equality throughout forces~$Q_h=Q_h(q)$ almost everywhere on~$(0,q)$, $Q_h=Q_h(1-q)$ almost everywhere on~$(1-q,1)$, and~$Q_h(q)=Q_h(1-q)$, so that~$Q_h$ is almost everywhere constant on~$(0,1)$. If~$q>1/2$, removing the common integral over~$(1-q,q)$ from both sides, the equality reads~$\int_0^{1-q}Q_h(t)\,dt=\int_q^1Q_h(t)\,dt$, and the same argument on the intervals~$(0,1-q)$ and~$(q,1)$ forces~$Q_h$ to be almost everywhere constant on~$(0,1)$. In both cases, $h$ is~$\nu$-almost surely equal to its mean~$\int h\,d\nu=1$, that is, $F_{\mathbf P}=\nu$. Finally, $F_{\mathbf P\vert X=x}=F_{\mathbf P}$ is equivalent to the independence of~$\mathbf P$ and the event~$\{X=x\}$, hence of~$\eta$ and~$\{X=x\}$ by step~1 of the proof of proposition~\ref{prop:identification ASF}.

\emph{Relation with worst-case bounds.} Since~$f_{\mathbf P}=\sum_{x^\prime\in\mathcal X}\pi_{x^\prime}f_{\mathbf P\vert X=x^\prime}\geq\pi_xf_{\mathbf P\vert X=x}$ almost everywhere, we have~$\min(\pi_xf_{\mathbf P\vert X=x},f_{\mathbf P})=\pi_xf_{\mathbf P\vert X=x}$ and~$\max(\pi_xf_{\mathbf P\vert X=x},f_{\mathbf P})=f_{\mathbf P}$. Hence the lower bound objective in~(\ref{eq:dual bounds}) evaluated at~$\varphi=\pi_x$ equals~$\pi_x-\pi_x(1-q)=\mathbb P(D=1,X=x)=\mu^-(x)$, and the upper bound objective evaluated at~$\varphi=-\pi_x$ equals~$1-\pi_x(1-q)=1-\mathbb P(D=0,X=x)=\mu^+(x)$. Using~$\int(f_{\mathbf P\vert X=x}-f_{\mathbf P})_+\,d\varrho(\mathbf p)=\int(f_{\mathbf P}-f_{\mathbf P\vert X=x})_+\,d\varrho(\mathbf p)=\mathrm{TV}_x$, which holds because both densities integrate to one, the lower bound objective at~$\varphi=1$ equals~$\int\min(f_{\mathbf P\vert X=x},f_{\mathbf P})\,d\varrho(\mathbf p)-(1-q)=1-\mathrm{TV}_x-1+q=q-\mathrm{TV}_x$, and the upper bound objective at~$\varphi=-1$ equals~$\int\max(f_{\mathbf P\vert X=x},f_{\mathbf P})\,d\varrho(\mathbf p)-(1-q)=1+\mathrm{TV}_x-1+q=q+\mathrm{TV}_x$. Since the lower bound is a supremum and the upper bound an infimum over~$\varphi$, the displayed outer bounds follow.
\end{proof}

\begin{proof}[Proof of theorem~\ref{thm:bounds}]
Fix $x\in\mathcal X$ and let $s_{k,l}:=\mathbb P(Z=z_k\mid X=x,W=w_l)$, $q_{k,l} := \mathbb E [D \vert X=x, Z=z_k, W=w_l]$, $\nu_{k,l}:=F_{\mathbf P\mid X=x,Z=z_k,W=w_l}$, and $\nu_l:=F_{\mathbf P\mid X=x,W=w_l}=\sum_{k=1}^K s_{k,l}\nu_{k,l}$. For $d\in\{0,1\}$, let $\zeta_{d,k,l}(p):=\zeta(d,p\mid x,z_k,w_l)$ and \[h_l:=\frac{f_{\mathbf P\mid W=w_l}}{f_{\mathbf P\mid X=x,W=w_l}}=\frac{dF_{\mathbf P\mid W=w_l}}{d\nu_l},\] where the second equality follows as in Step~1 of the proof of Theorem~\ref{thm:sharp bounds}. Assumption~\ref{ass:full support'} gives $F_{\mathbf P\mid W=w_l}\ll\nu_l$, so that $h_l\ge0$ and $\int h_l\,d\nu_l=1$. Moreover, $s_{k,l}\nu_{k,l}\le\nu_l$, and hence \[ \int h_l\,d\nu_{k,l}\le\frac{1}{s_{k,l}}\int h_l\,d\nu_l=\frac{1}{s_{k,l}}<\infty. \] By part~(1) of Assumption~\ref{ass:copula}, $\zeta_{0,k,l},\zeta_{1,k,l}\in L^1(\nu_{k,l})$.

\emph{Step 1: conditional independence given~$(\mathbf P,W)$.} By step~1 of the proof of proposition~\ref{prop:identification ASF}, assumption~\ref{ass:invertibility} gives~$\sigma(\mathbf P)=\sigma(\eta)$, so that assumption~\ref{ass:endogeneity'} reads~$D(x)\perp\!\!\!\!\perp(X,Z)\vert\mathbf P,W$. For any $k=1,2,\ldots K$ and $l=1,2,\ldots,L$,
\begin{eqnarray*}
    \mathbb E[D(x)\vert\mathbf P=\mathbf p,W=w_l] & = & \mathbb E[D(x)\vert \mathbf P = \mathbf p, X=x, Z=z_k, W=w_l] \\
    & = & \mathbb E[D\vert \mathbf P = \mathbf p, X=x, Z=z_k, W=w_l] \; =: \; r_{k,l}(\mathbf p)
\end{eqnarray*}
for~$\nu_{k,l}$-almost every~$\mathbf p$, using~$D=D(x)$ on~$\{X=x\}$. 

Let $a_{k,l} := \frac{d (s_{k,l} \nu_{k,l})}{d \nu_l}$ and $\overline{r}_l := \sum_{k=1}^K a_{k,l} r_{k,l}$. Then, $a_{k,l} = \prob{Z=z_k \vert \Pbf = \pbf, X=x, W=w_l}$ and $\overline{r}_l(\pbf) = \expect{D \vert \Pbf = \pbf, X=x, W=w_l}$ by the law of iterated expectations. Since~$\sum_{k=1}^K a_{k,l}=1$ $\nu_l$-almost everywhere, and since~$a_{k,l}=0$ $\nu_l$-almost everywhere on every~$\nu_{k,l}$-null set, because~$\int_A a_{k,l}\,d\nu_l=s_{k,l}\nu_{k,l}(A)$, the displayed equalities give~$\overline{r}_l(\pbf)=\expect{D(x)\vert\Pbf=\pbf,W=w_l}$ for~$\nu_l$-almost every~$\pbf$, hence for~$F_{\mathbf P\vert W=w_l}$-almost every~$\pbf$ by Assumption~\ref{ass:full support'}, and for~$\nu_{k,l}$-almost every~$\pbf$ because~$s_{k,l}\nu_{k,l}\le\nu_l$. Hence, the exclusion restriction implies $r_{k,l} = \overline{r}_l$ $\nu_{k,l}$-almost everywhere for all $k=1,2,\ldots K$ and $l=1,2,\ldots,L$. 

\emph{Step 2: transport identity for the cell.} By iterated expectations given~$W=w_l$, step~1 and the change of measure~$dF_{\mathbf P\vert W=w_l}=h_l d\nu_l$,
\begin{eqnarray}\label{eq:cell transport}
    \mathbb E[D(x)\vert W=w_l] & = & \int\mathbb E[D(x)\vert\mathbf P=\mathbf p,W=w_l]\,dF_{\mathbf P\vert W=w_l}(\mathbf p) \nonumber \\ 
    & = & \int \overline{r}_l h_l d\nu_l \nonumber \\ 
    & = & \sum_{k=1}^K s_{k,l} \int r_{k,l} h_l d\nu_{k,l} \; =: \; J_l (\mathbf r_l), 
\end{eqnarray}
where $\mathbf r_l = (r_{1,l}, r_{2,l}, \ldots, r_{K,l})$. Then, 
\[
\mu(x) = \sum_{l=1}^L \prob{W=w_l} J_l(\mathbf r_l) =: J(\mathbf r), \qquad \mathbf r = (\mathbf r_l)_l. 
\]

\emph{Step 3: the identified set.} Data combination identifies the marginal distributions of $D$ and $\mathbf P$, but not their joint distribution $\pi_{k,l}$ within each $(X=x, Z=z_k, W=w_l)$-cell. Applying lemma~\ref{lem:couplings}, these couplings are parameterized by~$r_{k,l}=r_{\pi_{k,l}}$ that ranges over~$\mathcal R_{q_{k,l}}$. 

Let~$\mathcal R := \prod_{l=1}^L \prod_{k=1}^K \mathcal R_{q_{k,l}}$. For an arbitrary coupling represented by~$\mathbf r \in \mathcal R$, the dependence restriction in part~(2) of assumption~\ref{ass:copula} reads 
\begin{eqnarray*}
    C_{k,l}(r_{k,l}) & := & s_{k,l} \int \left\{ r_{k,l} \zeta_{1,k,l} + (1-r_{k,l}) \zeta_{0,k,l} \right\} d \nu_{k,l} \le 0. 
\end{eqnarray*}
for all $k=1,2,\ldots, K$ and $l=1,2,\ldots L$. Further define~$\Gamma := \prod_{l=1}^L \prod_{k=1}^K L^1(\nu_{k,l})$ and 
\[
G_\gamma(\mathbf r) := \sum_{l=1}^L \mathbb P \{W=w_l\} \sum_{k=1}^K s_{k,l} \int \gamma_{k,l} (\overline{r}_l - r_{k,l}) d \nu_{k,l} 
\]
for $\gamma = (\gamma_{k,l})_{k,l} \in \Gamma$. The exclusion restrictions are equivalent to $G_\gamma(\mathbf r) = 0$ for every $\gamma \in \Gamma$. 

Combining two restrictions, the feasible set for $\mathbf r$ is 
\[
\mathcal F := \left\{ \mathbf r \in \mathcal R \; : \; C_{k,l}(r_{k,l}) \le 0 \text{ for all } k, l, \text{ and } G_{\gamma}(\mathbf r) = 0 \text{ for all } \gamma \in \Gamma \right\}. 
\]
It is nonempty under the maintained assumptions and 
\begin{eqnarray}\label{eq:constrained ot bounds}
    LB \; := \; \inf_{\mathbf r \in \mathcal F} J(\mathbf r) \; \le \; \mu(x) \; \le \; \sup_{\mathbf r \in \mathcal F} J(\mathbf r) \; =: \; UB. 
\end{eqnarray}

We regard~$\mathcal R_{q_{k,l}}$ as a subset of~$L^\infty(\nu_{k,l})$, the dual of~$L^1(\nu_{k,l})$, equipped with its weak-$\ast$ topology. It is convex, it is weak-$\ast$ closed, being the intersection of the weak-$\ast$ closed sets~$\{r:\int rg\,d\nu_{k,l}\geq0\}$, $\{r:\int(1-r)g\,d\nu_{k,l}\geq0\}$ over non-negative~$g\in L^1(\nu_{k,l})$, and~$\{r:\int r\,d\nu_{k,l}=q_{k,l}\}$, and it is contained in the unit ball, so that it is weak-$\ast$ compact by the Banach--Alaoglu theorem. Hence, $\mathcal R$ is compact and convex in the finite product weak-$\ast$ topology. 

The functionals~$J$ and~$C_{k,l}$ are affine and weak-$\ast$ continuous, their integrands being integrable with respect to the corresponding $\nu_{k,l}$. To establish the same property for~$G_\gamma$, fix~$\gamma \in \Gamma$ and set~$\overline{\gamma}_l := \sum_{j=1}^K a_{j,l} \gamma_{j,l}$. Since $s_{k,l} \nu_{k,l} \le \nu_l$, 
\[
\int | \overline{\gamma}_l | d \nu_{k,l} \; \le \; \frac{1}{s_{k,l}} \int | \overline{\gamma}_l | d \nu_l \; \le \; \frac{1}{s_{k,l}} \sum_{j=1}^K s_{j,l} \int | \gamma_{j,l} | d \nu_{j,l} \; < \; \infty, 
\]
so $\overline{\gamma}_l \in L^1(\nu_{k,l})$. Moreover, 
\[
G_\gamma (\mathbf r) = \sum_{l=1}^L \mathbb P \{W=w_l\} \sum_{k=1}^K s_{k,l} \int r_{k,l} (\overline{\gamma}_l - \gamma_{k,l}) d \nu_{k,l}, 
\]
so $G_\gamma$ is a continuous linear functional in the product weak-$\ast$ topology. Hence, $\mathcal F$ is convex and weak-$\ast$ compact, and the infimum and supremum in~(\ref{eq:constrained ot bounds}) are attained with
\begin{equation} \label{eq:JF}
    J(\mathcal F) = [LB,UB].
\end{equation}

\emph{Step 4: Lagrangian duality.} For~$\lambda = (\lambda_{k,l})_{k,l} \in \mathbb R_+^{KL}$ and~$\gamma \in \Gamma$, define 
\begin{eqnarray*}
    &\Phi_+ (\mathbf r, \lambda, \gamma) := J(\mathbf r) + G_\gamma (\mathbf r) + \sum_{l=1}^L \mathbb P \{W=w_l\} \sum_{k=1}^K \lambda_{k,l} C_{k,l}(r_{k,l}), \\ 
    &\Phi_- (\mathbf r, \lambda, \gamma) := J(\mathbf r) + G_\gamma (\mathbf r) - \sum_{l=1}^L \mathbb P \{W=w_l\} \sum_{k=1}^K \lambda_{k,l} C_{k,l}(r_{k,l}).
\end{eqnarray*}
For given~$(\lambda,\gamma),$ these functions are affine and continuous in~$\mathbf r$ under the product weak-$\ast$ topology. For given~$\mathbf r$, they are also affine and continuous in~$(\lambda,\gamma)$, where the continuity in~$\gamma$ follows from~$| \overline{r}_l - r_{k,l} | \le 1$. 

For a fixed~$\mathbf r$, the supremum of~$\Phi_+$ over~$(\lambda,\gamma)$ equals~$J(\mathbf r)$ if~$\mathbf r \in \mathcal F$, and~$+\infty$ otherwise. More specifically, a violation of dependence restriction can be detected by taking~$\lambda_{k,l} \to \infty$ and a violation of exclusion restriction can be detected by taking~$\gamma_{k,l} = t \, \mbox{sgn} (\overline{r}_l - r_{k,l})$ and~$t \to \infty$. Hence,~$LB = \inf_{\mathbf r \in \mathcal R} \sup_{\lambda \in \mathbb R_+^{KL}, \gamma \in \Gamma} \Phi_+(\mathbf r, \lambda, \gamma)$. Applying Sion's minimax theorem to compact convex set $\mathcal R$ and the convex set $\mathbb R_+^{KL} \times \Gamma$ therefore yields 
\begin{eqnarray}\label{eq:sion lower}
    LB = \sup_{\lambda \in \mathbb R_+^{KL}, \gamma \in \Gamma} \inf_{\mathbf r \in \mathcal R} \ \Phi_+(\mathbf r, \lambda, \gamma). 
\end{eqnarray}

Likewise, the infimum of~$\Phi_-$ equals~$J(\mathbf r)$ if~$\mathbf r \in \mathcal F$ and~$-\infty$ otherwise, and Sion's theorem gives 
\begin{eqnarray}\label{eq:sion upper}
    UB = \inf_{\lambda \in \mathbb R_+^{KL}, \gamma \in \Gamma} \sup_{\mathbf r \in \mathcal R} \ \Phi_-(\mathbf r, \lambda, \gamma). 
\end{eqnarray}

\emph{Step 5: Kantorovich duality for the inner problems.} For any given $(\lambda,\gamma) \in \mathbb R_+^{KL} \times \Gamma$, define 
\[
g^+_{\lambda,\gamma,k,l} := h_l + \overline{\gamma}_l - \gamma_{k,l} + \lambda_{k,l} ( \zeta_{1,k,l} - \zeta_{0,k,l} ). 
\]
Then, $g^+_{\lambda,\gamma,k,l} \in L^1(\nu_{k,l})$ and 
\[
\Phi_+(\mathbf r, \lambda, \gamma) = \sum_{l=1}^L \mathbb P\{W=w_l\} \sum_{k=1}^K s_{k,l} \left\{ \lambda_{k,l} \int \zeta_{0,k,l} d \nu_{k,l} + \int r_{k,l} g^+_{\lambda,\gamma,k,l} d \nu_{k,l} \right\}. 
\]
Then, part~(i) of lemma~\ref{lem:duality} gives 
\[
\begin{aligned}
&\inf_{\mathbf r \in \mathcal R} \ \Phi_+(\mathbf r, \lambda, \gamma) \\ 
&= \sum_{l=1}^L \mathbb P\{W=w_l\} \sum_{k=1}^K s_{k,l} \left\{ \lambda_{k,l} \int \zeta_{0,k,l} d \nu_{k,l} + \inf_{\mathbf r_{k,l} \in \mathcal R_{q_{k,l}}} \int r_{k,l} g^+_{\lambda,\gamma,k,l} d \nu_{k,l} \right\} \\ 
&= \sum_{l=1}^L \mathbb P\{W=w_l\} \sum_{k=1}^K s_{k,l} \left\{ \lambda_{k,l} \int \zeta_{0,k,l} d \nu_{k,l} + \sup_{\varphi_{k,l} \in \mathbb R} \left\{ \int \min \left(0,g^+_{\lambda,\gamma,k,l}-\varphi_{k,l} \right) d \nu_{k,l} + \varphi_{k,l} q_{k,l} \right\} \right\} \\ 
&= \sum_{l=1}^L \mathbb P\{W=w_l\} \ \sup_{\varphi_l \in \mathbb R^K} \sum_{k=1}^K s_{k,l} \left\{ \lambda_{k,l} \int \zeta_{0,k,l} d \nu_{k,l} + \int \min \left(0,g^+_{\lambda,\gamma,k,l}-\varphi_{k,l} \right) d \nu_{k,l} + \varphi_{k,l} q_{k,l} \right\}, 
\end{aligned}
\]
so plugging this into~(\ref{eq:sion lower}) and rearranging terms give us~$LB = \sum_{l=1}^L \mathbb P\{W=w_l\} \; LB_l$, where $\Gamma_l = \prod_{k=1}^K L^1(\nu_{k,l})$ and 
\[
\begin{aligned}
    LB_l = &\sup_{\varphi_l \in \mathbb R^K} \sup_{\lambda_l \in \mathbb R_+^K} \sup_{\gamma_l \in \Gamma_l} \sum_{k=1}^K s_{k,l} \bigg\{ \int \min \left(\varphi_{k,l} + \lambda_{k,l} \zeta_{0,k,l},h_l + \overline{\gamma}_l - \gamma_{k,l} + \lambda_{k,l} \zeta_{1,k,l} \right) d\nu_{k,l} \\ 
    &\hskip 320pt - \varphi_{k,l} (1-q_{k,l}) \bigg\}. 
\end{aligned}
\]

Similarly, we can rewrite~(\ref{eq:sion upper}) as~$UB = \sum_{l=1}^L \mathbb P\{W=w_l\} \; UB_l$, where
\[
\begin{aligned}
    UB_l = &\inf_{\varphi_l \in \mathbb R^K} \inf_{\lambda_l \in \mathbb R_+^K} \inf_{\gamma_l \in \Gamma_l} \sum_{k=1}^K s_{k,l} \bigg\{ \int \max \left(\varphi_{k,l} - \lambda_{k,l} \zeta_{0,k,l},h_l + \overline{\gamma}_l - \gamma_{k,l} - \lambda_{k,l} \zeta_{1,k,l} \right) d\nu_{k,l} \\ 
    &\hskip 320pt - \varphi_{k,l} (1-q_{k,l}) \bigg\}. 
\end{aligned}
\]

\emph{Step 6: Optimization over the equality multipliers.} Fix $l \in \{1,2,\ldots L\}$. Letting $f_{k,l} (\mathbf p) := f_{\mathbf P \vert X=x, Z=z_k, W=w_l}(\mathbf p)$, we have $\sum_{k=1}^K s_{k,l} f_{k,l} (\overline{\gamma}_l - \gamma_{k,l}) = 0$, so 
\[
\begin{aligned}
    A_{l,\varphi_l,\lambda_l}(\gamma_l) &:= \sum_{k=1}^K s_{k,l} \int \min \left(\varphi_{k,l} + \lambda_{k,l} \zeta_{0,k,l},h_l + \overline{\gamma}_l - \gamma_{k,l} + \lambda_{k,l} \zeta_{1,k,l} \right) d\nu_{k,l} \\ 
    &\le \int \min \left(\sum_{k=1}^K s_{k,l} f_{k,l} (\varphi_{k,l} + \lambda_{k,l} \zeta_{0,k,l}), f_{\mathbf P \vert W=w_l} + \sum_{k=1}^K s_{k,l} f_{k,l} \lambda_{k,l} \zeta_{1,k,l} \right) d\varrho \\ 
    &=: A_{l,\varphi_l,\lambda_l}^*. 
\end{aligned}
\]

Also, let~$\gamma_l^* = (\gamma_{k,l}^*)_{k} \in \Gamma_l$ with~$\gamma_{k,l}^* := - \varphi_{k,l} + \lambda_{k,l} \left( \zeta_{1,k,l} - \zeta_{0,k,l}\right)$ for each~$k=1,2,\ldots K$. For this choice, the second argument of the cellwise minimum exceeds the first by~$h_l+\overline{\gamma}^*_l$, which does not depend on~$k$. For all real~$a_1,\ldots,a_K$ and~$c$, and for the above~$\gamma_l^*$, 
\begin{eqnarray*}
    \sum_{k=1}^K s_{k,l} f_{k,l} \min(a_k, a_k+c) & = & \min\Big(\sum_{k=1}^K s_{k,l} f_{k,l} a_k, \; \sum_{k=1}^K s_{k,l} f_{k,l} a_k + c\, f_{\mathbf P \vert X=x, W=w_l}\Big), \\
    f_{\mathbf P \vert X=x, W=w_l}\,(h_l+\overline{\gamma}^*_l) & = & f_{\mathbf P \vert W=w_l} + \sum_{k=1}^K s_{k,l} f_{k,l} \gamma^*_{k,l} \qquad \mbox{almost everywhere}, 
\end{eqnarray*}
so that $A_{l,\varphi_l,\lambda_l}(\gamma_l^*) = A_{l,\varphi_l,\lambda_l}^*$, which means~$\sup_{\gamma_l \in \Gamma_l} A_{l,\varphi_l,\lambda_l}(\gamma_l) = A_{l,\varphi_l,\lambda_l}^*.$ This gives us $LB_l = \sup_{\varphi_l \: \in \: \mathbb R^K} \sup_{\lambda_l \: \in \: \mathbb R_+^K} L(\lambda_l, \varphi_l,x,l),$ which yields the lower bound of the theorem. For the upper bound, the same two identities with~$\max$ in place of~$\min$ give, for every~$\gamma_l \in \Gamma_l$, 
\begin{eqnarray*}
    \lefteqn{\sum_{k=1}^K s_{k,l} \int \max \left(\varphi_{k,l} - \lambda_{k,l} \zeta_{0,k,l},h_l + \overline{\gamma}_l - \gamma_{k,l} - \lambda_{k,l} \zeta_{1,k,l} \right) d\nu_{k,l}} \\
    & \ge & \int \max \Big(\sum_{k=1}^K s_{k,l} f_{k,l} (\varphi_{k,l} - \lambda_{k,l} \zeta_{0,k,l}), \; f_{\mathbf P \vert W=w_l} - \sum_{k=1}^K s_{k,l} f_{k,l} \lambda_{k,l} \zeta_{1,k,l} \Big) d\varrho, 
\end{eqnarray*}
with equality at~$\gamma_{k,l}^{*,U} := - \varphi_{k,l} - \lambda_{k,l} \left( \zeta_{1,k,l} - \zeta_{0,k,l}\right)$, $k=1,2,\ldots K$, for which the second argument of the cellwise maximum exceeds the first by~$h_l+\overline{\gamma}^{*,U}_l$. Hence $UB_l = \inf_{\varphi_l \in \mathbb R^K} \inf_{\lambda_l \in \mathbb R_+^K} U(\lambda_l, \varphi_l, x,l),$ which yields the upper bound of the theorem. 

%The integrands are measurable functions of~$w$: $U(\varphi,x,B_z,w)$ and~$L(\varphi,x,B_z,w)$ are convex, respectively concave, and finite in~$\varphi$, hence continuous, so that the optimization over~$\varphi$ may be restricted to the rationals.

\emph{Step 7: Sharpness.} We next verify that $\mathcal I(x)=J(\mathcal F)$. The sharpness of our bounds then follows from (\ref{eq:JF}). By (\ref{eq:constrained ot bounds}), $\mathcal I(x)\subseteq J(\mathcal F)$. Conversely, fix any $\mathbf r\in\mathcal F$. The exclusion restrictions imply $\overline r_l=r_{k,l}$ $\nu_{k,l}$-almost everywhere for every $k,l$.

Let $\big(\eta,\mathbf P,X,Z,W,(D(x^\prime))_{x^\prime\in\mathcal X}\big)$ be a data generating process satisfying Assumptions~\ref{ass:endogeneity'}, \ref{ass:invertibility}, \ref{ass:full support'} and~\ref{ass:copula} and reproducing the laws of $(\mathbf P,X,Z,W)$ and $(D,X,Z,W)$ identified from the two data sets; one exists by the maintained assumptions. Enlarge the probability space with a variable $V$ uniformly distributed on $(0,1)$ and independent of all the other variables of the process, and replace the potential decision $D(x)$ with 
\[ 
D_{\mathbf r}(x):=1\left\{V\le\sum_{l=1}^L1\{W=w_l\}\overline r_l(\mathbf P)\right\}, 
\]
keeping $\eta$, the reports $\mathbf P$, the observed covariates $(X,Z,W)$ and the potential decisions $D(x^\prime)$ for $x^\prime\neq x$ unchanged, and let $D_{\mathbf r}:=D_{\mathbf r}(x)$ on $\{X=x\}$ and $D_{\mathbf r}:=D$ on $\{X\neq x\}$ be the new observed decision. Since $\mathbf P$ is a measurable function of $\eta$ and $V$ is independent of $(\eta,X,Z,W)$, $D_{\mathbf r}(x)$ is independent of $(X,Z)$ conditional on $(\eta,W)$, so that Assumption~\ref{ass:endogeneity'} holds for the modified process; Assumptions~\ref{ass:invertibility} and~\ref{ass:full support'} and part~(1) of Assumption~\ref{ass:copula} hold because $\eta$, $\mathbf P$, $(X,Z,W)$ and $\zeta$ are unchanged. On $\{X\neq x\}$, the law of $(D_{\mathbf r},X,Z,W)$ and part~(2) of Assumption~\ref{ass:copula} are those of the original process. On $\{X=x\}$, since $V$ is independent of $(\mathbf P,X,Z,W)$ and $\mathbf P$ has law $\nu_{k,l}$ given $(X=x,Z=z_k,W=w_l)$, 
\[ 
\mathbb E[D_{\mathbf r}\mid X=x,Z=z_k,W=w_l]=\int\overline r_l\,d\nu_{k,l}=\int r_{k,l}\,d\nu_{k,l}=q_{k,l}, 
\] 
which, $D_{\mathbf r}$ being binary, reproduces the observed conditional distribution of $D$ in each cell $(x,z_k,w_l)$, and 
\begin{eqnarray*}
\mathbb E[\zeta(D_{\mathbf r},\mathbf P\mid x,z_k,w_l)\mid X=x,Z=z_k,W=w_l] & = & \int\left\{\overline r_l\,\zeta_{1,k,l}+(1-\overline r_l)\,\zeta_{0,k,l}\right\}d\nu_{k,l} \\
& = & \frac{1}{s_{k,l}}C_{k,l}(r_{k,l}) \; \le \; 0, 
\end{eqnarray*}
which verifies part~(2) of Assumption~\ref{ass:copula} in every cell with $X=x$. Hence the modified process satisfies all the maintained assumptions and reproduces the identified laws, and 
\[ 
\mathbb E[D_{\mathbf r}(x)]=\sum_{l=1}^L\mathbb P\{W=w_l\}\int\overline r_l\,dF_{\mathbf P\mid W=w_l}=J(\mathbf r)\in\mathcal I(x), 
\] 
where the last equality holds because $r_{k,l}=\overline r_l$ $\nu_{k,l}$-almost everywhere, by the computation in~(\ref{eq:cell transport}). Since $\mathbf r\in\mathcal F$ was chosen arbitrarily, $J(\mathcal F)\subseteq\mathcal I(x)$, which completes the proof.

\end{proof}

\subsection{Proofs for section~\ref{sec:partial}}
\label{app:partial}

\subsubsection*{Setting and notation}

Throughout this subsection, $G$ is a coarsening by a partition, as in example~\ref{ex:partition}, with cells~$B_1,\ldots,B_K$, and~$x$ is fixed. Since~$\tilde{\mathbf P}\in B_i$ if and only if~$\mathbf P\in B_i$, the cell of the latent report is observed. We write
\begin{eqnarray*}
    \nu \; := \; F_{\mathbf P\vert X=x},
    & \qquad &
    \Lambda_x(A) \; := \; \mathbb P(\mathbf P\in A,X\neq x),
\end{eqnarray*}
so that~$\Lambda_x$ is the (sub-probability) distribution of the latent reports of the agents with~$X\neq x$, with total mass~$1-\pi_x$, and we denote the cell quantities
\begin{eqnarray*}
    p_i \; := \; \mathbb P(\tilde{\mathbf P}\in B_i\vert X=x), \qquad
    \varkappa_i \; := \; \mathbb P(\tilde{\mathbf P}\in B_i,X\neq x), \qquad
    p_i^{x^\prime} \; := \; \mathbb P(\tilde{\mathbf P}\in B_i\vert X=x^\prime), \\
    s_i \; := \; \mathbb E[D\vert X=x,\tilde{\mathbf P}\in B_i], \qquad
    a_i \; := \; s_ip_i,
\end{eqnarray*}
the last two being defined when~$p_i>0$. Hence~$\nu(B_i)=p_i$, $\Lambda_x(B_i)=\varkappa_i=\sum_{x^\prime\neq x}\pi_{x^\prime}p_i^{x^\prime}$, $a_i=\mathbb P(D=1,\tilde{\mathbf P}\in B_i\vert X=x)$, $\sum_ia_i=p(x)$, $\sum_i\varkappa_i=1-\pi_x$ and~$\mu^-(x)=\pi_xp(x)$. In the notation of section~\ref{sec:coarsening}, $a_i=\kappa_i^1$ and~$p_i-a_i=\kappa_i^0$. With matched data, all these quantities are identified; with unmatched data, $p_i$, $\varkappa_i$, $p_i^{x^\prime}$, $\pi_x$ and~$p(x)$ are identified but~$s_i$ is not. If~$\pi_x=1$, then~$\mu^-(x)=\mu^+(x)=\mu(x)$ and all the results of section~\ref{sec:partial} are trivial, so we assume~$\pi_x<1$.

Under assumption~\ref{ass:full support}, $F_{\mathbf P}\ll F_{\mathbf P\vert X=x^\prime}$ for every~$x^\prime\in\mathcal X$, so that~$F_{\tilde{\mathbf P}}(B_i)=F_{\mathbf P}(B_i)>0$ implies~$p_i>0$ and~$p_i^{x^\prime}>0$ for all~$x^\prime$, hence~$\varkappa_i>0$; conversely, $p_i>0$ or~$\varkappa_i>0$ implies~$F_{\mathbf P}(B_i)>0$. Hence, under the maintained assumptions,
\begin{eqnarray}\label{eq:compatibility}
    p_i>0 \quad\Longleftrightarrow\quad \varkappa_i>0 \quad\Longleftrightarrow\quad F_{\tilde{\mathbf P}}(B_i)>0,
\end{eqnarray}
and data violating~(\ref{eq:compatibility}) are incompatible with the assumptions, in which case the identified set is empty and there is nothing to prove. We therefore assume~(\ref{eq:compatibility}) throughout. In the setting of section~\ref{sec:coarsening} we assume, as is implicit in the statements there, that~$p_i>0$ for every~$i$, so that~$0=c_0<c_1<\cdots<c_K=1$.

\begin{lemma}[Transport identity]\label{lem:transport}
    Under assumptions~\ref{ass:endogeneity}, \ref{ass:invertibility} and part~(1) of assumption~\ref{ass:full support}, $\Lambda_x\ll\nu$ and
    \begin{eqnarray}\label{eq:transport}
        \mu(x) & = & \mu^-(x)+\int m(x,\mathbf p)\;d\Lambda_x(\mathbf p), \qquad m(x,\cdot):=\mathbb E[D\vert X=x,\mathbf P=\cdot\,]\in[0,1].
    \end{eqnarray}
\end{lemma}

\begin{proof}
    Splitting on~$\{X=x\}$, where~$D=D(x)$, $\mu(x)=\mathbb E[D(x)]=\mu^-(x)+\mathbb E[D(x)1\{X\neq x\}]$. By step~2 of the proof of proposition~\ref{prop:identification ASF}, $D(x)\perp\!\!\!\!\perp X\vert\mathbf P$, so that~$\mathbb E[D(x)1\{X\neq x\}\vert\mathbf P]=\mathbb E[D(x)\vert\mathbf P]\,\mathbb P(X\neq x\vert\mathbf P)$ and~$\mathbb E[D(x)\vert\mathbf P=\mathbf p]=m(x,\mathbf p)$ for~$\nu$-almost every~$\mathbf p$. Since~$\Lambda_x(d\mathbf p)=\mathbb P(X\neq x\vert\mathbf P=\mathbf p)F_{\mathbf P}(d\mathbf p)$, integrating gives~$\mathbb E[D(x)1\{X\neq x\}]=\int m(x,\mathbf p)\,d\Lambda_x(\mathbf p)$, which is~(\ref{eq:transport}). Finally~$\Lambda_x\leq F_{\mathbf P}\ll\nu$ by part~(1) of assumption~\ref{ass:full support}, so that~$m(x,\cdot)$ is defined~$\Lambda_x$-almost everywhere.
\end{proof}

The next lemma characterizes the identified set under coarsening by a partition in terms of triples~$(\nu^\prime,\Lambda^\prime,m^\prime)$, which play the roles of~$\nu$, $\Lambda_x$ and~$m(x,\cdot)$. Let~$\mathcal A$ be the set of triples~$(\nu^\prime,\Lambda^\prime,m^\prime)$, where~$\nu^\prime$ is a probability measure on~$[0,1]^T$, $\Lambda^\prime$ is a non-negative measure on~$[0,1]^T$ with~$\Lambda^\prime\ll\nu^\prime$, and~$m^\prime:[0,1]^T\rightarrow[0,1]$ is measurable, such that, for~$i=1,\ldots,K$,
\begin{eqnarray}\label{eq:cell restrictions}
    \nu^\prime(B_i) \; = \; p_i,
    & \qquad
    \Lambda^\prime(B_i) \; = \; \varkappa_i,
    \qquad &
    \int_{B_i}m^\prime\,d\nu^\prime \; = \; a_i,
\end{eqnarray}
the last restriction being void when~$p_i=0$. In the setting of section~\ref{sec:coarsening}, where~$\mathbf P=P_x$ is continuously distributed, $\nu^\prime$ is in addition required to be atomless. With unmatched data (section~\ref{sec:coarsening unmatched}), the third restriction in~(\ref{eq:cell restrictions}) is replaced by~$\int m^\prime\,d\nu^\prime=p(x)$. Define the value
\begin{eqnarray*}
    V(\nu^\prime,\Lambda^\prime,m^\prime) & := & \mu^-(x)+\int m^\prime\,d\Lambda^\prime.
\end{eqnarray*}
For a subset~$S$ of the shape restrictions~$\{$assumption~\ref{ass:monotone m}, assumption~\ref{ass:MLR1}$\}$, with~$S=\emptyset$ outside the setting of section~\ref{sec:coarsening}, let~$\mathcal A_S$ be the set of triples in~$\mathcal A$ such that~$m^\prime$ is non-decreasing if assumption~\ref{ass:monotone m} belongs to~$S$, and~$d\Lambda^\prime/d\nu^\prime$ has a non-increasing version if assumption~\ref{ass:MLR1} belongs to~$S$. Let~$\mathcal I_S(x)$ be the identified set for~$\mu(x)$ under assumptions~\ref{ass:endogeneity}--\ref{ass:full support}, \ref{ass:coarsening} and the restrictions in~$S$.

\begin{lemma}[Identified set under coarsening by a partition]\label{lem:coarsening}
    \begin{enumerate}
        \item[(i)] Every data generating process satisfying assumptions~\ref{ass:endogeneity}--\ref{ass:full support} and~\ref{ass:coarsening} induces the triple~$(\nu,\Lambda_x,m(x,\cdot))\in\mathcal A$, with~$\mu(x)=V(\nu,\Lambda_x,m(x,\cdot))$. In the setting of section~\ref{sec:coarsening}, it satisfies assumption~\ref{ass:monotone m} if and only if~$m(x,\cdot)$ is non-decreasing, and assumption~\ref{ass:MLR1} if and only if~$d\Lambda_x/d\nu$ has a non-increasing version. Hence~$\mathcal I_S(x)\subseteq V(\mathcal A_S)$.
        \item[(ii)] Conversely, let~$(\nu^\prime,\Lambda^\prime,m^\prime)\in\mathcal A$ be such that~$\nu^\prime$ has a bounded density with respect to counting or Lebesgue measure, and~$\rho^\prime:=d\Lambda^\prime/d\nu^\prime$ is bounded and~$\nu^\prime$-almost everywhere strictly positive. Then there is a data generating process satisfying assumptions~\ref{ass:endogeneity}--\ref{ass:full support} and~\ref{ass:coarsening}, reproducing the identified features of the data, with~$F_{\mathbf P\vert X=x}=\nu^\prime$, $\Lambda_x=\Lambda^\prime$ and~$m(x,\cdot)=m^\prime$ $\nu^\prime$-almost everywhere, so that~$\mu(x)=V(\nu^\prime,\Lambda^\prime,m^\prime)$. In the setting of section~\ref{sec:coarsening}, it satisfies assumption~\ref{ass:monotone m} if~$m^\prime$ is non-decreasing and assumption~\ref{ass:MLR1} if~$\rho^\prime$ is non-increasing.
        \item[(iii)] Let~$LB\leq UB$ be such that~$V(\mathcal A_S)\subseteq[LB,UB]$, and let~$\mathcal A^\prime\subseteq\mathcal A_S$ be a convex set of triples with a common first component~$\nu^\prime$, each satisfying the hypotheses of~(ii), such that~$\inf V(\mathcal A^\prime)=LB$ and~$\sup V(\mathcal A^\prime)=UB$. Then~$\mathcal I_S(x)$ is an interval with infimum~$LB$ and supremum~$UB$; that is, $[LB,UB]$ are sharp bounds for~$\mu(x)$ under assumptions~\ref{ass:endogeneity}--\ref{ass:full support}, \ref{ass:coarsening} and~$S$.
    \end{enumerate}
    With unmatched data, the same statements hold with the third restriction in~(\ref{eq:cell restrictions}) replaced by~$\int m^\prime\,d\nu^\prime=p(x)$.
\end{lemma}

\begin{proof}
    \emph{(i)} By lemma~\ref{lem:transport}, $\Lambda_x\ll\nu$, $0\leq m(x,\cdot)\leq1$ and~$\mu(x)=V(\nu,\Lambda_x,m(x,\cdot))$. Since the cell of~$\mathbf P$ is observed, $\nu(B_i)=p_i$, $\Lambda_x(B_i)=\varkappa_i$ and~$\int_{B_i}m(x,\cdot)\,d\nu=\mathbb E[D1\{\mathbf P\in B_i\}\vert X=x]=a_i$, and with unmatched data~$\int m(x,\cdot)\,d\nu=p(x)$; in the setting of section~\ref{sec:coarsening}, $\nu$ is atomless because~$\pi_x\nu\leq F_{\mathbf P}$ and~$F_{\mathbf P}$ is atomless. Assumption~\ref{ass:monotone m} is, by definition, the monotonicity of~$m(x,\cdot)$. For assumption~\ref{ass:MLR1}, write~$\varpi(p):=\mathbb P(X=x\vert P_x=p)$, which is~$F_{\mathbf P}$-almost everywhere positive since~$F_{\mathbf P}\ll\nu$. By Bayes' rule, $d\Lambda_x/dF_{\mathbf P}=1-\varpi$ and~$d\nu/dF_{\mathbf P}=\varpi/\pi_x$, whence~$d\Lambda_x/d\nu=\pi_x(1-\varpi)/\varpi$, and~$t\mapsto\pi_x(1-t)/t$ is strictly decreasing on~$(0,1]$. Therefore~$\varpi$ has a non-decreasing version if and only if~$d\Lambda_x/d\nu$ has a non-increasing version.

    \emph{(ii)} Let~$\mathcal K:=\{i:p_i>0\}$. By~(\ref{eq:compatibility}), $\varkappa_i>0$ and~$p_i^{x^\prime}>0$ for~$i\in\mathcal K$ and every~$x^\prime$, while~$\nu^\prime(B_i)=\Lambda^\prime(B_i)=p_i^{x^\prime}=0$ for~$i\notin\mathcal K$. Construct the data generating process as follows.
    \begin{enumerate}
        \item[(a)] \emph{Types and reports.} Set~$\eta:=\mathbf P$, so that assumption~\ref{ass:invertibility} holds with~$\tilde m(x_t,\eta)=\eta_t$. Let~$\mathbb P(X=x^\prime):=\pi_{x^\prime}$ for every~$x^\prime\in\mathcal X$, $F_{\mathbf P\vert X=x}:=\nu^\prime$ and, for~$x^\prime\neq x$,
        \begin{eqnarray*}
            F_{\mathbf P\vert X=x^\prime} & := & \sum_{i\in\mathcal K}\frac{p_i^{x^\prime}}{\varkappa_i}\,\Lambda^\prime\vert_{B_i},
        \end{eqnarray*}
        which is a probability measure since~$\sum_ip_i^{x^\prime}=1$, and gives~$\mathbb P(\mathbf P\in B_i\vert X=x^\prime)=p_i^{x^\prime}$. Since~$\sum_{x^\prime\neq x}\pi_{x^\prime}p_i^{x^\prime}=\varkappa_i$, we get~$\sum_{x^\prime\neq x}\pi_{x^\prime}F_{\mathbf P\vert X=x^\prime}=\Lambda^\prime$, i.e., $\Lambda_x=\Lambda^\prime$, and~$F_{\mathbf P}=\pi_x\nu^\prime+\Lambda^\prime$.
        \item[(b)] \emph{Overlap.} $F_{\mathbf P}\ll\nu^\prime$ because~$\Lambda^\prime\ll\nu^\prime$, and for~$x^\prime\neq x$, $F_{\mathbf P\vert X=x^\prime}$ is equivalent to~$\Lambda^\prime$ on~$\bigcup_{i\in\mathcal K}B_i$, which is equivalent to~$\nu^\prime$ because~$\rho^\prime$ is bounded and~$\nu^\prime$-almost everywhere positive; hence~$F_{\mathbf P}\ll F_{\mathbf P\vert X=x^\prime}$, and part~(1) of assumption~\ref{ass:full support} holds for every value of~$X$. With~$f^\prime$ the bounded density of~$\nu^\prime$, the densities~$f_{\mathbf P\vert X=x}=f^\prime$, $f_{\mathbf P}=(\pi_x+\rho^\prime)f^\prime$ and~$f_{\mathbf P\vert X=x^\prime}=\sum_{i\in\mathcal K}(p_i^{x^\prime}/\varkappa_i)\rho^\prime f^\prime1\{\cdot\in B_i\}$ are bounded, so that part~(2) holds as well.
        \item[(c)] \emph{Potential decisions.} Let~$V,V^\prime$ be independent uniform variables on~$(0,1)$, independent of~$(\mathbf P,X)$. Set~$D(x):=1\{m^\prime(\mathbf P)\geq V\}$ and, for~$x^\prime\neq x$, $D(x^\prime):=1\{s_i^{x^\prime}\geq V^\prime\}$ on~$\{\mathbf P\in B_i\}$, where~$s_i^{x^\prime}:=\mathbb E[D\vert X=x^\prime,\tilde{\mathbf P}\in B_i]$ with matched data and~$s_i^{x^\prime}:=\mathbb P(D=1\vert X=x^\prime)$ with unmatched data. Each~$D(x^{\prime\prime})$ is a function of~$(\mathbf P,V,V^\prime)$, so assumption~\ref{ass:endogeneity} holds for every~$x^{\prime\prime}\in\mathcal X$. Let~$D:=D(X)$. Then~$\mathbb E[D\vert X=x,\mathbf P=\mathbf p]=m^\prime(\mathbf p)$, so that~$m(x,\cdot)=m^\prime$ $\nu^\prime$-almost everywhere, $\mathbb E[D\vert X=x,\mathbf P\in B_i]=a_i/p_i=s_i$ for~$i\in\mathcal K$ with matched data, $\mathbb E[D\vert X=x]=\int m^\prime\,d\nu^\prime=p(x)$ in both cases, and~$\mathbb E[D\vert X=x^\prime,\mathbf P\in B_i]=s_i^{x^\prime}$.
        \item[(d)] \emph{Observed reports.} Conditionally on~$(D,X,\mathbf P)$, draw~$\tilde{\mathbf P}$ from the identified conditional law of~$\tilde{\mathbf P}$ given~$D$, $X$ and the event~$\{\tilde{\mathbf P}\in B_i\}$, $i$ being the cell of~$\mathbf P$ (given~$X$ and the cell only, with unmatched data). This law is carried by~$B_i$, so that~$\tilde{\mathbf P}\in G(\mathbf P)$ and assumption~\ref{ass:coarsening} holds. By~(c), the law of~$(D,X,1\{\mathbf P\in B_i\}_i)$ coincides with the identified law of~$(D,X,1\{\tilde{\mathbf P}\in B_i\}_i)$ with matched data, and the laws of~$(D,X)$ and of~$(X,1\{\mathbf P\in B_i\}_i)$ with the identified laws of~$(D,X)$ and~$(X,1\{\tilde{\mathbf P}\in B_i\}_i)$ with unmatched data; the identified features of the data are therefore reproduced.
    \end{enumerate}
    By lemma~\ref{lem:transport}, $\mu(x)=\mu^-(x)+\int m^\prime\,d\Lambda^\prime=V(\nu^\prime,\Lambda^\prime,m^\prime)$. Finally, in the setting of section~\ref{sec:coarsening}, $\mathbb E[D\vert X=x,P_x=p]=m^\prime(p)$ and~$\mathbb P(X=x\vert P_x=p)=\pi_x/(\pi_x+\rho^\prime(p))$, so that assumption~\ref{ass:monotone m} holds if~$m^\prime$ is non-decreasing and assumption~\ref{ass:MLR1} holds if~$\rho^\prime$ is non-increasing.

    \emph{(iii)} By~(ii), $V(\mathcal A^\prime)\subseteq\mathcal I_S(x)$, and by~(i), $\mathcal I_S(x)\subseteq V(\mathcal A_S)\subseteq[LB,UB]$. The set~$V(\mathcal A^\prime)$ is an interval: for two elements~$(\nu^\prime,\Lambda_0,m_0)$ and~$(\nu^\prime,\Lambda_1,m_1)$ of~$\mathcal A^\prime$, the segment~$t\mapsto(\nu^\prime,(1-t)\Lambda_0+t\Lambda_1,(1-t)m_0+tm_1)$, $t\in[0,1]$, lies in~$\mathcal A^\prime$ by convexity, and the value along it is a polynomial of degree at most two in~$t$, hence continuous, so that every number between the two endpoint values is attained by the intermediate value theorem. Since~$\inf V(\mathcal A^\prime)=LB$ and~$\sup V(\mathcal A^\prime)=UB$, we get~$(LB,UB)\subseteq V(\mathcal A^\prime)\subseteq\mathcal I_S(x)\subseteq[LB,UB]$, which proves the claim.
\end{proof}

\begin{proof}[Proof of theorem~\ref{thm:impossibility}]
    Let~$\mathcal K:=\{i:p_i>0\}=\{i:\varkappa_i>0\}$ by~(\ref{eq:compatibility}). In the setting of example~\ref{ex:partition}, the set~$G^{-1}(\tilde{\mathbf P})$ is the cell containing~$\tilde{\mathbf P}$, so that~$\mathbb E[D\vert X=x,G^{-1}(\tilde{\mathbf P})]=s_i$ on~$\{\tilde{\mathbf P}\in B_i\}$, $i\in\mathcal K$, and the condition of the theorem reads
    \begin{eqnarray}\label{eq:unanimous}
        \sum_{i\in\mathcal K:\,s_i\in\{0,1\}}\varkappa_i \; = \; 0,
        \qquad\mbox{i.e.,}\qquad s_i\in(0,1) \mbox{ for every } i\in\mathcal K.
    \end{eqnarray}
    We show that~$\mathcal I(x)$ is an interval with
    \begin{eqnarray}\label{eq:worst case sharp}
        \inf\mathcal I(x) \; = \; \mu^-(x)+\sum_{i\in\mathcal K:\,s_i=1}\varkappa_i \; =: \; LB
        & \quad\mbox{and}\quad &
        \sup\mathcal I(x) \; = \; \mu^+(x)-\sum_{i\in\mathcal K:\,s_i=0}\varkappa_i \; =: \; UB.
    \end{eqnarray}
    Since~$\varkappa_i>0$ on~$\mathcal K$, the worst case bounds~$[\mu^-(x),\mu^+(x)]$ are then sharp if and only if both sums in~(\ref{eq:worst case sharp}) vanish, i.e., if and only if~(\ref{eq:unanimous}) holds, which is the claim.

    \emph{Validity.} Let~$(\nu^\prime,\Lambda^\prime,m^\prime)\in\mathcal A$. Since~$0\leq m^\prime\leq1$ and~$\Lambda^\prime(B_i)=\varkappa_i$, each cell contribution satisfies~$0\leq\int_{B_i}m^\prime\,d\Lambda^\prime\leq\varkappa_i$. If~$s_i=1$, then~$\int_{B_i}m^\prime\,d\nu^\prime=p_i=\nu^\prime(B_i)$ with~$m^\prime\leq1$ forces~$m^\prime=1$ $\nu^\prime$-almost everywhere on~$B_i$, hence~$\Lambda^\prime$-almost everywhere since~$\Lambda^\prime\ll\nu^\prime$, and the contribution is~$\varkappa_i$; if~$s_i=0$, it is~$0$ by the same argument. Summing over cells and using~$\sum_{i\in\mathcal K}\varkappa_i=1-\pi_x$ and~$\mu^-(x)+1-\pi_x=\mu^+(x)$,
    \begin{eqnarray*}
        LB \; \leq \; V(\nu^\prime,\Lambda^\prime,m^\prime) \; \leq \; \mu^-(x)+\sum_{i\in\mathcal K:\,s_i>0}\varkappa_i \; = \; UB.
    \end{eqnarray*}
    Hence~$V(\mathcal A)\subseteq[LB,UB]$.

    \emph{Approach.} By assumption, $G^{-1}(\tilde{\mathbf p})$ is not a singleton for~$F_{\tilde{\mathbf P}}$-almost every~$\tilde{\mathbf p}$, so that by~(\ref{eq:compatibility}) each cell~$B_i$, $i\in\mathcal K$, contains two distinct points~$\mathbf p_i^0\neq\mathbf p_i^1$. Define the discrete probability measure and the function
    \begin{eqnarray*}
        \nu^\sharp \; := \; \sum_{i\in\mathcal K}p_i\big((1-s_i)\delta_{\mathbf p_i^0}+s_i\delta_{\mathbf p_i^1}\big),
        & \qquad &
        m^\sharp \; := \; \sum_{i\in\mathcal K}1\{\cdot=\mathbf p_i^1\},
    \end{eqnarray*}
    so that~$\nu^\sharp(B_i)=p_i$, $\int_{B_i}m^\sharp\,d\nu^\sharp=s_ip_i=a_i$, and~$\nu^\sharp$ has a density bounded by one with respect to counting measure. For a vector~$\alpha=(\alpha_i)_{i\in\mathcal K}$ with~$\alpha_i\in(0,1)$ if~$s_i\in(0,1)$, $\alpha_i=1$ if~$s_i=0$ and~$\alpha_i=0$ if~$s_i=1$, let
    \begin{eqnarray*}
        \Lambda_\alpha & := & \sum_{i\in\mathcal K}\varkappa_i\big(\alpha_i\delta_{\mathbf p_i^0}+(1-\alpha_i)\delta_{\mathbf p_i^1}\big).
    \end{eqnarray*}
    Then~$\Lambda_\alpha(B_i)=\varkappa_i$, $\Lambda_\alpha\ll\nu^\sharp$, and~$d\Lambda_\alpha/d\nu^\sharp$ is bounded and strictly positive on the support of~$\nu^\sharp$: when~$s_i\in(0,1)$ both~$\mathbf p_i^0$ and~$\mathbf p_i^1$ are charged by both measures, and when~$s_i\in\{0,1\}$ both measures are carried by the same single point of~$B_i$. The set~$\mathcal A^\prime$ of triples~$(\nu^\sharp,\Lambda_\alpha,m^\sharp)$ is convex, is contained in~$\mathcal A$, and each of its elements satisfies the hypotheses of part~(ii) of lemma~\ref{lem:coarsening}. Its values are
    \begin{eqnarray*}
        V(\nu^\sharp,\Lambda_\alpha,m^\sharp) & = & \mu^-(x)+\sum_{i\in\mathcal K:\,s_i=1}\varkappa_i+\sum_{i\in\mathcal K:\,s_i\in(0,1)}\varkappa_i(1-\alpha_i),
    \end{eqnarray*}
    whose infimum over~$\alpha$ is~$LB$ and whose supremum is~$LB+\sum_{i\in\mathcal K:\,s_i\in(0,1)}\varkappa_i=UB$. Part~(iii) of lemma~\ref{lem:coarsening} now yields~(\ref{eq:worst case sharp}) and the interval property.
\end{proof}

\subsubsection*{Quantile coordinates for the scalar case}

We now place ourselves in the setting of section~\ref{sec:coarsening}: $T=1$, $\mathbf P=P_x$ is continuously distributed, the cells~$B_1<\cdots<B_K$ are non-degenerate intervals partitioning~$[0,1]$, and~$p_i>0$ for every~$i$. Recall~$c_i=\sum_{l\leq i}p_l$ and the concentration curve~$L$ of section~\ref{sec:coarsening matched}. Let~$\mathcal L_0$ be the set of absolutely continuous non-decreasing functions~$L:[0,1]\rightarrow[0,1-\pi_x]$ with~$L(0)=0$ and~$L(c_i)=\sum_{l\leq i}\varkappa_l$ for~$i=1,\ldots,K$, i.e., satisfying~(\ref{eq:concentration}). A concave continuous function on~$[0,1]$ is absolutely continuous, since it is locally Lipschitz on~$(0,1)$ and continuous at the endpoints, so that~$\mathcal L=\{L\in\mathcal L_0:L\mbox{ concave}\}$. Let~$\mathcal M_0$ be the set of measurable~$\tilde m:[0,1]\rightarrow[0,1]$ with~$\int_{c_{i-1}}^{c_i}\tilde m(u)\,du=a_i$ for~$i=1,\ldots,K$ (matched data), and~$\mathcal M_u$ the set of measurable~$\tilde m:[0,1]\rightarrow[0,1]$ with~$\int_0^1\tilde m(u)\,du=p(x)$ (unmatched data). For a non-decreasing absolutely continuous~$L$, $dL$ denotes the measure with distribution function~$L$, and~$\int\tilde m\,dL=\int\tilde m(u)L^\prime(u)\,du$. Finally, $\nu_0:=\sum_ip_i\,\mathcal U(B_i)$, where~$\mathcal U(B_i)$ is the uniform distribution on~$B_i$, denotes the canonical atomless measure with cell masses~$p_i$; its distribution function~$F_0$ is continuous, piecewise linear and strictly increasing on~$[0,1]$, with inverse~$Q_0$, and its Lebesgue density is bounded.

\begin{lemma}[Quantile coordinates]\label{lem:quantile}
    \begin{enumerate}
        \item[(a)] Let~$(\nu^\prime,\Lambda^\prime,m^\prime)\in\mathcal A$ with~$\nu^\prime$ atomless, with distribution function~$F^\prime$ and quantile function~$Q^\prime$, and define~$L(v):=\Lambda^\prime(\{p:F^\prime(p)\leq v\})$ and~$\tilde m:=m^\prime\circ Q^\prime$. Then~$(L,\tilde m)\in\mathcal L_0\times\mathcal M_0$ (respectively~$\mathcal L_0\times\mathcal M_u$ with unmatched data) and~$\int m^\prime\,d\Lambda^\prime=\int_0^1\tilde m\,dL$. If~$m^\prime$ is non-decreasing, then~$\tilde m$ is non-decreasing; if~$d\Lambda^\prime/d\nu^\prime$ has a non-increasing version, then~$L$ is concave.
        \item[(b)] Conversely, let~$(L,\tilde m)\in\mathcal L_0\times\mathcal M_0$ (respectively~$\mathcal L_0\times\mathcal M_u$), let~$\ell\geq0$ be a version of~$L^\prime$, and define~$m^\prime:=\tilde m\circ F_0$ and~$\Lambda^\prime:=(\ell\circ F_0)\,\nu_0$. Then~$(\nu_0,\Lambda^\prime,m^\prime)\in\mathcal A$, it induces~$(L,\tilde m)$ in the sense of~(a), and~$\int m^\prime\,d\Lambda^\prime=\int_0^1\tilde m\,dL$. If~$\tilde m$ is non-decreasing, so is~$m^\prime$; if~$L$ is concave, $\ell$ may be taken non-increasing and~$d\Lambda^\prime/d\nu_0=\ell\circ F_0$ is then non-increasing. If~$\ell$ is bounded and bounded away from zero, $(\nu_0,\Lambda^\prime,m^\prime)$ satisfies the hypotheses of part~(ii) of lemma~\ref{lem:coarsening}.
        \item[(c)] For every~$L\in\mathcal L$ and every measurable~$\tilde m:[0,1]\rightarrow[0,1]$, there are~$L_n\in\mathcal L$, each with a derivative bounded and bounded away from zero, such that~$\int\tilde m\,dL_n\rightarrow\int\tilde m\,dL$.
        \item[(d)] Let~$0\leq\alpha<\beta\leq1$, let~$\ell:(\alpha,\beta]\rightarrow[0,\infty)$ be non-increasing and integrable, and let~$\tilde m:(\alpha,\beta]\rightarrow[0,1]$ be measurable with~$\int_\alpha^\beta\tilde m(u)\,du=a$. Then
        \begin{eqnarray}\label{eq:bathtub}
            \int_{\beta-a}^\beta\ell(u)\,du \; \leq \; \int_\alpha^\beta\tilde m(u)\ell(u)\,du \; \leq \; \int_\alpha^{\alpha+a}\ell(u)\,du,
        \end{eqnarray}
        with equality on the left when~$\tilde m=1\{u>\beta-a\}$ and on the right when~$\tilde m=1\{u\leq\alpha+a\}$. If moreover~$\tilde m$ is non-decreasing, then
        \begin{eqnarray}\label{eq:chebyshev}
            \int_\alpha^\beta\tilde m(u)\ell(u)\,du & \leq & \frac{1}{\beta-\alpha}\int_\alpha^\beta\tilde m(u)\,du\int_\alpha^\beta\ell(u)\,du,
        \end{eqnarray}
        with equality when~$\tilde m$ is constant.
        \item[(e)] Let~$S\subseteq\{$assumption~\ref{ass:monotone m}, assumption~\ref{ass:MLR1}$\}$, and let~$\mathcal P_S$ be the set of pairs~$(L,\tilde m)\in\mathcal L_0\times\mathcal M_0$ (respectively~$\mathcal L_0\times\mathcal M_u$) such that~$\tilde m$ is non-decreasing if assumption~\ref{ass:monotone m} belongs to~$S$ and~$L$ is concave if assumption~\ref{ass:MLR1} belongs to~$S$. Suppose that~$\mu^-(x)+\int\tilde m\,dL\in[LB,UB]$ for every~$(L,\tilde m)\in\mathcal P_S$, and that there is a convex set~$\mathcal P^\prime\subseteq\mathcal P_S$ of pairs whose first components have derivatives bounded and bounded away from zero, with~$\inf_{\mathcal P^\prime}\{\mu^-(x)+\int\tilde m\,dL\}=LB$ and~$\sup_{\mathcal P^\prime}\{\mu^-(x)+\int\tilde m\,dL\}=UB$. Then~$[LB,UB]$ are sharp bounds for~$\mu(x)$ under assumptions~\ref{ass:endogeneity}--\ref{ass:full support}, \ref{ass:coarsening} and~$S$, and the identified set is an interval.
    \end{enumerate}
\end{lemma}

\begin{proof}
    \emph{(a)} Since~$\nu^\prime$ is atomless, $F^\prime$ is continuous, the image of~$\nu^\prime$ under~$F^\prime$ is the Lebesgue measure on~$[0,1]$, and~$Q^\prime(F^\prime(p))=p$ for~$\nu^\prime$-almost every~$p$, the exceptional set being contained in the union of the flat pieces of~$F^\prime$, which is~$\nu^\prime$-null. Hence~$m^\prime=\tilde m\circ F^\prime$ $\nu^\prime$-almost everywhere, and~$\Lambda^\prime$-almost everywhere since~$\Lambda^\prime\ll\nu^\prime$. For~$p\in B_i$ we have~$c_{i-1}\leq F^\prime(p)\leq c_i$ because~$B_1\cup\cdots\cup B_{i-1}\subseteq[0,p)$ and~$[0,p]\subseteq B_1\cup\cdots\cup B_i$; since~$\nu^\prime(F^\prime=c)=0$ for every~$c$, it follows that~$1\{p\in B_i\}=1\{F^\prime(p)\in(c_{i-1},c_i]\}$ for~$\nu^\prime$-almost every, hence~$\Lambda^\prime$-almost every, $p$. The image of~$\Lambda^\prime$ under~$F^\prime$ is the measure~$dL$: it is absolutely continuous with respect to Lebesgue measure because~$\Lambda^\prime\ll\nu^\prime$ and the image of~$\nu^\prime$ is Lebesgue measure, it has total mass~$\Lambda^\prime([0,1])=1-\pi_x$, and~$L(0)=\Lambda^\prime(F^\prime=0)=0$. Thus~$L$ is absolutely continuous, non-decreasing, with values in~$[0,1-\pi_x]$, and~$L(c_i)-L(c_{i-1})=\Lambda^\prime(B_i)=\varkappa_i$, so that~$L\in\mathcal L_0$. By the change of variables~$u=F^\prime(p)$,
    \begin{eqnarray*}
        \int_{c_{i-1}}^{c_i}\tilde m(u)\,du \; = \; \int_{B_i}m^\prime\,d\nu^\prime \; = \; a_i,
        & \quad &
        \int_0^1\tilde m(u)\,du \; = \; \int m^\prime\,d\nu^\prime,
        \quad
        \int_0^1\tilde m\,dL \; = \; \int m^\prime\,d\Lambda^\prime,
    \end{eqnarray*}
    so that~$\tilde m\in\mathcal M_0$ (respectively~$\mathcal M_u$). If~$m^\prime$ is non-decreasing, $\tilde m=m^\prime\circ Q^\prime$ is non-decreasing because~$Q^\prime$ is. If~$\rho^\prime$ is a non-increasing version of~$d\Lambda^\prime/d\nu^\prime$, then~$\rho^\prime=\rho^\prime\circ Q^\prime\circ F^\prime$ $\nu^\prime$-almost everywhere and~$L(v)=\int1\{F^\prime(p)\leq v\}\rho^\prime(p)\,d\nu^\prime(p)=\int_0^v\rho^\prime(Q^\prime(u))\,du$, so that~$L$ has the non-increasing derivative~$\rho^\prime\circ Q^\prime$ and is concave.

    \emph{(b)} Since~$F_0$ is continuous and strictly increasing, $Q_0=F_0^{-1}$ and~$Q_0\circ F_0$ is the identity, so that~$\nu_0$ is atomless and the pair induced by~$(\nu_0,\Lambda^\prime,m^\prime)$ in~(a) is~$(\Lambda^\prime(F_0\leq\cdot),\tilde m\circ F_0\circ Q_0)=(\int_0^\cdot\ell(u)\,du,\tilde m)=(L,\tilde m)$, where we used the change of variables~$u=F_0(p)$ and the absolute continuity of~$L$. The restrictions~(\ref{eq:cell restrictions}) then follow from~(a): $\nu_0(B_i)=p_i$, $\Lambda^\prime(B_i)=L(c_i)-L(c_{i-1})=\varkappa_i$ and~$\int_{B_i}m^\prime\,d\nu_0=\int_{c_{i-1}}^{c_i}\tilde m(u)\,du=a_i$ (respectively~$\int m^\prime\,d\nu_0=p(x)$); $\Lambda^\prime\ll\nu_0$ by construction. The monotonicity claims hold because~$F_0$ is non-decreasing and, when~$L$ is concave, its right derivative is a non-increasing version of~$L^\prime$. The last claim follows from~$d\Lambda^\prime/d\nu_0=\ell\circ F_0$ and the boundedness of the Lebesgue density of~$\nu_0$.

    \emph{(c)} Let~$L_\ast\in\mathcal L$ be the piecewise linear interpolant of the points~$(c_i,L(c_i))$, $i=0,\ldots,K$: its slopes~$\sigma_i:=\varkappa_i/p_i$ are positive by~(\ref{eq:compatibility}) and non-increasing in~$i$, being the chord slopes of the concave function~$L$ over consecutive intervals. For~$\delta\in(0,c_1]$, let~$L^{(\delta)}$ coincide with~$L$ on~$[\delta,1]$ and be linear on~$[0,\delta]$ with~$L^{(\delta)}(0)=0$. Since~$L(\delta)/\delta\geq L^\prime_+(\delta)$ for a concave~$L$ with~$L(0)=0$, $L^{(\delta)}$ is concave; it is continuous, non-decreasing, takes the same values as~$L$ at~$c_1,\ldots,c_K$, and its derivative is bounded by~$L(\delta)/\delta$. Set~$L_n:=(1-1/n)L^{(\delta_n)}+(1/n)L_\ast$ with~$\delta_n:=\min(c_1,1/n)$. By convexity of~$\mathcal L$, $L_n\in\mathcal L$, and its derivative is bounded and bounded below by~$\min_i\sigma_i/n>0$. Since~$dL^{(\delta)}$ and~$dL$ coincide on~$(\delta,1]$ and each has mass~$L(\delta)$ on~$[0,\delta]$, and since~$0\leq\tilde m\leq1$,
    \begin{eqnarray*}
        \left\vert\int\tilde m\,dL_n-\int\tilde m\,dL\right\vert
        & \leq & \frac{1}{n}\left\vert\int\tilde m\,d(L_\ast-L^{(\delta_n)})\right\vert+\left\vert\int\tilde m\,d(L^{(\delta_n)}-L)\right\vert \\
        & \leq & \frac{2(1-\pi_x)}{n}+2L(\delta_n) \; \longrightarrow \; 0.
    \end{eqnarray*}

    \emph{(d)} The cases~$a=0$ and~$a=\beta-\alpha$ are trivial, so let~$0<a<\beta-\alpha$ and~$\ell_\ast:=\ell(\beta-a)$, which is finite. Since~$\ell$ is non-increasing, $\ell-\ell_\ast\geq0$ on~$(\alpha,\beta-a]$ and~$\ell-\ell_\ast\leq0$ on~$(\beta-a,\beta]$; since~$\int_\alpha^\beta(\tilde m-1\{u>\beta-a\})\,du=0$,
    \begin{eqnarray*}
        \int_\alpha^\beta\big(\tilde m-1\{u>\beta-a\}\big)\ell\,du
        & = & \int_\alpha^{\beta-a}\tilde m\,(\ell-\ell_\ast)\,du+\int_{\beta-a}^\beta(\tilde m-1)(\ell-\ell_\ast)\,du \; \geq \; 0,
    \end{eqnarray*}
    which is the left inequality in~(\ref{eq:bathtub}). The right inequality follows in the same way with~$\ell_\ast:=\ell(\alpha+a)$, from~$\int_\alpha^\beta(1\{u\leq\alpha+a\}-\tilde m)\ell\,du=\int_\alpha^{\alpha+a}(1-\tilde m)(\ell-\ell_\ast)\,du-\int_{\alpha+a}^\beta\tilde m\,(\ell-\ell_\ast)\,du\geq0$. For~(\ref{eq:chebyshev}), if~$\tilde m$ is non-decreasing and~$\ell$ non-increasing then~$(\tilde m(u)-\tilde m(v))(\ell(u)-\ell(v))\leq0$ for all~$u,v$; integrating over~$(\alpha,\beta]^2$ gives~$2(\beta-\alpha)\int_\alpha^\beta\tilde m\ell\,du-2\int_\alpha^\beta\tilde m\,du\int_\alpha^\beta\ell\,du\leq0$.

    \emph{(e)} By part~(i) of lemma~\ref{lem:coarsening} and~(a), every element of~$\mathcal A_S$ induces a pair in~$\mathcal P_S$ with the same value, so that~$V(\mathcal A_S)\subseteq[LB,UB]$. Let~$\mathcal A^\prime$ be the image of~$\mathcal P^\prime$ under the map~$(L,\tilde m)\mapsto(\nu_0,(\ell\circ F_0)\nu_0,\tilde m\circ F_0)$ of~(b), where~$\ell$ is the right derivative of~$L$ when~$L$ is concave and any version of~$L^\prime$ otherwise, chosen so that the map is linear. By~(b), $\mathcal A^\prime$ is a convex subset of~$\mathcal A_S$ with common first component~$\nu_0$, each of its elements satisfies the hypotheses of part~(ii) of lemma~\ref{lem:coarsening}, and~$V(\mathcal A^\prime)=\{\mu^-(x)+\int\tilde m\,dL:(L,\tilde m)\in\mathcal P^\prime\}$ has infimum~$LB$ and supremum~$UB$. Part~(iii) of lemma~\ref{lem:coarsening} concludes.
\end{proof}

\begin{proof}[Proof of proposition~\ref{prop:coarsened monotone}]
    Here~$S=\{$assumption~\ref{ass:monotone m}$\}$, and by part~(i) of lemma~\ref{lem:coarsening} and part~(a) of lemma~\ref{lem:quantile}, every admissible data generating process induces a pair~$(L,\tilde m)\in\mathcal L_0\times\mathcal M_0$ with~$\tilde m$ non-decreasing and~$\mu(x)=\mu^-(x)+\int\tilde m\,dL$. Write~$\sigma_i:=\varkappa_i/p_i$ and, for such a pair, let~$\ell_i$ and~$t_i$ be the essential infimum and supremum of~$\tilde m$ over the cell~$(c_{i-1},c_i]$ with respect to Lebesgue measure.

    \emph{Step 1: validity.} Since~$dL$ is absolutely continuous with respect to Lebesgue measure and has mass~$\varkappa_i$ on~$(c_{i-1},c_i]$,
    \begin{eqnarray}\label{eq:cell contribution}
        \varkappa_i\ell_i \; \leq \; \int_{c_{i-1}}^{c_i}\tilde m\,dL \; \leq \; \varkappa_it_i, \qquad i=1,\ldots,K.
    \end{eqnarray}
    Monotonicity of~$\tilde m$ gives~$t_i\leq\ell_{i+1}$ for~$i=1,\ldots,K-1$, and averaging over the cell gives $\ell_i\leq s_i\leq t_i$, since $\int_{c_{i-1}}^{c_i}\tilde m\,du=a_i=s_ip_i$. Hence, with the conventions~$s_0=t_0:=0$ and~$s_{K+1}=\ell_{K+1}:=1$, which are consistent with~$0\leq\tilde m\leq1$,
    \begin{eqnarray}\label{eq:chain}
        s_{i-1} \; \leq \; t_{i-1} \; \leq \; \ell_i \; \leq \; s_i \; \leq \; t_i \; \leq \; \ell_{i+1} \; \leq \; s_{i+1}, \qquad i=1,\ldots,K.
    \end{eqnarray}
    In particular~$s_1\leq\cdots\leq s_K$, which is therefore necessary; it is also sufficient for the existence of an admissible pair, since~$\tilde m=\sum_is_i1\{u\in(c_{i-1},c_i]\}$ is then non-decreasing and belongs to~$\mathcal M_0$, so that this inequality is the refutable content of assumption~\ref{ass:monotone m}. If~$s_{i-1}=s_i$, then~(\ref{eq:chain}) gives~$\ell_i=s_i$, so that~$\tilde m=s_i$ almost everywhere on the cell and~$t_i=s_i$; symmetrically, if~$s_i=s_{i+1}$, then~$t_i=s_i$, $\tilde m=s_i$ almost everywhere on the cell, and~$\ell_i=s_i$. Combining the two cases with~(\ref{eq:chain}) yields~$\ell_i\geq\underline s_i$ and~$t_i\leq\bar s_i$ for every~$i$, and summing~(\ref{eq:cell contribution}) over cells,
    \begin{eqnarray*}
        LB \; := \; \mu^-(x)+\sum_{i=1}^K\varkappa_i\underline s_i \; \leq \; \mu^-(x)+\int\tilde m\,dL \; \leq \; \mu^-(x)+\sum_{i=1}^K\varkappa_i\bar s_i \; =: \; UB,
    \end{eqnarray*}
    which are the bounds of the proposition, since~$\varkappa_i=\mathbb P(P_x\in B_i,X\neq x)$.

    \emph{Step 2: approach.} Let~$\mathcal N:=\{i:s_{i-1}<s_i<s_{i+1}\}$ be the set of cells on which~$\tilde m$ is not forced to be constant, and let~$\bar\delta>0$ be the smallest positive increment~$s_{i+1}-s_i$, $i=0,\ldots,K$. For~$\delta\in(0,\bar\delta)$, define~$\tilde m^-_\delta$ and~$\tilde m^+_\delta$ on the cell~$(c_{i-1},c_i]$ as follows: if~$i\notin\mathcal N$, both are equal to~$s_i$; if~$i\in\mathcal N$, with~$\alpha_i:=\delta/(s_i-s_{i-1})$ and~$\beta_i:=\delta/(s_{i+1}-s_i)$, both in~$(0,1)$,
    \begin{eqnarray*}
        \tilde m^-_\delta & := & (s_{i-1}+\delta)\,1\{u\leq c_{i-1}+\alpha_ip_i\}+(s_i+\delta)\,1\{u>c_{i-1}+\alpha_ip_i\}, \\
        \tilde m^+_\delta & := & (s_i-\delta)\,1\{u\leq c_i-\beta_ip_i\}+(s_{i+1}-\delta)\,1\{u>c_i-\beta_ip_i\}.
    \end{eqnarray*}
    The cell integrals are~$p_i[s_i+\delta-\alpha_i(s_i-s_{i-1})]=a_i$ and~$p_i[s_i-\delta+\beta_i(s_{i+1}-s_i)]=a_i$ respectively, and the values lie in~$[0,1]$ because~$\delta<\bar\delta\leq\min(s_i-s_{i-1},s_{i+1}-s_i)$ for~$i\in\mathcal N$. Both functions are non-decreasing: within cells by construction, and across the boundary between cells~$i$ and~$i+1$ because the largest value on cell~$i$ never exceeds the smallest value on cell~$i+1$, as one checks in the four cases according to whether~$i$ and~$i+1$ belong to~$\mathcal N$, using~$s_i\leq s_{i+1}$ and, when exactly one of the two cells is in~$\mathcal N$, $s_{i+1}-s_i\geq\bar\delta>\delta$. Hence~$\tilde m^\pm_\delta$ are non-decreasing elements of~$\mathcal M_0$. Next, let~$\ell_\ast:=\sum_i\sigma_i1\{u\in(c_{i-1},c_i]\}$, the derivative of the piecewise linear interpolant~$L_\ast\in\mathcal L_0$ of the points~$(c_i,\sum_{l\leq i}\varkappa_l)$, and
    \begin{eqnarray*}
        \ell^-_\delta & := & \sum_{i\in\mathcal N}\frac{\varkappa_i}{\alpha_ip_i}1\{c_{i-1}<u\leq c_{i-1}+\alpha_ip_i\}+\sum_{i\notin\mathcal N}\sigma_i1\{c_{i-1}<u\leq c_i\}, \\
        \ell^+_\delta & := & \sum_{i\in\mathcal N}\frac{\varkappa_i}{\beta_ip_i}1\{c_i-\beta_ip_i<u\leq c_i\}+\sum_{i\notin\mathcal N}\sigma_i1\{c_{i-1}<u\leq c_i\}.
    \end{eqnarray*}
    For~$\varepsilon\in(0,1)$, let~$L^\pm_{\delta,\varepsilon}(v):=\int_0^v[(1-\varepsilon)\ell^\pm_\delta+\varepsilon\ell_\ast]\,du$. Each~$L^\pm_{\delta,\varepsilon}$ belongs to~$\mathcal L_0$, since its increment over every cell is~$\varkappa_i$, and its derivative is bounded and bounded below by~$\varepsilon\min_i\sigma_i>0$. Since~$\int\tilde m\,dL_\ast=\sum_i\sigma_ia_i=\sum_i\varkappa_is_i$ for every~$\tilde m\in\mathcal M_0$, while~$\int\tilde m^-_\delta\,\ell^-_\delta\,du=\sum_{i\notin\mathcal N}\varkappa_is_i+\sum_{i\in\mathcal N}\varkappa_i(s_{i-1}+\delta)$ and~$\int\tilde m^+_\delta\,\ell^+_\delta\,du=\sum_{i\notin\mathcal N}\varkappa_is_i+\sum_{i\in\mathcal N}\varkappa_i(s_{i+1}-\delta)$,
    \begin{eqnarray*}
        \lim_{\delta\downarrow0}\lim_{\varepsilon\downarrow0}\left\{\mu^-(x)+\int\tilde m^-_\delta\,dL^-_{\delta,\varepsilon}\right\} & = & \mu^-(x)+\sum_{i\notin\mathcal N}\varkappa_is_i+\sum_{i\in\mathcal N}\varkappa_is_{i-1} \; = \; LB,
    \end{eqnarray*}
    where the last equality uses~$\underline s_i=s_i$ for~$i\notin\mathcal N$ and~$\underline s_i=s_{i-1}$ for~$i\in\mathcal N$; symmetrically, the values of~$(L^+_{\delta,\varepsilon},\tilde m^+_\delta)$ converge to~$\mu^-(x)+\sum_{i\notin\mathcal N}\varkappa_is_i+\sum_{i\in\mathcal N}\varkappa_is_{i+1}=UB$. Let~$\mathcal P^\prime$ be the convex hull of the pairs~$(L^-_{\delta,\varepsilon},\tilde m^-_\delta)$ and~$(L^+_{\delta,\varepsilon},\tilde m^+_\delta)$, $\delta\in(0,\bar\delta)$, $\varepsilon\in(0,1)$. Its elements are pairs~$(L,\tilde m)\in\mathcal L_0\times\mathcal M_0$ with~$\tilde m$ non-decreasing and the derivative of~$L$ bounded and bounded away from zero, so that~$\mathcal P^\prime\subseteq\mathcal P_S$, and by step~1 its values lie in~$[LB,UB]$ while approaching both endpoints. Part~(e) of lemma~\ref{lem:quantile} concludes.
\end{proof}

\begin{proof}[Proof of proposition~\ref{prop:coarsened MLR}]
    Here~$S=\{$assumption~\ref{ass:MLR1}$\}$, and the admissible pairs are the~$(L,\tilde m)\in\mathcal L\times\mathcal M_0$. For~$d\in\{0,1\}$ and~$L\in\mathcal L$, let~$\Phi^d(L):=\sum_{i=1}^K\big(L(c_i)-L(c_i-\kappa_i^d)\big)$, where~$\kappa_i^1=a_i$ and~$\kappa_i^0=p_i-a_i$.

    \emph{Step 1: validity.} Let~$L\in\mathcal L$ with non-increasing right derivative~$\ell$, and~$\tilde m\in\mathcal M_0$. Applying~(\ref{eq:bathtub}) on each cell with~$a=a_i$,
    \begin{eqnarray*}
        L(c_i)-L(c_i-a_i) \; \leq \; \int_{c_{i-1}}^{c_i}\tilde m\,dL \; \leq \; L(c_{i-1}+a_i)-L(c_{i-1}).
    \end{eqnarray*}
    Summing over cells, using~$c_{i-1}+a_i=c_i-\kappa_i^0$ and~$\sum_i(L(c_i)-L(c_{i-1}))=L(1)=1-\pi_x$,
    \begin{eqnarray*}
        \Phi^1(L) \; \leq \; \int\tilde m\,dL \; \leq \; 1-\pi_x-\Phi^0(L),
    \end{eqnarray*}
    and therefore, with~$\mu^-(x)+1-\pi_x=\mu^+(x)$,
    \begin{eqnarray*}
        LB \; := \; \mu^-(x)+\inf_{L\in\mathcal L}\Phi^1(L) \; \leq \; \mu^-(x)+\int\tilde m\,dL \; \leq \; \mu^+(x)-\inf_{L\in\mathcal L}\Phi^0(L) \; =: \; UB.
    \end{eqnarray*}

    \emph{Step 2: the infima are attained.} Fix~$d$ and let~$\mathcal T_d:=\{c_0,\ldots,c_K\}\cup\{c_i-\kappa_i^d:i=1,\ldots,K\}$. For~$L\in\mathcal L$, let~$\hat L$ be the piecewise linear interpolant of~$L$ at the knots~$\mathcal T_d$. Its slopes over consecutive knot intervals are chord slopes of the concave function~$L$ over consecutive intervals, hence non-increasing, and they are non-negative since~$L$ is non-decreasing; it is continuous and coincides with~$L$ at~$c_0,\ldots,c_K$. Hence~$\hat L\in\mathcal L$, and~$\Phi^d(\hat L)=\Phi^d(L)$ since~$\hat L=L$ on~$\mathcal T_d$. The infimum of~$\Phi^d$ over~$\mathcal L$ is therefore its infimum over the piecewise linear elements of~$\mathcal L$ with knots in~$\mathcal T_d$. Such an element is determined by the vector~$t=(t_i)_{i=1}^K$ of its values~$t_i=L(c_i-\kappa_i^d)\in[L(c_{i-1}),L(c_i)]$, the values at~$c_0,\ldots,c_K$ being fixed, and it belongs to~$\mathcal L$ if and only if~$t$ satisfies finitely many non-strict linear inequalities (non-negativity and monotonicity of the slopes). The admissible~$t$ form a compact subset of~$\mathbb R^K$, non-empty because~$\mathcal L$ is non-empty under the maintained assumptions, on which~$\Phi^d$ is linear; the infimum is therefore attained, at some~$L^d_\ast\in\mathcal L$. This justifies the minima in the statement, and shows that each is the value of a linear program in the slopes of~$L$ on the segments~$(c_{i-1},c_i-\kappa_i^d]$ and~$(c_i-\kappa_i^d,c_i]$, as claimed after the proposition.

    \emph{Step 3: approach.} Let~$\tilde m^1:=\sum_i1\{c_i-a_i<u\leq c_i\}$ and~$\tilde m^0:=\sum_i1\{c_{i-1}<u\leq c_{i-1}+a_i\}$, both in~$\mathcal M_0$, so that~$\int\tilde m^1\,dL=\Phi^1(L)$ and~$\int\tilde m^0\,dL=1-\pi_x-\Phi^0(L)$ for every~$L\in\mathcal L$. By part~(c) of lemma~\ref{lem:quantile}, there are~$L^d_n\in\mathcal L$, $n\geq1$, with derivatives bounded and bounded away from zero, such that~$\int\tilde m^d\,dL^d_n\rightarrow\int\tilde m^d\,dL^d_\ast$. Let~$\mathcal P^\prime$ be the convex hull of the pairs~$(L^1_n,\tilde m^1)$ and~$(L^0_n,\tilde m^0)$, $n\geq1$. By convexity of~$\mathcal L$ and~$\mathcal M_0$, $\mathcal P^\prime\subseteq\mathcal P_S$, and its elements have first components with derivatives bounded and bounded away from zero. By step~1, the values on~$\mathcal P^\prime$ lie in~$[LB,UB]$; by construction, they approach~$\mu^-(x)+\Phi^1(L^1_\ast)=LB$ and~$\mu^-(x)+1-\pi_x-\Phi^0(L^0_\ast)=UB$. Part~(e) of lemma~\ref{lem:quantile} concludes.

    Finally, $\mathcal L$ is non-empty if and only if~$\varkappa_1/p_1\geq\cdots\geq\varkappa_K/p_K$: the ratios are the chord slopes of any~$L\in\mathcal L$ over consecutive cells, hence non-increasing when~$\mathcal L\neq\emptyset$, and conversely the piecewise linear interpolant~$L_\ast$ belongs to~$\mathcal L$ when they are non-increasing. This inequality is the refutable content of assumption~\ref{ass:MLR1}.
\end{proof}

\begin{proof}[Proof of proposition~\ref{prop:coarsened joint}]
    Here~$S=\{$assumption~\ref{ass:monotone m}, assumption~\ref{ass:MLR1}$\}$, and the admissible pairs are the~$(L,\tilde m)$ with~$L\in\mathcal L$ and~$\tilde m\in\mathcal M_0$ non-decreasing.

    \emph{Step 1: the curve~$R$.} For a non-decreasing~$\tilde m\in\mathcal M_0$, the function~$R(v):=\int_0^v\tilde m(u)\,du$ belongs to~$\mathcal R$: it is convex because~$\tilde m$ is non-decreasing, $R(0)=0$, $0\leq R(v)-R(u)\leq v-u$ because~$0\leq\tilde m\leq1$, $R(c_i)=\sum_{l\leq i}a_l=\sum_{l\leq i}\kappa_l^1$, and~$R$ takes values in~$[0,R(1)]=[0,p(x)]$. Conversely, every~$R\in\mathcal R$ is Lipschitz, hence absolutely continuous, with right derivative~$R^\prime_+$ non-decreasing, with values in~$[0,1]$, and with~$\int_{c_{i-1}}^{c_i}R^\prime_+(u)\,du=R(c_i)-R(c_{i-1})=\kappa_i^1$, so that~$R^\prime_+$ is a non-decreasing element of~$\mathcal M_0$. For the true data generating process, $R$ is the curve defined in section~\ref{sec:coarsening matched}, since~$R(v)=\mathbb P(D=1,F_{P_x\vert X=x}(P_x)\leq v\vert X=x)=\int_0^v\tilde m(u)\,du$ by the change of variables of part~(a) of lemma~\ref{lem:quantile}. Hence the set of values of~$\mu^-(x)+\int\tilde m\,dL$ over admissible pairs is~$\{\mu^-(x)+\int_0^1R^\prime\,dL:(L,R)\in\mathcal L\times\mathcal R\}$, with~$R^\prime$ any version of the derivative of~$R$, and in particular
    \begin{eqnarray*}
        LB \; := \; \mu^-(x)+\inf_{L\in\mathcal L,\,R\in\mathcal R}\int_0^1R^\prime(u)\,dL(u) \; \leq \; \mu^-(x)+\int\tilde m\,dL
    \end{eqnarray*}
    for every admissible pair.

    \emph{Step 2: validity of the upper bound.} Let~$L\in\mathcal L$ with non-increasing right derivative~$\ell$ and~$\tilde m\in\mathcal M_0$ non-decreasing. By~(\ref{eq:chebyshev}) on each cell, $\int_{c_{i-1}}^{c_i}\tilde m\ell\,du\leq p_i^{-1}a_i\varkappa_i=s_i\varkappa_i$. Summing and using~$\mu^-(x)=\pi_xp(x)=\pi_x\sum_is_ip_i$ and~$\pi_xp_i+\varkappa_i=\mathbb P(P_x\in B_i)$,
    \begin{eqnarray*}
        \mu^-(x)+\int\tilde m\,dL \; \leq \; \mu^-(x)+\sum_{i=1}^Ks_i\varkappa_i \; = \; \sum_{i=1}^Ks_i\,\mathbb P(P_x\in B_i) \; =: \; UB.
    \end{eqnarray*}

    \emph{Step 3: approach.} Let~$\tilde m_s:=\sum_is_i1\{c_{i-1}<u\leq c_i\}$, which belongs to~$\mathcal M_0$ and is non-decreasing because~$s_1\leq\cdots\leq s_K$ (step~1 of the proof of proposition~\ref{prop:coarsened monotone}), and let~$L_\ast\in\mathcal L$ be the piecewise linear interpolant of the points~$(c_i,\sum_{l\leq i}\varkappa_l)$, whose derivative~$\sum_i(\varkappa_i/p_i)1\{c_{i-1}<u\leq c_i\}$ is bounded and bounded away from zero. Then~$\mu^-(x)+\int\tilde m_s\,dL_\ast=\mu^-(x)+\sum_is_i\varkappa_i=UB$, so that the upper bound is attained; it is the value of~$\mu(x)$ when the structural function is constant on each cell. For the lower bound, let~$(L,R)\in\mathcal L\times\mathcal R$ and let~$L_n\in\mathcal L$ be given by part~(c) of lemma~\ref{lem:quantile} with~$\tilde m=R^\prime_+$, so that~$\int R^\prime_+\,dL_n\rightarrow\int R^\prime_+\,dL$. Let~$\mathcal P^\prime$ be the convex hull of~$(L_\ast,\tilde m_s)$ and of the pairs~$(L_n,R^\prime_+)$, $(L,R)\in\mathcal L\times\mathcal R$, $n\geq1$. By convexity of~$\mathcal L$ and of the set of non-decreasing elements of~$\mathcal M_0$, $\mathcal P^\prime\subseteq\mathcal P_S$, its elements have first components with derivatives bounded and bounded away from zero, its values lie in~$[LB,UB]$ by steps~1 and~2, its supremum is~$UB$, and its infimum is~$LB$ because~$\inf_n\{\mu^-(x)+\int R^\prime_+\,dL_n\}\leq\mu^-(x)+\int R^\prime\,dL$ for every~$(L,R)\in\mathcal L\times\mathcal R$. Part~(e) of lemma~\ref{lem:quantile} concludes.
\end{proof}

\begin{proof}[Proof of proposition~\ref{prop:censored}]
    With unmatched data, the identified features are the laws of~$(D,X)$ and of~$(\tilde{\mathbf P},X)$, which identify~$p(x)$, $\pi_x$, $p_i$ and~$\varkappa_i$, hence~$c_i$ and the values~(\ref{eq:concentration}) of~$L$, but not the~$s_i$. Here~$S=\{$assumption~\ref{ass:monotone m}, assumption~\ref{ass:MLR1}$\}$, and by the unmatched versions of lemmas~\ref{lem:coarsening} and~\ref{lem:quantile}, the admissible pairs are the~$(L,\tilde m)$ with~$L\in\mathcal L$ and~$\tilde m\in\mathcal M_u$ non-decreasing.

    \emph{Step 1: validity.} Let~$L\in\mathcal L$ with non-increasing right derivative~$\ell$, and~$\tilde m\in\mathcal M_u$ non-decreasing. By~(\ref{eq:chebyshev}) on~$(0,1]$, $\int_0^1\tilde m\ell\,du\leq p(x)(1-\pi_x)$, hence~$\mu^-(x)+\int\tilde m\,dL\leq\pi_xp(x)+(1-\pi_x)p(x)=p(x)=:UB$, which is the upper bound of the proposition; in the notation of section~\ref{sec:coarsening unmatched}, this is Chebyshev's inequality~$\mathbb E[m(x,P_x)h(P_x)\vert X=x]\leq\mathbb E[m(x,P_x)\vert X=x]\,\mathbb E[h(P_x)\vert X=x]$. By~(\ref{eq:bathtub}) on~$(0,1]$ with~$a=p(x)$, $\int\tilde m\,dL\geq L(1)-L(1-p(x))$, hence
    \begin{eqnarray*}
        \mu^-(x)+\int\tilde m\,dL \; \geq \; \mu^-(x)+\inf_{L\in\mathcal L}\big(L(1)-L(1-p(x))\big) \; =: \; LB.
    \end{eqnarray*}
    The infimum is attained: as in step~2 of the proof of proposition~\ref{prop:coarsened MLR}, the piecewise linear interpolant of~$L\in\mathcal L$ at the knots~$\{c_0,\ldots,c_K\}\cup\{1-p(x)\}$ belongs to~$\mathcal L$ and has the same value of~$L(1-p(x))$, and such interpolants are parametrized by the single value~$t=L(1-p(x))$ ranging over a compact interval cut out by finitely many non-strict linear inequalities, on which the objective is linear. Let~$L^c_\ast\in\mathcal L$ be a minimizer.

    \emph{Step 2: approach.} The constant~$\tilde m\equiv p(x)$ is a non-decreasing element of~$\mathcal M_u$ and~$\mu^-(x)+\int p(x)\,dL_\ast=\pi_xp(x)+p(x)(1-\pi_x)=UB$, with~$L_\ast\in\mathcal L$ the piecewise linear interpolant of the points~$(c_i,\sum_{l\leq i}\varkappa_l)$, whose derivative is bounded and bounded away from zero; the upper bound is therefore attained. Let~$\tilde m_p:=1\{1-p(x)<u\leq1\}$, a non-decreasing element of~$\mathcal M_u$ with~$\int\tilde m_p\,dL=L(1)-L(1-p(x))$ for every~$L$, and let~$L_n\in\mathcal L$ be given by part~(c) of lemma~\ref{lem:quantile} for~$L^c_\ast$ and~$\tilde m_p$, so that~$\mu^-(x)+\int\tilde m_p\,dL_n\rightarrow LB$. The convex hull~$\mathcal P^\prime$ of~$(L_\ast,p(x))$ and of the pairs~$(L_n,\tilde m_p)$, $n\geq1$, is contained in~$\mathcal P_S$, its first components have derivatives bounded and bounded away from zero, and its values lie in~$[LB,UB]$ with infimum~$LB$ and supremum~$UB$. Part~(e) of lemma~\ref{lem:quantile}, in its unmatched version, concludes. The lower bound allocates the agents with~$D=1$ to the highest~$p(x)$ ranks of the reports, as stated after the proposition.
\end{proof}

\subsection{Proofs for section~\ref{sec:inference}}

Recall from proposition~\ref{prop:directional} the space~$\mathbb B=\ell^\infty([0,1]^d)\times\ell^\infty([0,1]^d)\times\mathbb R$ with norm~$\|\theta\|_{\mathbb B}=\|\theta_1\|_\infty+\|\theta_2\|_\infty+\vert\theta_3\vert$, the function~$U(\varphi,\theta)=\int_{[0,1]^d}\max\{-\varphi\theta_1(\mathbf p),\theta_2(\mathbf p)\}\,d\mathbf p+\varphi\theta_3$, and the feasible set~$\Phi:=[-M,M]$. For~$\varphi\in\Phi$ and~$\theta\in\mathbb B$, define the three sets
\begin{eqnarray*}
    A_-(\varphi,\theta) \; := \; \{\mathbf p:-\varphi\theta_1(\mathbf p)>\theta_2(\mathbf p)\},
    & \qquad &
    A_+(\varphi,\theta) \; := \; \{\mathbf p:-\varphi\theta_1(\mathbf p)<\theta_2(\mathbf p)\}, \\
    A_0(\varphi,\theta) \; := \; \{\mathbf p:-\varphi\theta_1(\mathbf p)=\theta_2(\mathbf p)\}. & &
\end{eqnarray*}
At~$\theta=\theta_x$, $A_0(\varphi,\theta_x)=\{\mathbf p:f_{\mathbf P}(\mathbf p)=-\varphi f_{\mathbf P\vert X=x}(\mathbf p)\}$ is the contact set of assumption~\ref{ass:nullcontact}, which is Lebesgue-null for every~$\varphi\in\Phi$ under that assumption. Throughout, we use that~$(a,b)\mapsto\max\{a,b\}$ is~$1$-Lipschitz for the supremum norm on~$\mathbb R^2$, i.e., $\vert\max\{a^\prime,b^\prime\}-\max\{a,b\}\vert\leq\max\{\vert a^\prime-a\vert,\vert b^\prime-b\vert\}$.

\begin{lemma}[Hadamard directional differentiability of the dual objective]\label{lem:Uhdd}
For each fixed~$\varphi\in\Phi$, the map~$\theta\mapsto U(\varphi,\theta)$ is convex and Hadamard directionally differentiable at every~$\theta\in\mathbb B$, with derivative, for~$h=(h_1,h_2,h_3)\in\mathbb B$,
\[
U^\prime_\theta(\varphi,\theta;h)=
\int_{A_-(\varphi,\theta)}\!\!(-\varphi\,h_1)\,d\mathbf p
+\int_{A_+(\varphi,\theta)}\!\! h_2\,d\mathbf p
+\int_{A_0(\varphi,\theta)}\!\!\max\{-\varphi\,h_1,\,h_2\}\,d\mathbf p
+\varphi\,h_3 .
\]
\end{lemma}

\begin{proof}
\emph{Convexity.} For fixed~$\varphi$, the integrand~$(\theta_1,\theta_2)\mapsto\max\{-\varphi\theta_1(\mathbf p),\theta_2(\mathbf p)\}$ is a maximum of two affine functions of~$(\theta_1(\mathbf p),\theta_2(\mathbf p))$, hence convex; integration and the linear term~$\varphi\theta_3$ preserve convexity.

\emph{Hadamard directional differentiability.} Fix $\theta$ and $h$, and take arbitrary sequences $t_n\downarrow0$ and $h_n=(h_{1n},h_{2n},h_{3n})\rightarrow h$ in $\mathbb B$. The difference quotient is
\begin{eqnarray*}
\frac{U(\varphi,\theta+t_nh_n)-U(\varphi,\theta)}{t_n}
& = & \int_{[0,1]^d}R_n(\mathbf p)\,d\mathbf p+\varphi\,h_{3n}, \\
t_nR_n(\mathbf p) & := & \max\{-\varphi\theta_1(\mathbf p)-\varphi t_n h_{1n}(\mathbf p),\,\theta_2(\mathbf p)+t_n h_{2n}(\mathbf p)\} \\
& & \hskip10pt -\max\{-\varphi\theta_1(\mathbf p),\,\theta_2(\mathbf p)\}.
\end{eqnarray*}
By the Lipschitz property of the maximum,
\[
\vert R_n(\mathbf p)\vert\leq\max\{\vert\varphi\vert\,\vert h_{1n}(\mathbf p)\vert,\,\vert h_{2n}(\mathbf p)\vert\}
\leq \max\{M,1\}\,\sup_n\big(\|h_{1n}\|_\infty\vee\|h_{2n}\|_\infty\big)=:C<\infty,
\]
where~$C$ is finite because~$h_n\rightarrow h$ in~$\mathbb B$; the constant~$C$ is integrable on~$[0,1]^d$. For the pointwise limit, fix~$\mathbf p$ and write~$a=-\varphi\theta_1(\mathbf p)$, $b=\theta_2(\mathbf p)$, $\alpha_n=-\varphi h_{1n}(\mathbf p)\rightarrow\alpha:=-\varphi h_1(\mathbf p)$ and~$\beta_n=h_{2n}(\mathbf p)\rightarrow\beta:=h_2(\mathbf p)$. If~$a>b$, then for~$n$ large~$a+t_n\alpha_n>b+t_n\beta_n$ and~$R_n(\mathbf p)=\alpha_n\rightarrow\alpha$; if~$a<b$, symmetrically~$R_n(\mathbf p)=\beta_n\rightarrow\beta$; if~$a=b$, then~$\max\{a+t_n\alpha_n,a+t_n\beta_n\}=a+t_n\max\{\alpha_n,\beta_n\}$, so that~$R_n(\mathbf p)=\max\{\alpha_n,\beta_n\}\rightarrow\max\{\alpha,\beta\}$. Hence~$R_n\rightarrow\psi_\varphi(\cdot;\theta,h)$ pointwise, where~$\psi_\varphi(\cdot;\theta,h)$ equals~$-\varphi h_1$ on~$A_-(\varphi,\theta)$, $h_2$ on~$A_+(\varphi,\theta)$, and~$\max\{-\varphi h_1,h_2\}$ on~$A_0(\varphi,\theta)$. By dominated convergence, $\int R_n\,d\mathbf p\rightarrow\int\psi_\varphi(\cdot;\theta,h)\,d\mathbf p$, and~$\varphi h_{3n}\rightarrow\varphi h_3$. Since~$t_n\downarrow0$ and~$h_n\rightarrow h$ were arbitrary, $\theta\mapsto U(\varphi,\theta)$ is Hadamard directionally differentiable at~$\theta$ with the stated derivative.
\end{proof}

\begin{proof}[Proof of proposition~\ref{prop:directional}]
We verify the hypotheses of Proposition~4.12, page~272 of \cite{bonnans2013perturbation}, applied to~$f(\varphi,\theta):=U(\varphi,\theta)$ with feasible set~$\Phi=[-M,M]$ (independent of~$\theta$) and parameter~$\theta\in\mathbb B$, together with the Hadamard refinement stated in the remark following its proof (\citealp{bonnans2013perturbation}, p.~273). In the notation of that proposition, $x\leftrightarrow\varphi$, $u\leftrightarrow\theta$, $u_0\leftrightarrow\theta_x$, and~$v(\theta)\leftrightarrow\bar\phi(\theta)=\min_{\varphi\in\Phi}U(\varphi,\theta)$.

\emph{(i) Continuity of~$U$ on~$\Phi\times\mathbb B$.} For~$\varphi,\varphi^\prime\in\Phi$ and~$\theta,\theta^\prime\in\mathbb B$, the Lipschitz property of the maximum and~$\vert[0,1]^d\vert=1$ give
\[
\vert U(\varphi,\theta)-U(\varphi^\prime,\theta^\prime)\vert
\;\leq\; M\,\|\theta_1-\theta_1^\prime\|_\infty+\|\theta_2-\theta_2^\prime\|_\infty+M\,\vert\theta_3-\theta_3^\prime\vert
+\vert\varphi-\varphi^\prime\vert\,\big(\|\theta_1^\prime\|_\infty+\vert\theta_3^\prime\vert\big),
\]
so that~$U$ is jointly continuous on~$\Phi\times\mathbb B$, which is condition~(i).

\emph{(ii) Inf-compactness.} The feasible set~$\Phi$ is compact and does not depend on~$\theta$. Let~$a:=U(0,\theta_x)+1$; by~(i), $U(0,\theta)\leq a$ for all~$\theta$ in a neighborhood of~$\theta_x$, so that for such~$\theta$ the level set~$\{\varphi\in\Phi:U(\varphi,\theta)\leq a\}$ is non-empty and contained in the fixed compact set~$\Phi$, which is the inf-compactness condition. By~(i), the minimum defining~$\bar\phi(\theta)$ is attained and the solution set~$S(\theta):=\arg\min_{\varphi\in\Phi}U(\varphi,\theta)$ is non-empty and compact. For fixed~$\theta$, the integrand~$\varphi\mapsto\max\{-\varphi\theta_1(\mathbf p),\theta_2(\mathbf p)\}$ is a maximum of affine functions of~$\varphi$, so that~$\varphi\mapsto U(\varphi,\theta)$ is convex and~$S(\theta)$ is a compact subinterval of~$[-M,M]$.

\emph{(iii) Directional differentiability of~$U(\varphi,\cdot)$ at~$\theta_x$.} For each fixed~$\varphi\in\Phi$, lemma~\ref{lem:Uhdd} establishes that~$\theta\mapsto U(\varphi,\theta)$ is Hadamard directionally differentiable at~$\theta_x$, with derivative~$U^\prime_\theta(\varphi,\theta_x;\cdot)$ given there. In particular it is directionally differentiable, which is condition~(iii), and Hadamard directionally differentiable, which is the first half of the Hadamard refinement.

\emph{(iv) The liminf condition (4.27), in Hadamard form.} Let~$h\in\mathbb B$, $t_n\downarrow0$ and~$h_n\rightarrow h$ in~$\mathbb B$, and let~$\varphi_n\in\Phi$ be arbitrary. We must exhibit a limit point~$\bar\varphi$ of~$\{\varphi_n\}$ with
\begin{equation}\label{eq:427}
\liminf_{n\rightarrow\infty}\frac{U(\varphi_n,\theta_x+t_nh_n)-U(\varphi_n,\theta_x)}{t_n}
\;\geq\;U^\prime_\theta(\bar\varphi,\theta_x;h).
\end{equation}
By compactness of~$\Phi$, pass to a subsequence (not relabeled) along which the~$\liminf$ is attained as a limit and~$\varphi_n\rightarrow\bar\varphi$ for some~$\bar\varphi\in\Phi$. Write
\begin{eqnarray*}
\frac{U(\varphi_n,\theta_x+t_nh_n)-U(\varphi_n,\theta_x)}{t_n} & = & \int_{[0,1]^d}R_n(\mathbf p)\,d\mathbf p+\varphi_n h_{3n}, \\
t_nR_n(\mathbf p) & = & \max\{a_n(\mathbf p)-\varphi_n t_n h_{1n}(\mathbf p),\,b(\mathbf p)+t_n h_{2n}(\mathbf p)\} \\
& & \hskip10pt -\max\{a_n(\mathbf p),b(\mathbf p)\},
\end{eqnarray*}
where~$a_n(\mathbf p)=-\varphi_n\theta_{x,1}(\mathbf p)$ and~$b(\mathbf p)=\theta_{x,2}(\mathbf p)$. The Lipschitz property of the maximum gives the integrable bound~$\vert R_n(\mathbf p)\vert\leq\max\{M,1\}\sup_n(\|h_{1n}\|_\infty\vee\|h_{2n}\|_\infty)<\infty$. Fix~$\mathbf p\notin A_0(\bar\varphi,\theta_x)$. If~$-\bar\varphi\theta_{x,1}(\mathbf p)>\theta_{x,2}(\mathbf p)$ then~$a_n(\mathbf p)\rightarrow-\bar\varphi\theta_{x,1}(\mathbf p)>b(\mathbf p)$, so for~$n$ large both maxima are attained at their first arguments and~$R_n(\mathbf p)=-\varphi_n h_{1n}(\mathbf p)\rightarrow-\bar\varphi h_1(\mathbf p)$; symmetrically, if~$-\bar\varphi\theta_{x,1}(\mathbf p)<\theta_{x,2}(\mathbf p)$ then~$R_n(\mathbf p)=h_{2n}(\mathbf p)\rightarrow h_2(\mathbf p)$. Hence~$R_n(\mathbf p)\rightarrow\psi_{\bar\varphi}(\mathbf p;\theta_x,h)$ for every~$\mathbf p\notin A_0(\bar\varphi,\theta_x)$, with~$\psi$ as in the proof of lemma~\ref{lem:Uhdd}. Under assumption~\ref{ass:nullcontact}, $A_0(\bar\varphi,\theta_x)$ is Lebesgue-null, so~$R_n\rightarrow\psi_{\bar\varphi}(\cdot;\theta_x,h)$ almost everywhere. By dominated convergence,
\[
\int_{[0,1]^d}R_n\,d\mathbf p+\varphi_n h_{3n}\;\longrightarrow\;
\int_{[0,1]^d}\psi_{\bar\varphi}(\cdot;\theta_x,h)\,d\mathbf p+\bar\varphi h_3=U^\prime_\theta(\bar\varphi,\theta_x;h),
\]
where the last equality uses~$\vert A_0(\bar\varphi,\theta_x)\vert=0$ to drop the contact set from the formula of lemma~\ref{lem:Uhdd}. This is exactly~\eqref{eq:427}, with equality, so condition~(iv) holds in the sequential ($d_n\rightarrow d$) form required by the Hadamard refinement.

\emph{Conclusion.} Conditions~(i)--(iv) of Proposition~4.12 hold, and by~(iii) each~$U(\varphi,\cdot)$ is Hadamard directionally differentiable while~(iv) was verified along~$h_n\rightarrow h$. Hence the Hadamard refinement of \cite{bonnans2013perturbation} (p.~273) applies: $\bar\phi$ is Hadamard directionally differentiable at~$\theta_x$ and, by~(4.28),
\[
\bar\phi^\prime(\theta_x;h)=\min_{\varphi\in S(\theta_x)}U^\prime_\theta(\varphi,\theta_x;h),
\]
with~$U^\prime_\theta(\varphi,\theta_x;h)$ as in lemma~\ref{lem:Uhdd}; the minimum is attained because~$S(\theta_x)$ is compact and, under assumption~\ref{ass:nullcontact}, $\varphi\mapsto U^\prime_\theta(\varphi,\theta_x;h)$ is continuous: the contact set~$A_0(\varphi,\theta_x)$ being null, $\mathbf 1_{A_-(\varphi^\prime,\theta_x)}\rightarrow\mathbf 1_{A_-(\varphi,\theta_x)}$ and~$\mathbf 1_{A_+(\varphi^\prime,\theta_x)}\rightarrow\mathbf 1_{A_+(\varphi,\theta_x)}$ almost everywhere as~$\varphi^\prime\rightarrow\varphi$, whence dominated convergence. The lower bound~$\underline\phi$ is handled in the same way, with~$\min\{\varphi\theta_1,\theta_2\}$ in place of~$\max\{-\varphi\theta_1,\theta_2\}$ and a maximum over~$\Phi$ in place of a minimum.
\end{proof}

\end{appendix}
\newpage
% BIBLIOGRAPHY
\bibliographystyle{apalike}
\bibliography{Elicited}

\begin{thebibliography}{}

\bibitem[Adams-Prassl and Andrew, 2019]{adams2019}
Adams-Prassl, A. and Andrew, A. (2019).
\newblock Preferences and beliefs in the marriage market for young brides.
\newblock CEPR Discussion Paper No. DP13567.

\bibitem[Alm{\aa}s et~al., 2024]{almaas2024}
Alm{\aa}s, I., Attanasio, O., and Jervis, P. (2024).
\newblock Presidential address: Economics and measurement: New measures to
  model decision making.
\newblock {\em Econometrica}, 92(4):947--978.

\bibitem[Ameriks et~al., 2020a]{ameriks2020b}
Ameriks, J., Briggs, J., Caplin, A., Lee, M., Shapiro, M.~D., and Tonetti, C.
  (2020a).
\newblock Older {Americans} would work longer if jobs were flexible.
\newblock {\em American Economic Journal: Macroeconomics}, 12(1):174--209.

\bibitem[Ameriks et~al., 2020b]{ameriks2020a}
Ameriks, J., Briggs, J., Caplin, A., Shapiro, M.~D., and Tonetti, C. (2020b).
\newblock Long-term-care utility and late-in-life saving.
\newblock {\em Journal of Political Economy}, 128(6):2375--2451.

\bibitem[Ameriks et~al., 2011]{ameriks2011joy}
Ameriks, J., Caplin, A., Laufer, S., and Van~Nieuwerburgh, S. (2011).
\newblock The joy of giving or assisted living? {U}sing strategic surveys to
  separate public care aversion from bequest motives.
\newblock {\em The Journal of Finance}, 66(2):519--561.

\bibitem[Arcidiacono et~al., 2020]{arcidiacono2020}
Arcidiacono, P., Hotz, V.~J., Maurel, A., and Romano, T. (2020).
\newblock Ex ante returns and occupational choice.
\newblock {\em Journal of Political Economy}, 128(12):4475--4522.

\bibitem[Athey et~al., 2025]{athey2025combining}
Athey, S., Chetty, R., and Imbens, G. (2025).
\newblock Combining experimental and observational data to estimate treatment
  effects on long term outcomes.
\newblock {\em arXiv preprint arXiv:2006.09676}.

\bibitem[Attanasio et~al., 2019]{attanasio2019}
Attanasio, O., Cunha, F., and Jervis, P. (2019).
\newblock Subjective parental beliefs. {T}heir measurement and role.
\newblock Technical report, National Bureau of Economic Research.

\bibitem[Batista et~al., 2025]{batista2025}
Batista, C., Costa, D.~M., Freitas, P., Lima, G., and Reis, A.~B. (2025).
\newblock What matters for the decision to study abroad? {A} lab-in-the-field
  experiment in {Cape Verde}.
\newblock {\em Journal of Development Economics}, 173:103401.

\bibitem[Bernheim et~al., 2021]{bernheim2022}
Bernheim, B.~D., Bj{\"o}rkegren, D., Naecker, J., and Pollmann, M. (2021).
\newblock Causal inference from hypothetical evaluations.
\newblock Technical report, National Bureau of Economic Research.

\bibitem[Blass et~al., 2010]{blass2010}
Blass, A.~A., Lach, S., and Manski, C.~F. (2010).
\newblock Using elicited choice probabilities to estimate random utility
  models: Preferences for electricity reliability.
\newblock {\em International Economic Review}, 51(2):421--440.

\bibitem[Bonnans and Shapiro, 2013]{bonnans2013perturbation}
Bonnans, J.~F. and Shapiro, A. (2013).
\newblock {\em Perturbation analysis of optimization problems}.
\newblock Springer Science \& Business Media.

\bibitem[Bontemps et~al., 2025]{bontemps2025functional}
Bontemps, C., Florens, J.-P., and Meddahi, N. (2025).
\newblock Functional ecological inference.
\newblock {\em Journal of Econometrics}, 248:105918.

\bibitem[Boyer et~al., 2020]{boyer2020}
Boyer, M.~M., De~Donder, P., Fluet, C., Leroux, M.-L., and Michaud, P.-C.
  (2020).
\newblock Long-term care insurance: Information frictions and selection.
\newblock {\em American Economic Journal: Economic Policy}, 12(3):134--69.

\bibitem[Briggs et~al., 2020]{briggs2020}
Briggs, J., Caplin, A., Leth-Petersen, S., Tonetti, C., and Violante, G.
  (2020).
\newblock Estimating marginal treatment effects with survey instruments.
\newblock Technical report, Working Paper.

\bibitem[Cambanis et~al., 1976]{cambanis1976inequalities}
Cambanis, S., Simons, G., and Stout, W. (1976).
\newblock Inequalities for {E}k(x, y) when the marginals are fixed.
\newblock {\em Zeitschrift f{\"u}r Wahrscheinlichkeitstheorie und verwandte
  Gebiete}, 36(4):285--294.

\bibitem[Chernozhukov et~al., 2018]{chernozhukov2018double}
Chernozhukov, V., Chetverikov, D., Demirer, M., Duflo, E., Hansen, C., Newey,
  W., and Robins, J. (2018).
\newblock Double/debiased machine learning for treatment and structural
  parameters.
\newblock {\em The Econometrics Journal}, 21(1):C1--C68.

\bibitem[Cross and Manski, 2002]{cross2002regressions}
Cross, P.~J. and Manski, C.~F. (2002).
\newblock Regressions, short and long.
\newblock {\em Econometrica}, 70(1):357--368.

\bibitem[de~Bresser and van Soest, 2019]{debresser2019}
de~Bresser, J. and van Soest, A. (2019).
\newblock The predictive power of subjective probabilities: probabilistic and
  deterministic polling in the {Dutch} 2017 election.
\newblock {\em Journal of the Royal Statistical Society: Series A (Statistics
  in Society)}, 182(2):443--466.

\bibitem[Delavande and Manski, 2015]{delavande2015}
Delavande, A. and Manski, C.~F. (2015).
\newblock Using elicited choice probabilities in hypothetical elections to
  study decisions to vote.
\newblock {\em Electoral Studies}, 38:28--37.

\bibitem[Delavande and Zafar, 2019]{delavande2019}
Delavande, A. and Zafar, B. (2019).
\newblock University choice: The role of expected earnings, nonpecuniary
  outcomes, and financial constraints.
\newblock {\em Journal of Political Economy}, 127(5):2343--2393.

\bibitem[D'Haultfoeuille et~al., 2024]{gaillac2024linear}
D'Haultfoeuille, X., Gaillac, C., and Maurel, A. (2024).
\newblock Linear regressions with combined data.
\newblock {\em arXiv preprint arXiv:2412.04816}.

\bibitem[D'Haultfoeuille et~al., 2025]{gaillac2025partially}
D'Haultfoeuille, X., Gaillac, C., and Maurel, A. (2025).
\newblock Partially linear models under data combination.
\newblock {\em Review of Economic Studies}, 92(1):238--267.

\bibitem[Fan et~al., 2025]{fan2025partial}
Fan, Y., Pass, B., and Shi, X. (2025).
\newblock Partial identification in moment models with incomplete data via
  optimal transport.
\newblock {\em arXiv preprint arXiv:2503.16098}.

\bibitem[Fan et~al., 2014]{fan2014identifying}
Fan, Y., Sherman, R., and Shum, M. (2014).
\newblock Identifying treatment effects under data combination.
\newblock {\em Econometrica}, 82(2):811--822.

\bibitem[Fan et~al., 2016]{fan2016estimation}
Fan, Y., Sherman, R., and Shum, M. (2016).
\newblock Estimation and inference in an ecological inference model.
\newblock {\em Journal of Econometric Methods}, 5(1):17--48.

\bibitem[Fang and Santos, 2019]{fang2019inference}
Fang, Z. and Santos, A. (2019).
\newblock Inference on directionally differentiable functions.
\newblock {\em The Review of Economic Studies}, 86(1):377--412.

\bibitem[Galichon and Henry, 2026]{galichon2026econometrician}
Galichon, A. and Henry, M. (2026).
\newblock An econometrician's guide to optimal transport.
\newblock {\em arXiv:2604.04227}.

\bibitem[Giustinelli, 2023]{giustinelli2023}
Giustinelli, P. (2023).
\newblock Expectations in education.
\newblock {\em Handbook of Economic Expectations}, pages 193--224.

\bibitem[Giustinelli and Shapiro, 2024]{giustinelli2024}
Giustinelli, P. and Shapiro, M.~D. (2024).
\newblock {SeaTE}: Subjective ex ante treatment effect of health on retirement.
\newblock {\em American Economic Journal: Applied Economics}, 16(2):278--317.

\bibitem[Gong et~al., 2022]{gong2022}
Gong, Y., Stinebrickner, T., Stinebrickner, R., and Yao, Y. (2022).
\newblock The role of non-pecuniary considerations: Location decisions of
  college graduates from low income backgrounds.
\newblock Technical report, University of Western Ontario, Centre for Human
  Capital and Productivity (CHCP).

\bibitem[Haghani et~al., 2021a]{haghani2021a}
Haghani, M., Bliemer, M.~C., Rose, J.~M., Oppewal, H., and Lancsar, E. (2021a).
\newblock Hypothetical bias in stated choice experiments: {Part I}.
  {M}acro-scale analysis of literature and integrative synthesis of empirical
  evidence from applied economics, experimental psychology and neuroimaging.
\newblock {\em Journal of Choice Modelling}, 41:100309.

\bibitem[Haghani et~al., 2021b]{haghani2021b}
Haghani, M., Bliemer, M.~C., Rose, J.~M., Oppewal, H., and Lancsar, E. (2021b).
\newblock Hypothetical bias in stated choice experiments: {Part II}.
  {C}onceptualisation of external validity, sources and explanations of bias
  and effectiveness of mitigation methods.
\newblock {\em Journal of Choice Modelling}, 41:100322.

\bibitem[Hainmueller et~al., 2015]{hainmueller2015}
Hainmueller, J., Hangartner, D., and Yamamoto, T. (2015).
\newblock Validating vignette and conjoint survey experiments against
  real-world behavior.
\newblock {\em Proceedings of the National Academy of Sciences},
  112(8):2395--2400.

\bibitem[Hausman, 2012]{hausman2012}
Hausman, J. (2012).
\newblock Contingent valuation: {F}rom dubious to hopeless.
\newblock {\em Journal of Economic Perspectives}, 26(4):43--56.

\bibitem[Heckman and Robb, 1985]{heckman1985alternative}
Heckman, J.~J. and Robb, R. (1985).
\newblock Alternative methods for evaluating the impact of interventions: An
  overview.
\newblock {\em Journal of Econometrics}, 30(1-2):239--267.

\bibitem[Heckman et~al., 1997]{heckman1997}
Heckman, J.~J., Smith, J., and Clements, N. (1997).
\newblock Making the most out of programme evaluations and social experiments:
  Accounting for heterogeneity in programme impacts.
\newblock {\em The Review of Economic Studies}, 64(4):487--535.

\bibitem[Hong and Li, 2018]{hong2018numerical}
Hong, H. and Li, J. (2018).
\newblock The numerical delta method.
\newblock {\em Journal of Econometrics}, 206(2):379--394.

\bibitem[Horowitz and Manski, 1995]{horowitz1995identification}
Horowitz, J.~L. and Manski, C.~F. (1995).
\newblock Identification and robustness with contaminated and corrupted data.
\newblock {\em Econometrica}, pages 281--302.

\bibitem[Hurd, 2009]{hurd2009}
Hurd, M.~D. (2009).
\newblock Subjective probabilities in household surveys.
\newblock {\em Annual Review of Economics}, 1(1):543--562.

\bibitem[Ichimura and Martinez-Sanchis, 2005]{ichimura2005identification}
Ichimura, H. and Martinez-Sanchis, E. (2005).
\newblock Identification and estimation of {GMM} models by combining two data
  sets.
\newblock Technical report, Working paper, UCL and CEMMAP.

\bibitem[Juster and Shay, 1964]{juster1964}
Juster, F.~T. and Shay, R.~P. (1964).
\newblock Consumer sensitivity to finance rates.
\newblock In {\em Consumer sensitivity to finance rates: An empirical and
  analytical investigation}, pages 6--46. NBER.

\bibitem[Kesternich et~al., 2013]{kesternich2013}
Kesternich, I., Heiss, F., McFadden, D., and Winter, J. (2013).
\newblock Suit the action to the word, the word to the action: Hypothetical
  choices and real decisions in {Medicare Part D}.
\newblock {\em Journal of Health Economics}, 32(6):1313--1324.

\bibitem[Ko{\c{s}}ar and O'Dea, 2023]{kocsar2023}
Ko{\c{s}}ar, G. and O'Dea, C. (2023).
\newblock Expectations data in structural microeconomic models.
\newblock In {\em Handbook of Economic Expectations}, pages 647--675. Elsevier.

\bibitem[Ko{\c{s}}ar et~al., 2022]{kocsar2022}
Ko{\c{s}}ar, G., Ransom, T., and Van~der Klaauw, W. (2022).
\newblock Understanding migration aversion using elicited counterfactual choice
  probabilities.
\newblock {\em Journal of Econometrics}, 231(1):123--147.

\bibitem[Lee, 2019]{lee2019identification}
Lee, R. (2019).
\newblock Identification and estimation of treatment effects with instrumental
  variables under data combination.
\newblock Technical report, Northwestern University.

\bibitem[Low, 2024]{low2024}
Low, C. (2024).
\newblock Pricing the biological clock: The marriage market costs of aging to
  women.
\newblock {\em Journal of Labor Economics}, 42(2):395--426.

\bibitem[Maestas et~al., 2023]{maestas2023}
Maestas, N., Mullen, K.~J., Powell, D., Von~Wachter, T., and Wenger, J.~B.
  (2023).
\newblock The value of working conditions in the {United States} and
  implications for the structure of wages.
\newblock {\em American Economic Review}, 113(7):2007--2047.

\bibitem[Manski, 1990]{manski1990use}
Manski, C.~F. (1990).
\newblock The use of intentions data to predict behavior: A best-case analysis.
\newblock {\em Journal of the American Statistical Association},
  85(412):934--940.

\bibitem[Manski, 1999]{manski1999}
Manski, C.~F. (1999).
\newblock Analysis of choice expectations in incomplete scenarios.
\newblock {\em Journal of Risk and Uncertainty}, 19:49--66.

\bibitem[Manski, 2004]{manski2004}
Manski, C.~F. (2004).
\newblock Measuring expectations.
\newblock {\em Econometrica}, 72(5):1329--1376.

\bibitem[Mas and Pallais, 2017]{mas2017}
Mas, A. and Pallais, A. (2017).
\newblock Valuing alternative work arrangements.
\newblock {\em American Economic Review}, 107(12):3722--3759.

\bibitem[Mas-Colell, 1979]{mas1979}
Mas-Colell, A. (1979).
\newblock Homeomorphisms of compact, convex sets and the {J}acobian matrix.
\newblock {\em SIAM Journal on Mathematical Analysis}, 10(6):1105--1109.

\bibitem[McFadden, 2017]{mcfadden2017}
McFadden, D. (2017).
\newblock Stated preference methods and their applicability to environmental
  use and non-use valuations.
\newblock In {\em Contingent valuation of environmental goods}, pages 153--187.
  Edward Elgar Publishing.

\bibitem[M\'eango and Girsberger, 2024]{meango2024}
M\'eango, R. and Girsberger, E.~M. (2024).
\newblock Identification of ex ante returns using elicited choice
  probabilities: an application to preferences for public-sector jobs.
\newblock {\em arXiv preprint arXiv:2307.13966v2}.

\bibitem[M\'{e}ango and Poinas, 2023]{meango2022}
M\'{e}ango, R. and Poinas, F. (2023).
\newblock The option value of overstaying.
\newblock CESifo Working Paper 10536.

\bibitem[Mityagin, 2015]{mityagin2015zero}
Mityagin, B. (2015).
\newblock The zero set of a real analytic function.
\newblock {\em arXiv preprint arXiv:1512.07276}.

\bibitem[Molinari and Peski, 2006]{molinari2006generalization}
Molinari, F. and Peski, M. (2006).
\newblock Generalization of a result on ``regressions, short and long''.
\newblock {\em Econometric Theory}, 22(1):159--163.

\bibitem[Morikawa et~al., 2002]{morikawa2002}
Morikawa, T., Ben-Akiva, M., and McFadden, D. (2002).
\newblock Discrete choice models incorporating revealed preferences and
  psychometric data.
\newblock In {\em Advances in Econometrics}, pages 29--55. Emerald Group
  Publishing Limited.

\bibitem[Murphy et~al., 2005]{murphy2005}
Murphy, J.~J., Allen, P.~G., Stevens, T.~H., and Weatherhead, D. (2005).
\newblock A meta-analysis of hypothetical bias in stated preference valuation.
\newblock {\em Environmental and Resource Economics}, 30:313--325.

\bibitem[Pacini, 2019]{pacini2019two}
Pacini, D. (2019).
\newblock Two-sample least squares projection.
\newblock {\em Econometric Reviews}, 38(1):95--123.

\bibitem[Pantano and Zheng, 2013]{pantano2013}
Pantano, J. and Zheng, Y. (2013).
\newblock Using subjective expectations data to allow for unobserved
  heterogeneity in {Hotz--Miller} estimation strategies.
\newblock {\em Available at SSRN 2129303}.

\bibitem[R{\"u}schendorf, 1991]{ruschendorf1991bounds}
R{\"u}schendorf, L. (1991).
\newblock Bounds for distributions with multivariate marginals.
\newblock {\em Lecture Notes-Monograph Series}, pages 285--310.

\bibitem[Train, 2009]{train2009}
Train, K.~E. (2009).
\newblock {\em Discrete choice methods with simulation}.
\newblock Cambridge University Press.

\bibitem[Van~der Klaauw, 2012]{vanderklaauw2012}
Van~der Klaauw, W. (2012).
\newblock On the use of expectations data in estimating structural dynamic
  choice models.
\newblock {\em Journal of Labor Economics}, 30(3):521--554.

\bibitem[Wiswall and Zafar, 2015]{wiswall2015}
Wiswall, M. and Zafar, B. (2015).
\newblock Determinants of college major choice: Identification using an
  information experiment.
\newblock {\em The Review of Economic Studies}, 82(2):791--824.

\bibitem[Wiswall and Zafar, 2018]{wiswall2018}
Wiswall, M. and Zafar, B. (2018).
\newblock Preference for the workplace, investment in human capital, and
  gender.
\newblock {\em The Quarterly Journal of Economics}, 133(1):457--507.

\bibitem[Wiswall and Zafar, 2021]{wiswall2021}
Wiswall, M. and Zafar, B. (2021).
\newblock Human capital investments and expectations about career and family.
\newblock {\em Journal of Political Economy}, 129(5):1361--1424.

\end{thebibliography}

\end{document}